\documentclass{article}
\pdfoutput=1
\usepackage[T1]{fontenc}
\usepackage[utf8]{inputenc}
\usepackage{microtype}
\usepackage[total={6in, 8in}]{geometry}
\usepackage{comment}
\usepackage{amsmath}
\usepackage{amssymb}
\usepackage{amsthm}
\usepackage{graphicx}
\usepackage{mathrsfs}
\usepackage{braket}
\usepackage{bbm}
\usepackage{xcolor}
\usepackage{cite}
\usepackage{doi}
\usepackage{tikz}
\usepackage{booktabs}
\usepackage{threeparttable}
\usepackage{multirow}

\usetikzlibrary{quantikz2}
\usepackage{theoremref}

\usepackage{algorithm}
\usepackage{algorithmic}

\newtheorem{thm}{\protect\theoremname}
\theoremstyle{plain}
\newtheorem{lem}{\protect\lemmaname}
\theoremstyle{plain}
\newtheorem{rem}{\protect\remarkname}
\theoremstyle{plain}
\newtheorem*{lem*}{\protect\lemmaname}
\theoremstyle{plain}
\newtheorem*{thm*}{\protect\theoremname}
\theoremstyle{plain}
\newtheorem{prop}{\protect\propositionname}
\theoremstyle{plain}
\newtheorem{cor}{\protect\corollaryname}
\theoremstyle{plain}
\newtheorem*{cor*}{\protect\corollaryname}

\newtheorem{defn}{Definition}

\usepackage[USenglish]{babel}
  \providecommand{\corollaryname}{Corollary}
  \providecommand{\lemmaname}{Lemma}
  \providecommand{\propositionname}{Proposition}
  \providecommand{\remarkname}{Remark}
\providecommand{\theoremname}{Theorem}

\usepackage[normalem]{ulem}
\usepackage{authblk}

\hypersetup{colorlinks=true, linkcolor=blue, urlcolor=blue, citecolor = blue}

\newcommand{\RR}{\mathbb{R}}
\newcommand{\CC}{\mathbb{C}}

\newcommand{\ZZ}{\mathbb{Z}}
\newcommand{\dd}{\mathrm{d}}
\newcommand{\Tr}{\mathrm{Tr}}

\newcommand{\YZ}[1]{\textcolor{green}{[YZ: #1]}}

\title{Fast mixing of weakly interacting fermionic systems \\ at any temperature}
 \author[1,2,3]{Yu Tong}
    \affil[1]{Department of Mathematics, Duke University, Durham, NC 27708, USA}
    \affil[2]{Department of Electrical and Computer Engineering, Duke University, Durham, NC 27708, USA}
    \affil[3]{Duke Quantum Center, Duke University, Durham, NC 27701, USA}
\author[4,5]{Yongtao Zhan}
    \affil[4]{Institute for Quantum Information and Matter, California Institute of Technology, CA 91125, USA}
    \affil[5]{Department of Physics, California Institute of Technology, CA 91125, USA}
\date{\today}

\begin{document}

\maketitle

\begin{abstract}
   We study the mixing time of a recently proposed efficiently implementable Lindbladian designed to prepare the Gibbs states in the setting of weakly interacting fermionic systems. We show that at any temperature, the Lindbladian spectral gap for even parity observables is lower bounded by a constant that is independent of the system size, when the interaction strength (e.g., the on-site interaction strength for the Fermi-Hubbard model) is below a constant threshold, which is also independent of the system size. This leads to a mixing time estimate that is at most linear in the system size, thus showing that the corresponding Gibbs states can be prepared efficiently on quantum computers.
\end{abstract}

\tableofcontents
\clearpage

\section{Introduction}


Simulating non-integrable quantum systems is a notoriously difficult task. 
Numerical methods designed for this purpose generally suffer from the curse of dimensionality, the sign problem, or are not efficiently scalable either in the system size or achievable precision.
Non-integrable quantum systems that are efficiently simulable either on a classical or quantum computer are hard to find. Some examples include: ground states of 1D gapped local Hamiltonians \cite{Landau2015polynomial}, Gibbs states that have exponentially decaying correlation and are approximately Markov states \cite{BrandaoKastoryano2019finite}, Gibbs states of certain commuting Hamiltonians \cite{HwangJiang2024gibbs,KastoryanoBrandao2016quantum,KochanowskiAlhambraCapelRouze2024rapid,BardetCapelLiEtAl2023rapid}, ground states of quantum impurity models \cite{BravyiGosset2017complexity}, and time evolution of 1D many-body localized systems \cite{EhrenbergDeshpandeEtAl2022simulation}. For the last two only quasipolynomial time algorithms are known in some parameter regimes.


The recently proposed efficiently implementable Lindbladian (Eq.~\eqref{eq:QMC_lindbladian}) that satisfies the Kubo–-Martin–-Schwinger (KMS) detailed balance condition (a quantum generalization of the detailed balance condition for Markov semigroups) in \cite{ChenKastoryanoBrandaoGilyen2023quantum,chen2023efficient} offers a new and promising approach to prepare the Gibbs states of a wider class of non-integrable quantum systems. Similar to classical Markov chain Monte Carlo, the efficiency of the method crucially depends on the \emph{mixing time} of the Lindbladian dynamics, i.e., the time needed for the system to evolve close to the stationary state. Due to its importance in characterizing natural quantum systems and algorithmic performance, the mixing time has been the central topic of many studies \cite{temme2014hypercontractivity,ding2024polynomial,bardet2024entropy,alicki2009thermalization,ramkumar2024mixing,KochanowskiAlhambraCapelRouze2024rapid,BardetCapelLiEtAl2023rapid,RouzeFrancaAlhambra2024optimal,capel2024quasi,kastoryano2013quantum,temme2015fast,ChenLiLuYing2024randomized,Rakovszky2024bottlenecks,LiLu2024quantum,Barthel2022solving,Fang2024mixing,DingLiLin2024efficient,DingChenLin2024single}.


As an example of the above approach, in Ref.~\cite{rouze2024efficient}, the authors analyzed the mixing time of the above Lindbladian for preparing high-temperature Gibbs states of quantum spin systems, and thereby showing that these Gibbs states can be efficiently prepared on quantum computers. By considering the parent Hamiltonian of the Lindbladian, the authors were able to use the stability result for the spectral gap of frustration-free Hamiltonians \cite{MichalakisZwolak2013stability} to show that the spectral gap of the Lindbladian is lower bounded by a constant when the temperature is above a constant threshold that is independent of the system size.
This mixing time bound was later improved in \cite{RouzeFrancaAlhambra2024optimal}.
We should note, however, that this is not the only approach to prepare such high-temperature Gibbs states. In a subsequent work, Ref.~\cite{BakshiLiuMoitraTang2024high}, the authors showed that for temperature beyond a certain threshold, the Gibbs state of a quantum spin system becomes unentangled, and they propose an efficient classical algorithm to sample from the Gibbs state. It is at this moment unclear how the temperature thresholds in these works compare with each other.


From the above discussion we can see that, to study Gibbs states that have intrinsically quantum properties such as entanglement, one may need to go to lower temperature. Moreover, many quantum systems of interest are not spin systems, but are fermionic in nature, which requires us to deal with a different notion of locality.
We therefore ask the following question:
\begin{quote}
    \textit{Can we prove fast mixing of fermionic systems at any temperature, given sufficiently weak interaction?}
\end{quote}
In this work, we will give an affirmative answer to this question.

\paragraph{Main result.}

The quantum system we consider is a $D$-dimensional cubic lattice $\Lambda$ consisting of $n$ sites. 
On each site we have one Dirac fermion, which is equivalent to $2$ Majorana modes. This gives rise to $2n$ Majorana operators $\gamma_{j}$, $j=1,2,\cdots,2n$, satisfying the anti-commutation relation $\{\gamma_{j},\gamma_{k}\} = 2\delta_{jk}$.
The Hamiltonian we study takes the following form
\begin{equation}
\label{eq:general_hamiltonian}
    H = H_0 + V,
\end{equation}
where $H_0 = \sum_{jk} h_{jk} \gamma_j \gamma_k$ is the non-interacting (quadratic) part,
and $V$ is a general fermionic operator describing the interaction. $V$ here preserves parity, i.e., it only contain terms that are products of an even number of Majorana operators. 
Moreover, we assume that $H_0$ and $V$ are $(1,r_0)$-geometrically local and $(U,r_0)$-geometrically local respectively, following the definition below:
\begin{defn}[$(\xi,r_0)$-geometrically local operators]
\label{defn:geometrically_local}
    We say an operator $W$ is $(\xi,r_0)$-geometrically local if 
    \[
    W = \sum_{B\in \mathcal{B}(r_0)} W_{B},
    \]
    where $\mathcal{B}(r)$ denotes the set of balls of radius $r$ and $\|W_B\|\leq \xi$. 
\end{defn}
For $H_0$ this means that $h_{jk}=0$ if the distance between $j$ and $k$ is greater than $r_0$. $U$ here is the \emph{interaction strength}. One can see that the Hamiltonian is quadratic when $U=0$. Our main result is as follows:
\begin{thm}[Main result]
    \label{thm:main_result}
    Let $H=H_0+V$ be a Hamiltonian on a $D$-dimensional cubic lattice $\Lambda$ consisting of $n$ sites. $H_0$ is quadratic in Majorana operators and is $(1,r_0)$-geometrically local, and $V$ is parity-preserving and is $(U,r_0)$-geometrically local. Then for any inverse temperature $\beta>0$, there exists $U_\beta>0$ such that when $U<U_\beta$, the Lindbladian defined in \eqref{eq:QMC_lindbladian} has a unique stationary state satisfying the fermionic superselection rule, and the $\epsilon$-mixing time (Definition~\ref{defn:mixing_time}) of the Lindbladian is at most
    \[
    C(n+\log(1/\epsilon)).
    \]
    where $C,U_\beta>0$ are constants that depend only on $r_0,D,\beta$.
\end{thm}
The proof can be found in Section~\ref{sec:fast_mixing}. Note that here we only focus on quantum states that satisfy the fermionic superselection rule, i.e., there is no superposition between even- and odd-parity states. This affects the way we define the mixing time (see Definition~\ref{defn:mixing_time}).

\paragraph{Proof ideas.} We adopt the approach in \cite{rouze2024efficient} of mapping to a parent Hamiltonian and then lower bounding the spectral gap of that parent Hamiltonian.\footnote{\label{fn:spectral_gap} In this work, whenever we talk about the spectral gap of a Hamiltonian $K$, we mean the gap between the top eigenvalue and the second largest eigenvalue, unless otherwise stated. This is because the Lindbladian spectrum lies on the negative part of the real axis, and the stationary state corresponds to the 0 eigenvalue, which is the top eigenvalue. However, all results about the gap between the ground state and the first excited state that we need still apply because one can simply consider $-K$ instead of $K$.}
The hope is that for a fermionic Hamiltonian with weak interaction, it can be mapped to a parent Hamiltonian that consists of a large non-interacting part plus small and quasi-local interaction terms (though these terms can add up to be as large as linear in $n$). With such a parent Hamiltonian, we then use the stability result about the gap of free-fermion Hamiltonian \cite{Hastings2019stability} to obtain a gap lower bound for the parent Hamiltonian, which implies a gap lower bound for the Lindbladian. A mixing time bound then follows. However, we will see that we need to overcome several difficulties to use this approach .

The first difficulty is that the parent Hamiltonian construction in \cite{chen2023efficient} does not work for fermions. In \cite{rouze2024efficient}, for bosonic spin systems, this construction was sufficient because it is based on a ``vectorization'' mapping from super-operators to operators on an enlarged Hilbert space that preserves both \emph{locality} and the \emph{spectrum} using the tensor product of two algebras. The mapping preserves locality because the tensor product preserves locality in bosonic spin systems. The mapping preserves the spectrum because it is a $^*$-isomorphism between two $C^*$-algebras. For fermions, we have the additional requirement that the parent Hamiltonian needs to be generated by a set of Majorana operators satisfying the canonical anti-commutation relations. This immediately calls into question the use of the tensor product: one would need to redefine it as a native fermionic operation. In Appendix~\ref{sec:how_not_to_vectorize} we will use an example to show that the most obvious attempt to redefine the tensor product will fail to make the mapping an algebra isomorphism. While we cannot rule out the possibility that with sufficient effort one can get a usable redefinition, it seems a highly non-trivial task, and may likely end up yielding something resembling the third quantization \cite{prosen2008third}, which is what we are going to use to partially solve the problem.


The third quantization is a method to map a quasi-free Lindbladian to a quadratic operator in an enlarged fermionic system, which we extend to arbitrary Lindbladians. 
However, it does not offer as clean a solution as in the bosonic scenario, because in order to preserve locality, the third quantization mapping fails to be an algebra isomorphism, and thereby does not preserve the whole spectrum. In fact, one needs to decide to preserve the spectrum of either the even or the odd-parity sector, which leads to two different mappings.
Fortunately, the fermionic superselection rule, which prohibits the superposition of even and odd particle numbers, allows us to focus on observables that are even polynomials of Majorana operators, which correspond to the even-parity sector. 
This is sufficient for establishing the mixing time.
Several other questions that are obvious for the vectorization mapping in \cite{chen2023efficient} also require careful study in the new setting, such as 
\begin{enumerate}
    \item Do errors in the Lindbladian result in amplified errors in the parent Hamiltonian through the third quantization?
    \item Is the parent Hamiltonian self-adjoint? 
    \item Can the parent Hamiltonian have a odd-parity eigenstate with non-negative eigenvalue thus making the top eigenstate no longer uniquely correspond to the Gibbs state?
\end{enumerate}
We answer the first question in the negative in Appendix~\ref{sec:norm_preserve_third_quantization}, the second question in the affirmative in Section~\ref{sec:constructing_the_parent_hamiltonian} (note that it does not follow directly from the quantum detailed balance condition which only ensures self-adjointness in the even-parity sector), and the third question in the negative in Theorem~\ref{thm:spectrum_parent_ham}.


With the parent Hamiltonian constructed form third quantization, we then need to show that the non-interacting parent Hamiltonian, denoted by $H_0^{\mathrm{parent}}$, has a constant gap lower bound when the jump operators are generated from the full set of single Majorana operators. 
This is suggested by the numerical results in \cite{LiZhanLin2024dissipative}. We compute the spectral gap of $H_0^{\mathrm{parent}}$ as a function of the inverse temperature $\beta$ and the single-particle eigenenergies of $H_0$, thus obtaining the needed lower bound (Proposition~\ref{prop:non_interacting_part_gap}). 

Next, we show that the interacting part introduces a quasi-local perturbation to the parent Hamiltonian. More precisely, we show that the parent Hamiltonian $H^{\mathrm{parent}}$ can be written as $H_0^{\mathrm{parent}}+V^{\mathrm{parent}}$, where $H_0^{\mathrm{parent}}$ is quasi-local (Proposition~\ref{prop:non_interacting_part_decay}), and the interacting part $V^{\mathrm{parent}}$ consists of quasi-local terms, each of which scales linearly in the interaction strength $U$ (Proposition~\ref{prop:interacting_part_decay}). The proof makes extensive use of the Lieb-Robinson bound \cite{LiebRobinson1972finite}, which as discussed in \cite{Hastings2004decay} works equally well for bosonic and fermionic systems. This structure of $H^{\mathrm{parent}}$ then enables us to use the spectral gap stability result in \cite{Hastings2019stability} to lower bound the spectral gap of $H^{\mathrm{parent}}$. The spectral gap bound then yields a mixing time bound following standard analysis as done in Corollary~\ref{cor:convergence_of_state}. We note that carefully studying the structure of $V^{\mathrm{parent}}$ is necessary because our result does not follow from an eigenvalue perturbation theory: even though the perturbation is locally small in $U$, the entire interacting part $V^{\mathrm{parent}}$ has spectral norm on the order of $nU\gg 1$.

\paragraph{Applications and implications.}

Our result covers a wide range of quantum systems. As an example, we can consider the Fermi-Hubbard model on a $D$-dimensional cubic lattice $\Lambda$:
\begin{equation}
    \label{eq:fermi_hubbard_ham}
    H = -\sum_{\braket{i,j},\varsigma}(c_{i,\varsigma}^\dag c_{j,\varsigma}+c_{j,\varsigma}^\dag c_{i,\varsigma}) + U\sum_i \left(n_{i,\uparrow}-\frac{1}{2}\right)\left(n_{i,\downarrow}-\frac{1}{2}\right) - \mu\sum_i (n_{i,\uparrow}+n_{i,\downarrow}),
\end{equation}
where $c_{i,\varsigma}$ is the fermionic annihilation operator on site $i$ with spin $\varsigma\in\{\uparrow,\downarrow\}$, $n_{i,\varsigma}=c_{i,\varsigma}^\dag c_{i,\varsigma}$ is the corresponding number operator, $U$ is the interaction strength, and $\mu$ is the chemical potential. Theorem~\ref{thm:main_result} immediately implies that
\begin{cor}
    \label{cor:application_to_fermi_hubbard}
    For $H$ given in \eqref{eq:fermi_hubbard_ham}, for any $\beta>0$, there exists $U_\beta>0$ such that if $|U|\leq U_\beta$, then the mixing time of the Lindbladian given in \eqref{eq:QMC_lindbladian} is at most $C(n+\log(1/\epsilon))$, where $C,U_\beta>0$ are constants that depend only on $D,\beta,\mu$.
\end{cor}

The mixing time estimate implies an efficient quantum algorithm to prepare the Gibbs state. 
\begin{cor}
    \label{cor:efficient_prep_gibbs_state}
    Under the same assumption as Theorem~\ref{thm:main_result}, we can prepare the Gibbs state $\sigma=(1/Z)e^{-\beta H}$ where $Z=\Tr[e^{-\beta H}]$ to trace distance $\epsilon$ on a quantum computer with runtime 
    \[
    \mathcal{O}(n^3\mathrm{polylog}(n/\epsilon))
    \]
    for any inverse temperature $\beta=\mathcal{O}(1)$ if $U<U_\beta=\mathcal{O}(1)$.
\end{cor}
We will explain how we arrive at the $\mathcal{O}(n^3\mathrm{polylog}(n/\epsilon))$ runtime: using the algorithm in \cite[Theorem~I.2]{chen2023efficient} we can see that preparing the Gibbs state requires simulating the dynamics $e^{-iHt}$ for total time $\mathcal{O}(nt_{\mathrm{mix}}\mathrm{polylog}(n t_{\mathrm{mix}}/\epsilon))$ (the factor $n$ comes from normalizing the jump operators as is required in \cite{chen2023efficient}),\footnote{We thank {\v S}t{\v e}p{\' a}n {\v S}m{\' i}d for pointing it out.} where $t_{\mathrm{mix}}$ is the mixing time, and this is the dominant part of the cost. Using the mixing time bound in Theorem~\ref{thm:main_result} we can see that the require evolution time is $\mathcal{O}(n^2\mathrm{polylog}(n/\epsilon))$. Using the Hamiltonian simulation algorithm for lattice Hamiltonians in \cite{HaahHastingsKothariLow2021quantum}, we can simulate the dynamics with a $\mathcal{O}(n)$ overhead up to a poly-logarithmic factor in terms of the gate complexity. We then arrive at the total runtime of $\mathcal{O}(n^3\mathrm{polylog}(n/\epsilon))$ measured in gate complexity. The circuit depth is $\mathcal{O}(n^2\mathrm{polylog}(n/\epsilon))$.

Besides algorithmic applications, our result also reveals important information about properties of the Gibbs state itself using the correspondence between the top eigenstate of the parent Hamiltonian and the Gibbs state as discussed in Proposition~\ref{prop:correspondence_top_eigenstate_gibbs_state}. As an example, we have the following corollary about correlation decay:
\begin{cor}
\label{cor:exp_decay_correlation}
    Under the same assumptions as Theorem~\ref{thm:main_result}, there exists $C,\xi>0$ that depends only only on $\beta,r_0,D$, such that for an observable $X\in B(\mathcal{H})$ supported on a constant number of adjacent sites, and any observable $Y\in B(\mathcal{H})$ with support $S$, 
    \[
    |\Tr[\sigma X Y]-\Tr[\sigma X]\Tr[\sigma Y]|\leq C\|X\|\|Y\||S|e^{-d/\xi},
    \]
    where $d$ is the distance between the supports of $X$ and $Y$.
\end{cor}

A proof is provided in Section~\ref{sec:recovering_gibbs_state_from_top_eigenstate}. We note that decay of correlation for Gibbs states were previously only proved for a limited number of settings, typically for 1D systems or for high temperatures \cite{CapelMoscolariTeufelWessel2023decay,BluhmCapelEtAl2022exponential,KuwaharaKatoBrandao2020clustering}, whereas our result is not constrained by the geometry of the lattice or by the temperature, but is instead constrained by the non-quadratic interaction strength.


From the above we can see that, at any constant temperature, our mixing time result allows us to efficiently prepare the Gibbs states of a wide range of non-integrable fermionic Hamiltonians on a quantum computer. On the classical side, we are not aware of any such rigorous result. This then raises the question: is this a place where we can demonstrate significant quantum advantage? There are many reasons to be cautious in this regard. First of all, the weak interaction regime we are studying is generally considered easy for classical heuristic algorithms. Commonly used algorithms such as Hartree-Fock achieve good empirical performance even though they generally cannot achieve arbitrarily high accuracy. The recent result on the lack of entanglement for high-temperature Gibbs states of quantum spin systems in \cite{BakshiLiuMoitraTang2024high} also raises concern that in the weak interaction regime the Gibbs states may have structures that enable efficient classical simulation. All of these lead to an intriguing question: can we design a provably efficient classical algorithm to simulate Gibbs state properties in this regime, or can we prove that no classical algorithm can do so? Either outcome would significantly advance our understanding of the simulability of such systems and quantum advantage in general.

\paragraph{Notations.} We write $A\succcurlyeq B$ to indicate that $A-B$ is positive semi-definite.
We use the following convention for Fourier transforms for a function $f(t)$:
\begin{equation}
\label{eq:fourier_transform}
    \hat{f}(\omega) = \frac{1}{\sqrt{2\pi}}\int_{-\infty}^{\infty} e^{-i\omega t}f(t)\dd t.
\end{equation}
This leads to the corresponding definition of the inverse Fourier transform for a function $h(\omega)$:
\begin{equation}
\label{eq:inverse_fourier_transform}
    \check{h}(t) = \frac{1}{\sqrt{2\pi}}\int_{-\infty}^{\infty} e^{i\omega t}h(\omega)\dd \omega.
\end{equation}
For an operator $W$, we use $W(t)$ to denote its time evolution in the Heisenberg picture: $W(t)= e^{iHt}W e^{-iHt}$.

\paragraph{Organization.} The rest of the paper is organized as follows. In Section~\ref{sec:problem_setup} we introduce the basic setup, and in particular the Lindbladian for which we prove the mixing time bound. In Section~\ref{sec:background} we introduce several theoretical tools that will be needed in the proof: third quantization, the stability of free-fermion Hamiltonians, and the Lieb-Robinson bound. In Section~\ref{sec:constructing_the_parent_hamiltonian} we construct the parent Hamiltonian using third quantization, and study its properties (e.g., locality and how its spectrum reflects that of the Lindbladian). In Section~\ref{sec:structure_parent_ham} we state the results about the quasi-locality and the spectral gap of the parent Hamiltonian, which imply fast mixing. Sections~\ref{sec:spectral_gap_non_interacting}, \ref{sec:dissipative}, and \ref{sec:coherent_term} are devoted to proving these results. In Section~\ref{sec:spectral_gap_non_interacting} we compute the spectral gap for the non-interacting part of the parent Hamiltonian. In Section~\ref{sec:dissipative} we prove the quasi-locality of the part of the parent Hamiltonian corresponding to the dissipative part of the Lindbladian, and in Section~\ref{sec:coherent_term} we deal with the coherent part. The appendices contain a table of symbols used in this work (Appendix~\ref{sec:table_of_symbols}), and proofs of several technical lemmas.

\section{Problem setup}
\label{sec:problem_setup}

The quantum system we consider is a $D$-dimensional cubic lattice $\Lambda$ consisting of $n$ sites. 
We use $d(x,y)$ to denote the Euclidean distance between sites $x$ and $y$.
The Hilbert space we denote by $\mathcal{H}$, and the operator algebra by $B(\mathcal{H})$. The operators on $B(\mathcal{H})$, which we call super-operators, also form an algebra which we denote by $B(B(\mathcal{H}))$.
The set of quantum states we denote by $D(\mathcal{H})$, i.e.,
\begin{equation}
    \label{eq:set_of_quantum_states}
    D(\mathcal{H}) = \{\rho\in B(\mathcal{H}):\rho \succcurlyeq 0, \Tr[\rho]=1\}.
\end{equation}

On each site we have one Dirac fermion corresponding to $2$ Majorana modes. This gives rise to $2n$ Majorana operators $\gamma_{j}$, $j=1,2,\cdots,2n$, satisfying the anti-commutation relation
\begin{equation}
    \label{eq:anticommutation_majorana}
    \{\gamma_{j},\gamma_{k}\} = 2\delta_{jk}.
\end{equation}
These Majorana operators can be combined to define the fermionic creation and annihilation operators through
\begin{equation}
    \label{eq:fermion_creation_annihilation}
    c_j = \frac{1}{2}(\gamma_{2j-1}+i\gamma_{2j}),\quad c_j^\dag = \frac{1}{2}(\gamma_{2j-1}-i\gamma_{2j}).
\end{equation}
For a Majorana mode $j$, we denote by $\mathfrak{s}(j)$ its corresponding lattice site. Therefore the distance between two Majorana modes $j,k$ is $d(\mathfrak{s}(j),\mathfrak{s}(k))$.

\subsection{The Hamiltonian}

The Hamiltonian we focus on takes the form
\begin{equation}
    H = H_0 + V,
\end{equation}
where
\begin{equation}
    H_0 = \sum_{jk} h_{jk} \gamma_j \gamma_k,
\end{equation}
and $V$ is a general fermionic operator that preserves parity. In other words $V$ only contain terms that are products of an even number of Majorana operators. 
Moreover, we assume that $H_0$ and $V$ are $(1,r_0)$-geometrically local and $(U,r_0)$-geometrically local respectively, as defined in Definition~\ref{defn:geometrically_local}
For $H_0$ this means that $h_{jk}=0$ if $d(\mathfrak{s}(j),\mathfrak{s}(k))>r_0$.

\subsection{The Lindbladian}

We use the Lindbladian $\mathcal{L}:B(\mathcal{H})\to B(\mathcal{H})$ constructed in \cite{chen2023efficient} that was also used in \cite{rouze2024efficient}.
\begin{equation}
\label{eq:QMC_lindbladian}
    \mathcal{L}(\rho) = -i[B,\rho] + \sum_{j}\int_{-\infty}^{\infty}\dd \omega\eta(\omega)\left(A_j(\omega)\rho A_j(\omega)^\dag - \frac{1}{2}\{A_j(\omega)^\dag A_j(\omega),\rho\}\right),
\end{equation}
where
\begin{equation}
\label{eq:jump_operators}
    A_j(\omega) = \frac{1}{\sqrt{2\pi}}\int_{-\infty}^{\infty} \gamma_j(t) e^{-i\omega t} f(t) \dd t,
\end{equation}
and $B=\sum_j B_j$,
\begin{equation}
\label{eq:coherent_terms}
    B_j =  \int_{-\infty}^\infty \dd t b_1(t) \int_{-\infty}^\infty \dd t' b_2(t') \gamma_j(\beta (t'-t))^\dag \gamma_j(-\beta(t+t')).
\end{equation}
In the above $\gamma_j(t)=e^{iHt}\gamma_j e^{-iHt}$, and\begin{equation}
    \label{eq:f(t),b1(t),b2(t)}
    \begin{aligned}
        &f(t) = \frac{1}{\sqrt{\beta\sqrt{\pi/2}}}e^{-\frac{t^2}{\beta^2}},\quad b_1(t) = 2\sqrt{\pi}e^{\frac{1}{8}}\left(\frac{1}{\cosh(2\pi t)}*_t \sin(-t)e^{-2t^2}\right), \\
        &b_2(t) = \frac{1}{2\pi^{3/2}}e^{-4t^2-2it},\quad  \eta(\omega)=e^{-\frac{(\beta\omega+1)^2}{2}}.
    \end{aligned}
\end{equation}
Here $*_t$ denotes convolution in the variable $t$.

This Lindbladian is designed to satisfy the KMS-detailed balance condition (KMS-DBC), which we will introduce below. The KMS-inner product corresponding to a full-rank quantum state $\sigma$ is defined as
\begin{equation}
\label{eq:defn_KMS_inner_product}
    \braket{X,Y}_{\mathrm{KMS}}=\Tr[\sigma^{1/2}X^\dag \sigma^{1/2} Y],
\end{equation}
for any $X,Y\in B(\mathcal{H})$. A super-operator $\mathcal{X}$ satisfies the KMS-DBC if
\begin{equation}
    \label{eq:defn_KMS_DBC}
    \braket{\mathcal{X}(X),Y}_{\mathrm{KMS}} = \braket{X,\mathcal{X}(Y)}_{\mathrm{KMS}},
\end{equation}
i.e., it is self-adjoint under the KMS-inner product. For the Lindbladian $\mathcal{L}$ in \eqref{eq:QMC_lindbladian}, $\mathcal{L}^\dag$ is self-adjoint under the KMS-inner product and therefore satisfies the KMS-DBC.

\section{Background}
\label{sec:background}

In this section we will introduce some tools that are needed in the proof of our main result.

\subsection{Third quantization}
\label{sec:third_quantization}

Third quantization is a useful method to map a Lindbladian super-operator to an operator in an enlarged Hilbert space \cite{prosen2008third}. It was originally used to study quasi-free Lindbladians, i.e., those with jump operators that are linear in the Majorana operators, and coherent parts that are quadratic. Here we extend it to more general Lindbladians.

We define a $2^{2n}$-dimensional Hilbert space $\mathcal{H}'$ spanned by vectors $v_\alpha$, where $\alpha\in\{0,1\}^{2n}$. Define map $\varphi:B(\mathcal{H})\to \mathcal{H}'$,
\begin{equation}
    \label{eq:third_quantization_map}
    \varphi(\gamma_1^{\alpha_1}\gamma_2^{\alpha_2}\cdots \gamma_{2n}^{\alpha_{2n}}) = v_\alpha.
\end{equation}
We let $v_{\alpha}$ form an orthonormal basis by defining the inner product $\braket{\cdot,\cdot}$ through
\begin{equation}
    \label{eq:defn_innter_product_H_prime}
    \braket{v_\alpha,v_{\alpha'}}=\delta_{\alpha\alpha'}.
\end{equation}
This ensures that
\begin{equation}
\label{eq:third_quantization_preserves_frobenius_norm}
    \braket{\varphi(X),\varphi(Y)} = \frac{1}{2^{n}}\Tr[X^\dag Y],
\end{equation}
for any $X,Y\in B(\mathcal{H})$.
Define a new set of creation and annihilation operators through
\begin{equation}
    \label{eq:third_quantization_creation_annihilation}
    \hat{c}_j v_{\alpha} = \delta_{1,\alpha_j}\varphi(\gamma_j \gamma_1^{\alpha_1}\gamma_2^{\alpha_2}\cdots \gamma_{2n}^{\alpha_{2n}}),\quad \hat{c}_j^\dag v_{\alpha} = \delta_{0,\alpha_j}\varphi(\gamma_j \gamma_1^{\alpha_1}\gamma_2^{\alpha_2}\cdots \gamma_{2n}^{\alpha_{2n}}).
\end{equation}
One can check that they satisfy the canonical anti-commutation relation. 
Following \cite{prosen2008third}, we call the fermionic particles defined through these creation and annihilation operators $a$-fermions.
From these operators we then define the corresponding Majorana operators:
\begin{equation}
\label{eq:third_quantization_majorana_ops}
    \hat{\gamma}_{2j-1} = \hat{c}_j + \hat{c}_j^\dag,\quad \hat{\gamma}_{2j} = i(\hat{c}_j^\dag - \hat{c}_j).
\end{equation}
 
\begin{rem}[$a$-fermions on the lattice $\Lambda$]
\label{rem:a_fermion_geometry}
    There is an underlying geometry in this new set of fermionic operators.
    For each site $x$ in the original lattice $\Lambda$ associated with Majorana operators $\gamma_{j}$ and $\gamma_{k}$, we associate $\hat{c}_{j}$ and $\hat{c}_{k}$ with the same site. Therefore each site contains $2$ $a$-Dirac fermion modes and $4$ $a$-Majorana modes. We use $\mathfrak{s}'(j)$ to indicate the site in $\Lambda$ that corresponds to the $a$-Majorana mode $j$.
\end{rem}

$\varphi$ induces a representation of the super-operators on $B(\mathcal{H})$, where each super-operator $\mathcal{X}$ is represented by $\varphi\circ \mathcal{X}\circ \varphi^{-1}$, which is an operator on $\mathcal{H}'$. 
In particular, we consider the left- and right-multiplication with an operator $A$ defined as
\begin{equation}
\label{eq:left_and_right_multiplication}
    L_{A}(X) = A X,\quad R_{A}(X) = XA.
\end{equation}
We have the following rules for the their multiplication
\begin{equation}
\label{eq:left_right_multiply_repr_rules}
    L_{AB} = L_A L_B,\quad R_{AB} = R_B R_A.
\end{equation}
From \cite{prosen2008third}, and more directly from \eqref{eq:third_quantization_creation_annihilation} and \eqref{eq:third_quantization_majorana_ops}, we have
\begin{equation}
    \label{eq:repr_left_right_multiplication_single_majorana}
    \varphi\circ L_{\gamma_j}\circ \varphi^{-1} = \hat{\gamma}_{2j-1},\quad \varphi\circ R_{\gamma_j}\circ \varphi^{-1} = -i\hat{\gamma}_{2j}(-1)^{\hat{n}},
\end{equation}
where $\hat{n}=\sum_{j=1}^{2n} \hat{c}_j^\dag \hat{c}_j$ is the $a$-fermion number operator, which implies that $(-1)^{\hat{n}}$ is the $a$-fermion parity operator. From this we can get the corresponding representations of more general operators:
\begin{equation}
    \label{eq:repr_left_right_multiplication_more_majorana_ops}
    \begin{aligned}
        \varphi\circ L_{\gamma_{j_1}\cdots \gamma_{j_k}}\circ \varphi^{-1} &= \hat{\gamma}_{2j_1-1}\cdots \hat{\gamma}_{2j_k-1},\\
        \varphi\circ R_{\gamma_{j_1}\cdots \gamma_{j_k}}\circ \varphi^{-1} &= -i^k \hat{\gamma}_{2j_k}\cdots\hat{\gamma}_{2j_1} (-1)^{k\hat{n}}.
    \end{aligned}
\end{equation}

The locality of the operators obtained through the above mapping depends crucially on their parities, for which we give a precise definition.
\begin{defn}
\label{defn:parity_of_states}
    We define $\mathcal{H}_{\mathrm{even}}$ and $\mathcal{H}_{\mathrm{odd}}$ to be the $1$-eigenspace and $(-1)$-eigenspace of the parity operator $i^n\gamma_1\gamma_2\cdots\gamma_{2n}$ respectively. We say that the states in $\mathcal{H}_{\mathrm{even}}$ have even parity while those in $\mathcal{H}_{\mathrm{odd}}$ have odd parity.
\end{defn}
We can similarly define $\mathcal{H}'_{\mathrm{even}}$ and $\mathcal{H}'_{\mathrm{odd}}$ for the enlarged Hilbert space as well. For operators, we can also define
\begin{defn}
    \label{defn:parity_of_operators}
    We define $B(\mathcal{H})_{\mathrm{even}}$ to be the subspace of $B(\mathcal{H})$ spanned by $\gamma_1^{\alpha_1}\gamma_2^{\alpha_2}\cdots \gamma_{2n}^{\alpha_{2n}}$ for even $|\alpha|$. Similarly we define $B(\mathcal{H})_{\mathrm{odd}}$ to be the subspace spanned by $\gamma_1^{\alpha_1}\gamma_2^{\alpha_2}\cdots \gamma_{2n}^{\alpha_{2n}}$ for odd $|\alpha|$.
\end{defn}

\begin{rem}
    \label{rem:third_quantization_preserves_parity}
    These two notions of parity are related, as $v\in\mathcal{H}'$ has even/odd parity if and only if $\varphi^{-1}(v)\in B(\mathcal{H})$ has even/odd parity. One can see this from the mapping $\varphi$ defined in \eqref{eq:third_quantization_map} and by observing that $\mathcal{H}'_{\mathrm{even}}$ and $\mathcal{H}'_{\mathrm{odd}}$ are spanned by $v_{\alpha}$ for even $|\alpha|$ and odd $|\alpha|$ respectively.
\end{rem}

We will focus on super-operators that are parity-preserving, as defined below
\begin{defn}
    \label{defn:super_operator_parity_preserving}
    We say $\mathcal{X}\in B(B(\mathcal{H}))$ is parity-preserving if $\mathcal{X}(B(\mathcal{H})_{\mathrm{even}})\subset B(\mathcal{H})_{\mathrm{even}}$ and $\mathcal{X}(B(\mathcal{H})_{\mathrm{odd}})\subset B(\mathcal{H})_{\mathrm{odd}}$.
\end{defn}

We then use third quantization to build a map that preserves locality and also preserves the spectrum for the even sector.
\begin{defn}
    \label{defn:Phi_thrid_quantization}
    We define a linear map $\Phi:B(B(\mathcal{H}))\to B(\mathcal{H}')$ such that, for any $W,W'\in B(\mathcal{H})$,
    \[
    \begin{aligned}
        \Phi(L_W) &= \varphi \circ L_W \circ \varphi^{-1}, \\
        \Phi(R_W) &= \begin{cases}
            \varphi \circ R_W \circ \varphi^{-1}\ &\text{if } W \text{ is even,}  \\
            (\varphi \circ R_W \circ \varphi^{-1})(-1)^{\hat{n}}\ &\text{if } W \text{ is odd,} 
        \end{cases} \\
        \Phi(L_W R_{W'}) &= \Phi(L_W)\Phi(R_{W'}).
    \end{aligned}
    \]
\end{defn}

\begin{rem}
    \label{rem:Phi_is_uniquely_well_defined}
    Because $L_W R_{W'}$ with even/odd $W$ and even/odd $W'$ span disjoint subspaces of $B(B(\mathcal{H}))$, the above map is well-defined.
    Because $L_W R_{W'}$ for $W$ and $W'$ with even/odd parities span the entire $B(B(\mathcal{H}))$, the above definition uniquely defines the map $\Phi(\mathcal{X})$ for all $\mathcal{X}\in B(B(\mathcal{H}))$.
\end{rem}

\begin{rem}
\label{rem:Phi_preserves_locality}
    $\Phi$ defined in Definition~\ref{defn:Phi_thrid_quantization} strictly preserves locality. We can see this by computing
    \[
    \begin{aligned}
        \Phi(L_{\gamma_{j_1}\gamma_{j_2}\cdots \gamma_{j_k}}) &= \hat{\gamma}_{2j_1} \hat{\gamma}_{2j_2}\cdots \hat{\gamma}_{2j_k}, \\
        \Phi(R_{\gamma_{j_1}\gamma_{j_2}\cdots \gamma_{j_k}}) &= -i^k\hat{\gamma}_{2j_k} \hat{\gamma}_{2j_{k-1}}\cdots \hat{\gamma}_{2j_1}.
    \end{aligned}
    \]
\end{rem}

Note that $\Phi$ is not a homomorphism between algebras because even though $L_W$ and $R_{W'}$ commute with each other, $\Phi(L_W)$ and $\Phi(R_{W'})$ do not necessarily commute with each other. However, it still preserves part of the spectrum, as stated in the following lemma:

\begin{lem}
    \label{lem:Phi_preserves_spectrum}
    $\Phi$ satisfies
    \[
    \Phi(\mathcal{X})|_{\mathcal{H}'_{\mathrm{even}}} = \varphi\circ \mathcal{X}\circ \varphi^{-1}|_{\mathcal{H}'_{\mathrm{even}}}.
    \]
    Moreover, when $\mathcal{X}$ is parity-preserving, $\Phi(\mathcal{X})|_{\mathcal{H}'_{\mathrm{even}}}$ and $\mathcal{X}|_{B(\mathcal{H})_{\mathrm{even}}}$ share the same spectrum.
\end{lem}

\begin{proof}
    Note that for $v\in\mathcal{H}'$ with even parity, $(-1)^{\hat{n}}v = v$. Therefore $\Phi(L_W R_{W'})v =\Phi(L_W)\Phi(R_{W'})v =(\varphi\circ L_W R_{W'} \circ \varphi^{-1})v$. From this the equality follows by linearity. Because $\varphi$ preserves parity as discussed in Remark~\ref{rem:third_quantization_preserves_parity}, we can then see that $\Phi(\mathcal{X})|_{\mathcal{H}'_{\mathrm{even}}}$ and $\mathcal{X}|_{B(\mathcal{H})_{\mathrm{even}}}$ must share the same spectrum if $\mathcal{X}$ preserves parity. 
\end{proof}


\subsection{Stability of free-fermion Hamiltonians}
\label{sec:stability_of_free_fermion_hamiltonians}

We will introduce the stability result for the gap of free-fermion Hamiltonians proved in \cite{Hastings2019stability}. It relies crucially on notions of locality, which we will discuss below.

\begin{defn}[Quasi-locality around $x$]
\label{defn:quasi_local}
    Let $x\in\Lambda$ be a lattice site. We say an operator $W$ is $(C,\mu)$-quasi-local around $x$ if $\|W\|\leq C$, and for every integer $r\geq 1$, there exists $W_r$ supported on the ball of radius $r$ centered at $x$ such that
    \[
    \|W-W_r\|\leq C e^{-\mu r},
    \]
    for some constants $C,r>0$.
\end{defn}

We adopt the same definition in \cite{Hastings2019stability} of the decay property of an operator, which characterize the quasi-locality of an extensive operator in the following way 
\begin{defn}[Decay of operator]
\label{defn:decay_operator}
    We say an operator $W$ has $(C,\mu)$-decay if $W$ can be decomposed as
    \[
    W = \sum_{r\geq 1}\sum_{B\in \mathcal{B}(r)} W_{r,B},
    \]
    where $\mathcal{B}(r)$ denotes the set of balls of radius $r$ and $W_{r,B}$ is an operator supported on a ball $B$ and where
    \[
    \max_{B\in\mathcal{B}(r)}\|W_{r,B}\|\leq C e^{-\mu r},
    \]
    for some constants $C,r>0$.
\end{defn}

Naturally, an operator that is the sum of quasi-local terms has nice decay property:
\begin{lem}
    \label{lem:quasi_local_to_decay}
    For any $C,\mu>0$, if an operator $W$ can be written as $W=\sum_{x\in\Lambda} W_x$ where each $W_x$ is $(C,\mu)$-quasi-local around $x$, then $W$ has $(2Ce^\mu,\mu)$-decay.
\end{lem}
\begin{proof}
    Since $W=\sum_{x\in\Lambda} W_x$, and $W_x$ is $(C,\mu)$-quasi-local around $x$, we have
    \[
    W = \sum_{r\geq 0}\sum_{x\in\Lambda} (W_{x,r}-W_{x,r-1}),
    \]
    where $\|W_x-W_{x,r}\|\leq Ce^{-\mu r}$, and we let $W_{x,-1}=0$. Note that $W_{x,r}-W_{x,r-1}$ is supported on the ball with radius $r$ centered at $x$, and $\|W_{x,r}-W_{x,r-1}\|\leq 2Ce^{-\mu (r-1)}$ by the triangle inequality. This tells us that $W$ has $(2Ce^\mu,\mu)$-decay.
\end{proof}

For quadratic operators Ref.~\cite{Hastings2019stability} defined a somewhat different notion of decay, which we state here
\begin{defn}[Decay of coefficient matrix]
\label{defn:decay_coef_matrix}
    For a matrix $B=(B_{jk})$ where each $j$ and $k$ is associated with a lattice site $x\in\Lambda$, and no more than $n_0$ $j$ or $k$ is associated with the same site, then we say $B$ has $[K,\nu]$-decay if 
    \[
    \|B_{xy}\|\leq K e^{-\nu d(x,y)},
    \]
    where $x,y\in\Lambda$, and $B_{xy}$ is the submatrix of $B$ whose rows and columns are associated with $x$ and $y$ respectively. Here $n_0$ is a constant that is independent of $|\Lambda|$.
\end{defn}
We remark that in this work we only need to have $n_0\leq 4$.
With the above notions of locality, we will state the main result in \cite{Hastings2019stability}:
\begin{thm}[Corollary~1 of \cite{Hastings2019stability}]
\label{thm:stability_free_fermion}
    Let $H_{FF}=\sum_{jk} A_{jk} \gamma_j \gamma_k$ be a Hamiltonian on lattice $\Lambda$ with spectral gap $\Delta$. We assume that $A=(A_{jk})$ has $[K,\nu]$-decay. Let $H=H_{FF}+V$, where we assume $V$ has $(C,\mu)$-decay. Then there exists constants $C_0,c_1>0$ depending only on $K,\nu,\mu,\Delta,D$ such that when $C\leq C_0$, the spectral gap of $H$ is at least $\Delta-c_1 C$.
\end{thm}
Note that in this work we focus on the spectral gap between the top eigenvalue and the second largest eigenvalue as explained in Fn.~\ref{fn:spectral_gap}, while \cite{Hastings2019stability} was about the gap between the ground state and the first excited state. These two are however equivalent as one can simply consider $-H$ rather than $H$.

\subsection{The Lieb-Robinson bound}
\label{sec:lieb_robinson_bound}

The Lieb-Robinson bound was used extensively in \cite{Hastings2019stability,rouze2024efficient} and is also needed in our proof. We state the version used in \cite[Lemma~5]{HaahHastingsKothariLow2021quantum}, which follows \cite{LiebRobinson1972finite,Hastings2004lieb,nachtergaele2006lieb,Hastings2006spectral,Hastings2010locality}, with modification:
\begin{lem}
    \label{lem:LR_bound}
    Let $H$ be a fermionic Hamiltonian on $\Lambda$ with even parity that is $(\xi,r_0)$-geometrically local (Definition~\ref{defn:geometrically_local}). For an operator $O_X$ supported on $X\subset\Lambda$, and $\Omega\subset X$, we let  
    \[
    \ell = \lfloor\min\{d(x,y):x\in X, y\in \Lambda\setminus \Omega\} \rfloor.
    \]
    Then
    \[
    \|e^{iHt}O_X e^{-iHt}-e^{iH' t}O_X e^{-iH' t}\|\leq |X|\|O_X\|\frac{(2J|t|)^\ell}{\ell !},
    \]
    where $H'$ consists of the terms in $H$ that are supported within $\Omega$, $J\leq C\xi r_0^{2D}$ for some constant $C>0$ that depends only on $D$.
\end{lem}

We remark that we adapt \cite[Lemma~5]{HaahHastingsKothariLow2021quantum} to the fermionic setting, which, as commented in \cite[Section~4]{HaahHastingsKothariLow2021quantum}, works without needing any modification. $J$ is sometimes known as the Lieb-Robinson velocity. When $H$ is $(\xi,r_0)$-geometrically local, then by Definition~\ref{defn:geometrically_local} we can write it as $H=\sum_{B\in\mathcal{B}(r_0)} H_B$ with $\|H_B\|\leq \xi$. Then
\begin{equation}
\label{eq:LR_velocity_full_ham}
    J = \max_{x\in\Lambda}\sum_{B\ni x}|B|\|H_B\|\leq \xi\times C_1 r_0^D\times C_2 r_0^D=C\xi r_0^{2D}.
\end{equation}
We have used the fact that each term involves at most $|B|\leq C_1 r_0^D$ sites and each site is involved in at most $C_2 r_0^D$ terms, for constants $C_1,C_2>0$ that depend only on $D$.
We will also use the Lieb-Robinson velocity for the non-interacting part $H_0$, for which we can get a tighter bound:
\begin{equation}
\label{eq:LR_velocity_non_interacting}
    J_0 = \max_{x\in\Lambda}\sum_{B\ni x}|B|\|(H_0)_B\|\leq 1\times 2\times C_2' r_0^D=C'\xi r_0^{D}.
\end{equation}
In the above because $H_0$ only contains terms that act on two Majorana modes, $B$ for which the corresponding term $(H_0)_B$ is non-zero can have at most $|B|=2$, and $\|(H_0)_B\|\leq 1$.  

\section{Constructing the parent Hamiltonian}
\label{sec:constructing_the_parent_hamiltonian}

Even though third quantization provides a way to map the Lindbladian to a operator, we still do not have a Hamiltonian because the resulting operator may not be Hermitian. In this section we will use the stationary state $\sigma$ to construct a Hermitian parent Hamiltonian.
We define super-operator $\Gamma\in B(B(\mathcal{H}))$ as follows:
\begin{equation}
    \label{eq:Gamma}
    \Gamma(X) = \sigma^{1/4}X\sigma^{1/4}.
\end{equation}
We then define the map
\begin{equation}
    \label{eq:isometry_phi_tilde}
    \Tilde{\varphi} = \varphi\circ \Gamma,
\end{equation}
This map is an isometry, i.e. it satisfies
\begin{equation}
\begin{aligned}
    \braket{\Tilde{\varphi}(X),\Tilde{\varphi}(Y)} &= \frac{1}{2^n}\Tr[\Gamma(X)^\dag \Gamma(Y)] \\
    &= \frac{1}{2^n}\Tr[\sigma^{1/4}X^\dag \sigma^{1/2}Y\sigma^{1/4}]\\
    &=\frac{1}{2^n}\braket{X,Y}_{\mathrm{KMS}},
\end{aligned}
\end{equation}
where the first equality comes from \eqref{eq:third_quantization_preserves_frobenius_norm}.
We then define, for any $\mathcal{X}\in B(B(\mathcal{H}))$,
\begin{equation}
    \label{eq:defn_Phi_tilde_parent_ham_map}
    \tilde{\Phi}(\mathcal{X}) = \Phi({\Gamma}\circ \mathcal{X} \circ {\Gamma}^{-1}).
\end{equation}
By the definition of $\Phi$ in Definition~\ref{defn:Phi_thrid_quantization}, we have
\begin{equation}
\label{eq:similarity_structure_Phi_restricted_to_even}
\begin{aligned}
    \tilde{\Phi}(\mathcal{X})|_{\mathcal{H}'_{\mathrm{even}}} = \tilde{\varphi}\circ \mathcal{X} \circ \tilde{\varphi}^{-1}|_{\mathcal{H}'_{\mathrm{even}}}.
\end{aligned}
\end{equation}
This immediate implies the following:
\begin{lem}
    \label{lem:Phi_tilde_preserves_spectrum}
    Let $\mathcal{X}\in B(B(\mathcal{H}))$ be a super-operator that preserves parity according to Definition~\ref{defn:super_operator_parity_preserving}, then $\tilde{\Phi}(\mathcal{X})|_{\mathcal{H}'_{\mathrm{even}}}$ and $\mathcal{X}|_{B(\mathcal{H})_{\mathrm{even}}}$ share the same spectrum. 
\end{lem}
In particular for left- and right-multiplication we have
\begin{equation}
\label{eq:Phi_effect_on_left_right_multiply}
    \tilde{\Phi}(L_A) = \Phi(L_{\sigma^{1/4}A\sigma^{-1/4}}),\quad  \tilde{\Phi}(R_A) = \Phi(R_{\sigma^{-1/4}A\sigma^{1/4}}).
\end{equation}

We are now ready to define the parent Hamiltonian through
\begin{equation}
\label{eq:parent_hamiltonian_defn}
    H^{\mathrm{parent}} = \tilde{\Phi}(\mathcal{L}^\dag),
\end{equation}
where $\mathcal{L}$ is the Lindbladian defined in \eqref{eq:QMC_lindbladian} that satisfies the KMS-DBC as defined in \eqref{eq:defn_KMS_DBC}.
\begin{lem}
\label{lem:parent_ham_hermitian}
    The parent Hamiltonian $H^{\mathrm{parent}}$ has the following properties:
    \begin{itemize}
        \item[(i)] $H^{\mathrm{parent}}$ is self-adjoint in $\mathcal{H}'$.
        \item[(ii)] $0$ is an eigenvalue of both $H^{\mathrm{parent}}$ and $\mathcal{L}^\dag|_{B(\mathcal{H})_{\mathrm{even}}}$. 
        Let $\Delta$ and $\Delta'$ be the gap between 0 and the largest non-positive eigenvalue for $H^{\mathrm{parent}}$ and $\mathcal{L}^\dag|_{B(\mathcal{H})_{\mathrm{even}}}$ respectively.
        Then $\Delta\leq \Delta'$. 
    \end{itemize}
\end{lem}

\begin{proof}
    For (i), note that by \eqref{eq:similarity_structure_Phi_restricted_to_even} and the fact that $\mathcal{L}$ satisfies KMS-DBC (as defined in \eqref{eq:defn_KMS_DBC}), we can see that $H^{\mathrm{parent}}|_{\mathcal{H}'_{\mathrm{even}}}$ is self-adjoint. This implies
    \begin{equation}
    \label{eq:ham_self_adjoint_even}
        H^{\mathrm{parent}}|_{\mathcal{H}'_{\mathrm{even}}}=(H^{\mathrm{parent}})^\dag|_{\mathcal{H}'_{\mathrm{even}}}.
    \end{equation}
    We consider enlarging the Hilbert space $\mathcal{H}$ by one Dirac mode, but not adding anything to the Hamiltonian $H$. Then going through the same mapping from $H$ to $\mathcal{L}$ (through \eqref{eq:QMC_lindbladian}) and then $\tilde{\Phi}(\mathcal{L}^\dag)$, we obtain a new parent Hamiltonian $\bar{H}^{\mathrm{parent}}$ in an enlarged Hilbert space $\bar{\mathcal{H}}'$ involving $2n+2$ Dirac modes. By \eqref{eq:repr_left_right_multiplication_more_majorana_ops} and the fact that $\sigma$ does not involve the modes $2n+1$ and $2n+2$, we can see that $\bar{\mathcal{H}}'$ only acts on the Dirac modes $1$ to $2n$, and is of exactly the same form as $H^{\mathrm{parent}}$ when written as a multivariate polynomial of Majorana operators, even though they act on different Hilbert spaces. Moreover, by the same argument as for $H^{\mathrm{parent}}$ we have
    \begin{equation}
        \label{eq:enlarged_ham_self_adjoint_even}
        \bar{H}^{\mathrm{parent}}|_{\bar{\mathcal{H}}'_{\mathrm{even}}} = (\bar{H}^{\mathrm{parent}})^\dag|_{\bar{\mathcal{H}}'_{\mathrm{even}}}.
    \end{equation}

    We revert back to the creation and annihilation operators in \eqref{eq:third_quantization_majorana_ops}, where $\hat{c}_j = (\hat{\gamma}_{2j-1}+i\hat{\gamma}_{2j})/2$. This also gives us a vacuum state $\ket{\mathrm{vac}}$ for $\mathcal{H}'$ that is in the kernel of every $\hat{c}_j$ for $j=1,2,\cdots,2n$. For $\bar{\mathcal{H}}'$, we can likewise get a vacuum state $\ket{\overline{\mathrm{vac}}}$ that is also in the kernel of $\hat{c}_{2n+1}$ and $\hat{c}_{2n+2}$. For any state
    \[
    v = \sum_{\alpha\in \{0,1\}^{2n}} v_\alpha (\hat{c}_1^\dag)^{\alpha_1}\cdots (\hat{c}_{2n}^\dag)^{\alpha_{2n}}\ket{\mathrm{vac}}\in\mathcal{H}',
    \]
    we define
    \[
    \bar{v} = \sum_{\alpha\in \{0,1\}^{2n}} v_\alpha (\hat{c}_1^\dag)^{\alpha_1}\cdots (\hat{c}_{2n}^\dag)^{\alpha_{2n}}\ket{\overline{\mathrm{vac}}}\in\bar{\mathcal{H}}'.
    \]
    Because $H^{\mathrm{parent}}$ and $\bar{H}^{\mathrm{parent}}$ are represented by the same polynomial of Majorana operators, we have
    \begin{equation}
    \label{eq:H_and_H_bar_parent_same_expect}
        \braket{v,H^{\mathrm{parent}}w} = \braket{\bar{v},\bar{H}^{\mathrm{parent}}\bar{w}},
    \end{equation}
    for any $v,w\in\mathcal{H}'$.

    Then we consider $v,w\in\mathcal{H}'_{\mathrm{odd}}$. Note that $\hat{c}_{2n+1}^\dag \bar{v}, \hat{c}_{2n+1}^\dag \bar{w} \in \bar{\mathcal{H}}'_{\mathrm{even}}$. We then have
    \[
    \begin{aligned}
        &\braket{v,H^{\mathrm{parent}}w} = \braket{\bar{v},\bar{H}^{\mathrm{parent}}\bar{w}} = \braket{\hat{c}_{2n+1}^\dag\bar{v},\bar{H}^{\mathrm{parent}}\hat{c}_{2n+1}^\dag\bar{w}}\\
        &=\braket{\hat{c}_{2n+1}^\dag\bar{v},(\bar{H}^{\mathrm{parent}})^\dag\hat{c}_{2n+1}^\dag\bar{w}}=\braket{\bar{v},(\bar{H}^{\mathrm{parent}})^\dag\bar{w}}=\braket{v,(H^{\mathrm{parent}})^\dag w},
    \end{aligned}
    \]
    where in the third equality we have used \eqref{eq:H_and_H_bar_parent_same_expect}, and in the second and fourth equalities we have used the fact that $\bar{H}^{\mathrm{parent}}$ and $\bar{v},\bar{w}$ do not involve any operator on Dirac mode $2n+1$. This then implies
    \[
    H^{\mathrm{parent}}|_{\mathcal{H}'_{\mathrm{odd}}}=(H^{\mathrm{parent}})^\dag|_{\mathcal{H}'_{\mathrm{odd}}}.
    \]
    Because $\mathcal{H}'_{\mathrm{even}}$ and $\mathcal{H}'_{\mathrm{odd}}$ are both invariant subspaces of $H^{\mathrm{parent}}$, the above equation combined with \eqref{eq:ham_self_adjoint_even} yields $H^{\mathrm{parent}}=(H^{\mathrm{parent}})^\dag$.

    For (ii), let $\Delta''$ be the gap between 0 and the largest non-positive eigenvalue of $H^{\mathrm{parent}}|_{\mathcal{H}'_{\mathrm{even}}}$.
    Since the spectrum of $H^{\mathrm{parent}}|_{\mathcal{H}'_{\mathrm{even}}}$ is a subset of the spectrum of $H^{\mathrm{parent}}$, and 0 is an eigenvalue of both of them (corresponding to eigenvector $v=\varphi(\sigma^{1/2})\in\mathcal{H}'_{\mathrm{even}}$), we have $\Delta\leq \Delta''$. 
    By Lemma~\ref{lem:Phi_tilde_preserves_spectrum} we have $\Delta''=\Delta'$.
    Therefore $\Delta'\geq\Delta$.
\end{proof}

The spectral gap of $\mathcal{L}^\dag|_{B(\mathcal{H})_{\mathrm{even}}}$ is sufficient for establishing the convergence towards to stationary state.
\begin{lem}
    \label{lem:convergence_observable_mean_zero}
    We assume that $\mathcal{L}^\dag$ satisfies the KMS-DBC \eqref{eq:defn_KMS_DBC}, the spectral gap between 0 and the second largest eigenvalue of $\mathcal{L}^\dag|_{B(\mathcal{H})_{\mathrm{even}}}$ is $g>0$, and $\mathrm{ker}(\mathcal{L}^\dag|_{B(\mathcal{H})_{\mathrm{even}}})=\mathrm{span}(I)$.
    Let $Y\in B(\mathcal{H})$ be self-adjoint with even parity (Definition~\ref{defn:parity_of_operators}) satisfying $\Tr[\sigma Y]=0$, then, denoting $Y(t)=e^{t\mathcal{L}^\dag}(Y)$,
    \[
    \|Y(t)\|_{\mathrm{KMS}} \leq e^{-gt} \|Y\|_{\mathrm{KMS}},
    \]
    where for any $X\in B(\mathcal{H})$, $\|X\|_{\mathrm{KMS}}=\braket{X,X}_{\mathrm{KMS}}^{1/2}$.
\end{lem}

\begin{proof}
    Because $\mathcal{L}^\dag$ satisfies the KMS-DBC \eqref{eq:defn_KMS_DBC}, and that $\mathrm{ker}(\mathcal{L}^\dag|_{B(\mathcal{H})_{\mathrm{even}}})=\mathrm{span}(I)$, for the spectral gap we have the variational form
    \[
    -g=\sup_{\substack{X\in B(\mathcal{H})_{\mathrm{even}} \\ \braket{X,I}_{\mathrm{KMS}}=0}} \frac{\braket{X,\mathcal{L}^\dag|_{B(\mathcal{H})_{\mathrm{even}}} X}_{\mathrm{KMS}}}{\braket{X,X}_{\mathrm{KMS}}}.
    \]
    Therefore when $\braket{X,I}_{\mathrm{KMS}}=0$, we have 
    \[
    \braket{X,\mathcal{L}^\dag|_{B(\mathcal{H})_{\mathrm{even}}} X}_{\mathrm{KMS}}\leq -g \braket{X,X}_{\mathrm{KMS}}.
    \]
    The inequality we want to prove follows immediately from
    \[
    \begin{aligned}
        \frac{\dd}{\dd t}\braket{Y(t),Y(t)}_{\mathrm{KMS}} &= 2\braket{Y(t),\mathcal{L}^\dag Y(t)}_{\mathrm{KMS}} \\
        &= 2\braket{Y(t),\mathcal{L}^\dag|_{B(\mathcal{H})_{\mathrm{even}}} Y(t)}_{\mathrm{KMS}} \\
        &\leq -2g\braket{Y(t),Y(t)}_{\mathrm{KMS}},
    \end{aligned}
    \]
    where we have used the fact that $\braket{Y(t),I}_{\mathrm{KMS}}=\Tr[\sigma Y(t)]=\Tr[\sigma Y]=0$ for all $t$.
\end{proof}

\begin{cor}
\label{cor:convergence_of_expect_val}
   We assume that $\mathcal{L}^\dag$ satisfies the KMS-DBC \eqref{eq:defn_KMS_DBC}, the spectral gap between 0 and the second largest eigenvalue of $\mathcal{L}^\dag|_{B(\mathcal{H})_{\mathrm{even}}}$ is $g>0$, and $\mathrm{ker}(\mathcal{L}^\dag|_{B(\mathcal{H})_{\mathrm{even}}})=\mathrm{span}(I)$.
    Let $X\in B(\mathcal{H})$ be self-adjoint with even parity (Definition~\ref{defn:parity_of_operators}). For an initial state $\rho\in B(\mathcal{H})_{\mathrm{even}}$ (i.e., satisfying the fermion superselection rule), we have
    \[
    |\Tr[X(t)\rho]-\Tr[X\sigma]|\leq \frac{e^{-gt}}{\sqrt{\sigma_{\min}}}\|X\|,
    \]
    where $\sigma_{\min}$ is the minimum eigenvalue of $\sigma$.
\end{cor}

\begin{proof}
    We let $Y=X-I\braket{I,X}_{\mathrm{KMS}}$. Then $\Tr[Y\sigma]=0$. By Lemma~\ref{lem:convergence_observable_mean_zero} we have
    \[
    \|X(t)-I\braket{I,X}_{\mathrm{KMS}}\|=\|Y(t)\|_{\mathrm{KMS}}\leq e^{-gt}\|X\|,
    \]
    where we have used $\|Y\|_{\mathrm{KMS}}\leq \|X\|_{\mathrm{KMS}}\leq \|X\|$.
    Note that
    \[
    \Tr[X(t)\rho]-\Tr[X\sigma] = \braket{X(t)-I\braket{I,X}_{\mathrm{KMS}},\sigma^{-1/2}\rho\sigma^{-1/2}}_{\mathrm{KMS}}, 
    \]
    by the Cauchy-Schwarz inequality we have
    \[
    |\Tr[X(t)\rho]-\Tr[X\sigma]|\leq \|\sigma^{-1/2}\rho\sigma^{-1/2}\|_{\mathrm{KMS}}e^{-gt}\|X\|.
    \]
    The result follows by 
    \[
    \|\sigma^{-1/2}\rho\sigma^{-1/2}\|_{\mathrm{KMS}}^2 = \Tr[\rho\sigma^{-1/2}\rho\sigma^{-1/2}]\leq \|\sigma^{-1/2}\rho\sigma^{-1/2}\|\Tr[\rho]\leq \frac{1}{\sigma_{\min}}.
    \]
\end{proof}

\begin{cor}
\label{cor:convergence_of_state}
    We assume that $\mathcal{L}^\dag$ satisfies the KMS-DBC \eqref{eq:defn_KMS_DBC}, the spectral gap between 0 and the second largest eigenvalue of $\mathcal{L}^\dag|_{B(\mathcal{H})_{\mathrm{even}}}$ is $g>0$, and $\mathrm{ker}(\mathcal{L}^\dag|_{B(\mathcal{H})_{\mathrm{even}}})=\mathrm{span}(I)$.
    For an initial state $\rho\in B(\mathcal{H})_{\mathrm{even}}$ (i.e., satisfying the fermion superselection rule), we have
    \[
    \|\rho(t)-\sigma\|_{\mathrm{1}}\leq \frac{e^{-gt}}{\sqrt{\sigma_{\min}}},
    \]
    where $\rho(t)=e^{t\mathcal{L}}(\rho)$, $\sigma_{\min}$ is the minimum eigenvalue of $\sigma$, and $\|\cdot\|_1$ denotes the Schatten-1 norm.
\end{cor}

\begin{proof}
    For any observable $X\in B(\mathcal{H})$, we can decompose $X=X_{\mathrm{even}}+X_{\mathrm{odd}}$ where $X_{\mathrm{even}}$ and $X_{\mathrm{odd}}$ have even and odd parities as defined in Definition~\ref{defn:parity_of_operators}. We can readily see that $\Tr[\rho(t) X]=\Tr[\rho(t)X_{\mathrm{even}}]$ because $\rho\in B(\mathcal{H})_{\mathrm{even}}$, and the same is true for $\sigma$. Also, by the eigenvalue interlacing theorem, $\|X_{\mathrm{even}}\|\leq \|X\|$. Therefore by Corollary~\ref{cor:convergence_of_expect_val} we have
    \[
    |\Tr[(\rho(t)-\sigma)X]|\leq \frac{e^{-gt}}{\sqrt{\sigma_{\min}}}\|X_{\mathrm{even}}\|\leq \frac{e^{-gt}}{\sqrt{\sigma_{\min}}}\|X\|.
    \]
    Since this is true for any $X\in B(\mathcal{H})$, by the duality between the Schatten 1-norm and the Schatten $\infty$-norm we have the desired inequality.
\end{proof}

\section{Structure of the parent Hamiltonian}
\label{sec:structure_parent_ham}

We can write $\mathcal{L}^\dag$ in the super-operator form
\begin{equation}
    \label{eq:lindbladian_super_op_form}
    \mathcal{L}^\dag = i\sum_j (L_{B_j}-R_{B_j}) + \sum_j\int_{-\infty}^{\infty}\dd\omega\eta(\omega)\left(L_{A_j(\omega)^\dag}R_{A_j(\omega)}-\frac{1}{2}L_{A_j(\omega)^\dag A_j(\omega)}-\frac{1}{2}R_{A_j(\omega)^\dag A_j(\omega)}\right),
\end{equation}
This gives us a corresponding expression for the parent Hamiltonian through \eqref{eq:parent_hamiltonian_defn}:
\begin{equation}
    \label{eq:parent_hamiltonian_explicit}
    \begin{aligned}
        &H^{\mathrm{parent}} = \Phi\Bigg( i\sum_j(L_{\tilde{B}_j}-R_{\tilde{B}_j^\dag}) \\
        &+\sum_j\int\dd\omega \eta(\omega) \left( L_{\tilde{A}_j(-\omega)}R_{\tilde{A}_j(-\omega)^\dag}-\frac{1}{2}L_{\tilde{A}_j(-\omega) \tilde{A}_j(\omega)}-\frac{1}{2}R_{\tilde{A}_j(\omega)^\dag \tilde{A}_j(-\omega)^\dag}\right) \Bigg),
    \end{aligned}
\end{equation}
where
\begin{equation}
    \label{eq:A_tilde_and_B_tilde}
    \tilde{A}_j(\omega) = \sigma^{1/4}A_j(\omega)\sigma^{-1/4},\quad \tilde{B}_j=\sigma^{1/4}B_j\sigma^{-1/4},
\end{equation}
and we have used the fact that $A_j(\omega)^\dag=A_j(-\omega)$, which is a consequence of the hermitianity of $\gamma_j$.

We then separate the quadratic part from the rest. To do this we first define
\begin{equation}
    \label{eq:defn_gamma_free}
    \gamma^{\mathrm{free}}_j(t) = e^{iH_0 t}\gamma_j e^{-iH_0 t}.
\end{equation}
Then corresponding to $\tilde{A}_j(\omega)$ we define
\begin{equation}
    \label{eq:A_tilde_free}
    \tilde{A}_j^{\mathrm{free}}(\omega) = e^{-\beta H_0/4}\left(\frac{1}{\sqrt{2\pi}}\int  \gamma^{\mathrm{free}}_j(t) f(t) e^{-i\omega t} \dd t\right) e^{\beta H_0/4}.
\end{equation}
In other words $\tilde{A}^{\mathrm{free}}_j(\omega)$ is the jump operator we would have if we use $H_0$ as the Hamiltonian instead of $H$. Similarly, we define
\begin{equation}
    \label{eq:B_tilde_free}
    \tilde{B}^{\mathrm{free}}_j =  e^{-\beta H_0/4}\int_{-\infty}^\infty \dd t b_1(t) \int_{-\infty}^\infty \dd t' b_2(t') \gamma^{\mathrm{free}}_j(\beta (t'-t))  \gamma^{\mathrm{free}}_j(-\beta(t+t'))  e^{\beta H_0/4}.
\end{equation}
We now have the quadratic parts of both $\tilde{A}_j(\omega)$ and $\tilde{B}_j$, and this allows us to define the interacting parts:
\begin{equation}
\label{eq:A_tilde_B_tilde_int}
    \tilde{A}^{\mathrm{int}}_j(\omega) = \tilde{A}_j(\omega) - \tilde{A}^{\mathrm{free}}_j(\omega), \quad \tilde{B}^{\mathrm{int}}_j = \tilde{B}_j-\tilde{B}^{\mathrm{free}}_j.
\end{equation}

We will decompose the parent Hamiltonian into 2 parts, the coherent part and the dissipative part:
\begin{equation}
\label{eq:parent_ham_decompose_2_parts}
    H^{\mathrm{parent}} = H^{\mathrm{parent}}_{\mathrm{C}} + H^{\mathrm{parent}}_{\mathrm{D}}
\end{equation}
where 
\begin{equation}
\label{eq:parent_ham_coherent_part}
    H^{\mathrm{parent}}_{\mathrm{C}} = \Phi\Bigg( i\sum_j(L_{\tilde{B}_j}-R_{\tilde{B}_j^\dag})\Bigg),
\end{equation}
and
\begin{equation}
    \label{eq:parent_ham_dissipative_part}
    H^{\mathrm{parent}}_{\mathrm{D}} = \Phi\Bigg(\sum_j\int\dd\omega \eta(\omega) \left( L_{\tilde{A}_j(-\omega)}R_{\tilde{A}_j(-\omega)^\dag}-\frac{1}{2}L_{\tilde{A}_j(-\omega) \tilde{A}_j(\omega)}-\frac{1}{2}R_{\tilde{A}_j(\omega)^\dag \tilde{A}_j(-\omega)^\dag}\right) \Bigg)
\end{equation}
Each part is in turn decomposed into 2 parts, the free part and the interacting part:
\begin{equation}
    \label{eq:parent_ham_decompose_4_parts}
    H^{\mathrm{parent}}_{\mathrm{C}}=H^{\mathrm{parent}}_{\mathrm{C},\mathrm{free}} + H^{\mathrm{parent}}_{\mathrm{C},\mathrm{int}}, \quad H^{\mathrm{parent}}_{\mathrm{D}}=H^{\mathrm{parent}}_{\mathrm{D},\mathrm{free}}+ H^{\mathrm{parent}}_{\mathrm{D},\mathrm{int}},
\end{equation}
where
\begin{equation}
    \label{eq:parent_ham_coherent_free}
    H^{\mathrm{parent}}_{\mathrm{C},\mathrm{free}} = \Phi\Bigg( i\sum_j(L_{\tilde{B}^{\mathrm{free}}_j}-R_{(\tilde{B}^{\mathrm{free}}_j)^\dag})\Bigg),
\end{equation}
\begin{equation}
\label{eq:parent_ham_coherent_int}
    H^{\mathrm{parent}}_{\mathrm{C},\mathrm{int}} = H^{\mathrm{parent}}_{\mathrm{C}}-H^{\mathrm{parent}}_{\mathrm{C},\mathrm{free}},
\end{equation}
\begin{equation}
    \label{eq:parent_ham_dissipative_free}
    H^{\mathrm{parent}}_{\mathrm{D},\mathrm{free}} =\Phi\Bigg(\sum_j\int\dd\omega \eta(\omega) \left( L_{\tilde{A}^{\mathrm{free}}_j(-\omega)}R_{\tilde{A}^{\mathrm{free}}_j(-\omega)^\dag}-\frac{1}{2}L_{\tilde{A}^{\mathrm{free}}_j(-\omega) \tilde{A}^{\mathrm{free}}_j(\omega)}-\frac{1}{2}R_{\tilde{A}^{\mathrm{free}}_j(\omega)^\dag \tilde{A}^{\mathrm{free}}_j(-\omega)^\dag}\right) \Bigg),
\end{equation}
and
\begin{equation}
\label{eq:parent_ham_dissipative_int}
    H^{\mathrm{parent}}_{\mathrm{D},\mathrm{int}} = H^{\mathrm{parent}}_{\mathrm{D}}-H^{\mathrm{parent}}_{\mathrm{D},\mathrm{free}}.
\end{equation}
We then reassemble the parts to get
\begin{equation}
\label{eq:H_0_parent_and_V_defn}
    H_0^{\mathrm{parent}} = H^{\mathrm{parent}}_{\mathrm{C},\mathrm{free}} + H^{\mathrm{parent}}_{\mathrm{D},\mathrm{free}},\quad V^{\mathrm{parent}} = H^{\mathrm{parent}}_{\mathrm{C},\mathrm{int}} + H^{\mathrm{parent}}_{\mathrm{D},\mathrm{int}}.
\end{equation}
We can that $H_0^{\mathrm{parent}}$ is quadratic. Both $H_0^{\mathrm{parent}}$ and $H^{\mathrm{parent}}$ are Hermitian by Lemma~\ref{lem:parent_ham_hermitian}, which implies that $V^{\mathrm{parent}}$ is also Hermitian.
For $H_0^{\mathrm{parent}}$ we also have the following explicit expression:
\begin{equation}
    \label{eq:parent_hamiltonian_explicit_free}
    \begin{aligned}
        &H^{\mathrm{parent}}_{\mathrm{free}} = \Phi\Bigg( i\sum_j(L_{\tilde{B}^{\mathrm{free}}_j}-R_{(\tilde{B}^{\mathrm{free}}_j)^\dag}) \\
        &+\sum_j\int\dd\omega \eta(\omega) \left( L_{\tilde{A}^{\mathrm{free}}_j(-\omega)}R_{\tilde{A}^{\mathrm{free}}_j(-\omega)^\dag}-\frac{1}{2}L_{\tilde{A}^{\mathrm{free}}_j(-\omega) \tilde{A}^{\mathrm{free}}_j(\omega)}-\frac{1}{2}R_{\tilde{A}^{\mathrm{free}}_j(\omega)^\dag \tilde{A}^{\mathrm{free}}_j(-\omega)^\dag}\right) \Bigg).
    \end{aligned}
\end{equation}

\subsection{Quasi-locality and the spectral gap}
\label{quasi_locality_and_spectral_gap}

We will next prove the following: First, the non-interacting part is gapped.
\begin{prop}
    \label{prop:non_interacting_part_gap}
    For $H_0^{\mathrm{parent}}$ defined in \eqref{eq:H_0_parent_and_V_defn} corresponding to a $(1,r_0)$-geometrically local $H_0$, $0$ is its non-degenerate top eigenvalue, and is separated from the rest of its spectrum by a constant gap $\Delta>0$ that depends only on $\beta,r_0,D$.
\end{prop}
A proof is given at the end of Section~\ref{sec:spectral_gap_non_interacting}.
The non-interacting part also has nice locality:
\begin{prop}
    \label{prop:non_interacting_part_decay}
    For $H_0^{\mathrm{parent}}=\sum_{j,k=1}^{4n} m_{jk}\hat{\gamma}_j\hat{\gamma}_k$ defined in \eqref{eq:H_0_parent_and_V_defn} corresponding to a $(1,r_0)$-geometrically local $H_0$, its coefficient matrix $(m_{jk})$ has $[K,\nu]$-decay as defined in Definition~\ref{defn:decay_coef_matrix}. Here $K,\mu>0$ are constants that depend only on $\beta,r_0,D$.
\end{prop}
A proof is given at the end of Section~\ref{sec:quasi_locality_free_part}.
Next, we focus on the interacting part $V^{\mathrm{parent}}$, and show that it can be treated as quasi-local perturbation.
\begin{prop}
    \label{prop:interacting_part_decay}
    For $V^{\mathrm{parent}}$ defined defined in \eqref{eq:H_0_parent_and_V_defn} corresponding to a $(1,r_0)$-geometrically local $H_0$ and $(U,r_0)$-geometrically local $V$ for $U<U_0$, it has $(CU,\mu)$-decay as defined in Definition~\ref{defn:decay_operator}. Here $C,\mu>0$ are constants that depend only on $\beta,r_0,D,U_0$.
\end{prop}
A proof is provided at the end of Section~\ref{sec:coherent_term}.
Combining these three propositions, we then have the following result about the spectrum of the parent Hamiltonian
\begin{thm}
    \label{thm:spectrum_parent_ham}
    For any $\beta,r_0>0$, and integer $D>0$, there exists $U_\beta>0$ that depends only on $\beta,r_0,D$ such that when $H_0$ and $V$ in \eqref{eq:general_hamiltonian} are $(1,r_0)$-geometrically local and $(U,r_0)$-geometrically local respectively, for $U<U_\beta$, then the parent Hamiltonian $H^{\mathrm{parent}}$ given in \eqref{eq:parent_hamiltonian_defn} has $0$ as its non-degenerate top eigenvalue, and the spectral gap separating $0$ from the rest of its spectrum is lower bounded by a constant $\Delta>0$ that depends only on $\beta,r_0,D>0$.
\end{thm}

\begin{proof}
    In stead of directly considering $H=H_0+V$ as in \eqref{eq:general_hamiltonian}, we consider $H(u) = H_0 + uV$, $u\in[0,1]$. We then write $H^{\mathrm{parent}}$ and $V^{\mathrm{parent}}$ in \eqref{eq:H_0_parent_and_V_defn} as functions of $u$.
    
    First, Lemma~\ref{lem:parent_ham_hermitian} ensures that $H_0^{\mathrm{parent}}$ in \eqref{eq:H_0_parent_and_V_defn} and $V^{\mathrm{parent}}(u)$ are both Hermitian.
    By Proposition~\ref{prop:non_interacting_part_gap}, we know that the spectral gap of $H_0^{\mathrm{parent}}$ is lower bounded by a constant $\Delta_0$. By Proposition~\ref{prop:non_interacting_part_decay}, we know that the coefficient matrix of $H_0^{\mathrm{parent}}$ as a quadratic fermionic Hamiltonian has $[K,\nu]$-decay. By Proposition~\ref{prop:interacting_part_decay}, setting $U_0=1$, we know that $V^{\mathrm{parent}}(u)$ has $(CU u,\mu)$-decay for $U<1$. Combining all these results, by Theorem~\ref{thm:stability_free_fermion}, we can then see that there exists $U'_\beta>0$, such that when $U\leq U'_\beta$, the spectral gap between the top eigenvalue and the second largest eigenvalue of $H^{\mathrm{parent}}(u)=H_0^{\mathrm{parent}}+V^{\mathrm{parent}}(u)$ is lower bounded by $\Delta_0-c_1C U u $, for any $u\in[0,1]$.

    We will next assume $U<\min\{U'_\beta,\Delta_0/(c_1 C)\}$, and show that the top eigenvalue of $H^{\mathrm{parent}}(u)$ is $0$ and non-degenerate for $u\in [0,1]$. First we observe that $0$ is always an eigenvalue of $H^{\mathrm{parent}}(u)$ for $u\in [0,1]$ (corresponding to $\varphi(\sigma(u)^{1/2})$). If it becomes degenerate, then because the eigenvalues can only change continuously as a function of $u$ by Weyl's inequality, there must exist $u_\epsilon\in[0,1]$ such that the spectral gap is smaller than $\epsilon$ for any $\epsilon>0$, which contradicts the gap lower bound $\Delta_0-c_1C U u \geq \Delta_0-c_1C U >0$. This shows that $0$ must be a non-degenerate eigenvalue. If there is another eigenvalue of $H^{\mathrm{parent}}(u)$ that is positive for any $u\in[0,1]$, then because $0$ is the non-degenerate top eigenvalue of $H^{\mathrm{parent}}(0)=H^{\mathrm{parent}}_0$ by Proposition~\ref{prop:non_interacting_part_gap}, there must exists $u^*\in[0,1]$ such that $0$ is a degenerate eigenvalue of $H^{\mathrm{parent}}(u^*)$, which contradicts the non-degeneracy of $0$ that we just proved. Therefore $0$ is the non-degenerate top eigenvalue of $H^{\mathrm{parent}}(u)$ for all $u\in[0,1]$.

    Letting $u=1$, $U_\beta=0.99\min\{U'_\beta,\Delta_0/(c_1 C)\}$, we then conclude that when $U\leq U_\beta$, the spectral gap of $H^{\mathrm{parent}}$ defined in \eqref{eq:parent_hamiltonian_defn} is lower bounded by $\Delta=\Delta_0-c_1C U>0$, and that $0$ is the non-degenerate top eigenvalue of $H^{\mathrm{parent}}$. In the above $C,K,U'_\beta,\nu,\mu,c_1,\Delta_0>0$ are all constants that depend only on $\beta,D,r_0$. This implies that $U_\beta,\Delta>0$ are also constants that only depend on $\beta,D,r_0$.
\end{proof}

\subsection{Fast mixing}
\label{sec:fast_mixing}

The mixing time characterizes the maximal time for a quantum state to converge to the stationary state in time evolution. 
In the fermionic setting, however, the superselection rule places a constraint on what quantum states should be considered. Since no superposition between even- and odd-parity sectors are allowed, we only need to consider quantum states that can be expressed as linear combinations of Majorana operators of even parity. They make up the set
\begin{equation}
    \label{eq:quantum_states_satisfying_superselection_rule}
    D(\mathcal{H})_{\mathrm{even}} = D(\mathcal{H}) \cap B(\mathcal{H})_{\mathrm{even}},
\end{equation}
$D(\mathcal{H})$ is the set of all positive semi-definite and trace-1 operators defined in \eqref{eq:set_of_quantum_states}.
This requires us to modify the definition of the mixing time in the following way:
\begin{defn}[Mixing time for fermions]
\label{defn:mixing_time}
    For a Lindbladian operator $\mathcal{L}$ with unique stationary state in $D(\mathcal{H})_{\mathrm{even}}$, we define its $\epsilon$-mixing time to be
    \begin{equation}
        t_{\mathrm{mix}}(\epsilon) = \inf\{t \geq 0:\|e^{s \mathcal{L}}(\rho)-\sigma\|_1\leq \epsilon,\forall s\geq t, \rho\in D(\mathcal{H})_{\mathrm{even}}\}.
    \end{equation}
    Here $\sigma$ is the stationary state of $\mathcal{L}$.
\end{defn}
We are then ready to prove Theorem~\ref{thm:main_result}, which we restate here: 
\begin{thm*}[Main result]
    Let $H=H_0+V$ be a Hamiltonian on a $D$-dimensional cubic lattice $\Lambda$ consisting of $n$ sites. $H_0$ is quadratic in Majorana operators and is $(1,r_0)$-geometrically local, and $V$ is parity-preserving and is $(U,r_0)$-geometrically local. Then for any inverse temperature $\beta>0$, there exists $U_\beta>0$ such that when $U<U_\beta$, the Lindbladian defined in \eqref{eq:QMC_lindbladian} has a unique stationary state satisfying the fermionic superselection rule, and the $\epsilon$-mixing time (Definition~\ref{defn:mixing_time}) of the Lindbladian is at most
    \[
    C(n+\log(1/\epsilon)).
    \]
    where $C,U_\beta>0$ are constants that depend only on $r_0,D,\beta$.
\end{thm*}

\begin{proof}
    By Theorem~\ref{thm:spectrum_parent_ham}, we know that there exists $U_\beta>0$ such that when $U<U_\beta$, the spectral gap of $H^{\mathrm{parent}}$ is lower bounded by $\Delta>0$, and that $0$ is a non-degenerate eigenvalue. This implies that $0$ is also a non-degenerate eigenvalue of $H^{\mathrm{parent}}|_{\mathcal{H}'_{\mathrm{even}}}$ (since $\varphi(\sigma^{1/2})$, which spans the kernel, is in $\mathcal{H}'_{\mathrm{even}}$), and that the spectral gap of $H^{\mathrm{parent}}|_{\mathcal{H}'_{\mathrm{even}}}$ is also at least $\Delta$. By \eqref{eq:similarity_structure_Phi_restricted_to_even} and the definition of the parent Hamiltonian in \eqref{eq:parent_hamiltonian_defn}, we can then see that $\mathcal{L}^\dag|_{B(\mathcal{H})_{\mathrm{even}}}$ has spectral gap at least $\Delta$, and $0$ is its non-degenerate eigenvalue, implying that $\ker(\mathcal{L}^\dag|_{B(\mathcal{H})_{\mathrm{even}}})=\mathrm{span}(I)$.
    This shows that $\sigma$ is the unique fixed point of $\mathcal{L}$ that satisfies the fermionic superselection rule.
    Therefore all conditions of Corollary~\ref{cor:convergence_of_state} are satisfied, and it tells us that
    \[
    \|e^{t\mathcal{L}}(\rho)-\sigma\|_1\leq \frac{e^{-\Delta t}}{\sqrt{\sigma_{\min}}},
    \]
    for all $\rho\in D(\mathcal{H})_{\mathrm{even}}$. The $\epsilon$-mixing time is then upper bounded by
    \[
    t_{\mathrm{mix}}(\epsilon)\leq \frac{1}{\Delta}\left(\frac{1}{2}\log(1/\sigma_{\min})+\log(1/\epsilon)\right).
    \]
    Note that 
    \[
    \sigma_{\min} = \frac{1}{Z}e^{-\beta \|H\|} \geq \frac{1}{2^n e^{\beta\|H\|}}e^{-\beta \|H\|}=\frac{1}{2^n}e^{-2\beta \|H\|},
    \]
    and $\|H\|\leq C'n$ for some $C'>0$ that only depends on $D,r_0$, we then have $\log(1/\sigma_{\min})\leq n\log(2)+2\beta C'n$, which yields the mixing time bound in the theorem. In the above $\Delta,U_\beta>0$ are both constants that depend only on $\beta,r_0,D$. 
\end{proof}

\subsection{Recovering the Gibbs state from the top eigenstate}
\label{sec:recovering_gibbs_state_from_top_eigenstate}

In this section we will discuss recovering information about the Gibbs state from the top eigenstate of the parent Hamiltonian. Suppose we want to obtain the expectation value of an observable 
\begin{equation}
\label{eq:arbitrary_observable}
    X=\sum_{\alpha\in\{0,1\}^{2n}}X_{\alpha}\gamma_1^{\alpha_1}\gamma_2^{\alpha_2}\cdots\gamma_{2n}^{\alpha_{2n}}\in B(\mathcal{H}),
\end{equation}
then we will show that we can get it from an observable
\begin{equation}
\label{eq:observable_enlarged_hilbert_space}
    X^{\sharp}=\Phi(L_{X})=\sum_{\alpha\in\{0,1\}^{2n}}X_{\alpha}\gamma_1^{\alpha_1}\gamma_3^{\alpha_3}\cdots\gamma_{4n-1}^{\alpha_{4n-1}}\in B(\mathcal{H}').
\end{equation}
\begin{prop}
    \label{prop:correspondence_top_eigenstate_gibbs_state}
    Let $\sigma$ be the Gibbs state of $\mathcal{L}$ in \eqref{eq:QMC_lindbladian} and let $v$ be the (normalized) top eigenstate of the parent Hamiltonian $H^{\mathrm{parent}}$ defined in \eqref{eq:parent_hamiltonian_defn}. Then for any observable $X\in B(\mathcal{H})$ as expressed in \eqref{eq:arbitrary_observable}, we have a corresponding $X^{\sharp}\in B(\mathcal{H}')$ as defined in \eqref{eq:observable_enlarged_hilbert_space} such that $\Tr[\sigma X]=\braket{v, X^{\sharp}v}$.
\end{prop}

\begin{proof}
    By Theorem~\ref{thm:spectrum_parent_ham} we know that $v\propto \varphi(\sigma^{1/2})$. We can compute that, by \eqref{eq:third_quantization_preserves_frobenius_norm},
    \[
    \braket{\varphi(\sigma^{1/2}),\varphi(\sigma^{1/2})}=\frac{1}{2^n}\Tr[\sigma]=\frac{1}{2^n}.
    \]
    Therefore we have $v=2^{n/2}\varphi(\sigma^{1/2})$ (fixing an unimportant global phase). Then we have
    \[
    \braket{v, X^{\sharp}v} = \braket{v, \varphi\circ L_X \circ \varphi^{-1}(v)} = 2^n\braket{\varphi(\sigma^{1/2}),\varphi(X\sigma^{1/2})} = \Tr[\sigma^{1/2}X\sigma^{1/2}]=\Tr[\sigma X].
    \]
\end{proof}

\begin{rem}
    \label{rem:observable_mapping_homomorphism_preserve_norm}
    This mapping  $X\in B(\mathcal{H})\mapsto X^{\sharp}\in B(\mathcal{H}')$ is an algebra homomorphism. One can easily check that for $X,Y\in B(\mathcal{H})$ we have $(XY)^\sharp = X^\sharp Y^\sharp$. Moreover, we have $\|X^\sharp\|=\|X\|$ by Lemma~\ref{lem:norm_preservation}.
\end{rem}

As a result of this correspondence we can readily prove correlation decay in the Gibbs state, as stated in Corollary~\ref{cor:exp_decay_correlation}, which we restate here:
\begin{cor*}
    Under the same assumptions as Theorem~\ref{thm:main_result}, there exists $C,\xi>0$ that depends only only on $\beta,r_0,D$, such that for an observable $X\in B(\mathcal{H})$ supported on a constant number of adjacent sites, and any observable $Y\in B(\mathcal{H})$ with support $S$, 
    \[
    |\Tr[\sigma X Y]-\Tr[\sigma X]\Tr[\sigma Y]|\leq C\|X\|\|Y\||S|e^{-d/\xi},
    \]
    where $d$ is the distance between the supports of $X$ and $Y$.
\end{cor*}

\begin{proof}[Proof of Corollary~\ref{cor:exp_decay_correlation}]
    Let $v$ be the normalized top eigenstate of $H^{\mathrm{parent}}$. By Proposition~\ref{prop:correspondence_top_eigenstate_gibbs_state} and Remark~\ref{rem:observable_mapping_homomorphism_preserve_norm}, we have 
    \[
    \Tr[\sigma XY] = \braket{v, X^\sharp Y^\sharp v},\quad \Tr[\sigma X]=\braket{v, X^\sharp v},\quad \Tr[\sigma Y]=\braket{v, Y^\sharp v}.
    \]
    The exponential clustering theorem for gapped ground states (as stated in \cite[Theorem~2]{nachtergaele2006lieb} and also proved in \cite{Hastings2006spectral}), tells us that
    \[
    |\braket{v, X^\sharp Y^\sharp v}-\braket{v, X^\sharp v}\braket{v, Y^\sharp v}|\leq C\|X\|\|Y\||S|e^{-d/\xi}.
    \]
    We therefore have the desired inequality stated in the corollary.
\end{proof}


\section{Spectral gap of the non-interacting part}
\label{sec:spectral_gap_non_interacting}

We will first show that the non-interacting, coherent part of the parent Hamiltonian is zero.
\begin{lem}
\label{lem:free_parent_zero} 
$H^{\mathrm{parent}}_{\mathrm{C},\mathrm{free}}$ as defined in \eqref{eq:parent_ham_coherent_free} is $0$.
\end{lem}
\begin{proof}
Since 
\[
H^{\mathrm{parent}}_{\mathrm{C},\mathrm{free}} = \Phi\Bigg( i\sum_j(L_{\tilde{B}^{\mathrm{free}}_j}-R_{(\tilde{B}^{\mathrm{free}}_j)^\dag})\Bigg),
\]
it suffices to show that $\sum_j\tilde{B}^{\mathrm{free}}_j=\sum_j\sigma^{1/4}{B}^{\mathrm{free}}_j\sigma^{-1/2}=0$, which is equivalent to showing $B^{\mathrm{free}}=\sum_j B^{\mathrm{free}}_j=0$. For $B^{\mathrm{free}}$ we have
\begin{equation*}
\begin{aligned}
B^{\mathrm{free}}&= \sum_{k} \int_{-\infty}^{\infty} b_1(t) e^{-i\beta H_0 t} \left( \int_{-\infty}^{\infty} b_2(t') e^{i\beta H_0 t'} \gamma_k e^{-2i\beta H_0 t'} \gamma_k e^{i\beta H_0 t'} dt' \right) e^{i\beta H_0 t} dt\\
&=\sum_{k} \int_{-\infty}^{\infty} dt \int_{-\infty}^{\infty} dt' b_1(t) b_2(t') e^{i\beta H_0 (t'-t)} \gamma_k e^{-i\beta H_0 (t'-t)} e^{-i\beta H_0 (t'+t)} \gamma_k e^{i\beta H_0 (t'+t)}.
\end{aligned}
\end{equation*}
Using Lemma~\ref{lem:fermion_basis_transform} we then have
\begin{equation*}
    \begin{aligned}
        B^{\mathrm{free}}&=\sum_{k} \int_{-\infty}^{\infty} dt \int_{-\infty}^{\infty} dt' b_1(t) b_2(t') \sum_p \gamma_p \left(e^{2i\beta h (t'-t)}\right)_{p,k} \sum_q \gamma_q \left(e^{-2i\beta h (t'+t)}\right)_{q,k}\\
        &=\sum_{k} \int_{-\infty}^{\infty} dt \int_{-\infty}^{\infty} dt' b_1(t) b_2(t') \sum_p \gamma_p \left(e^{2i\beta h (t'-t)}\right)_{p,k} \sum_q \gamma_q \left(e^{2i\beta h (t'+t)}\right)_{k,q}\\
        &=\int_{-\infty}^{\infty} b_1(t) dt \int_{-\infty}^{\infty} dt'  b_2(t') \sum_p \sum_q  \left(e^{4i\beta h t'}\right)_{p,q}  \gamma_p \gamma_q \\
        &=C \sum_{p,q}\check{b}_2(4i\beta h)_{p,q} \gamma_p \gamma_q
    \end{aligned}
\end{equation*}
where $C=\int_{-\infty}^{\infty} b_1(t) dt$ and $\check{b}_2(x)=\int_{-\infty}^{\infty} b_2(t') e^{i x t'}dt'$. Then we compute the constant $C$.
\begin{equation}
\begin{aligned}
C&=\int_{-\infty}^{\infty} 2\sqrt{\pi} e^{1/8} \left( \frac{1}{\cosh(2\pi t)} \right) * \sin(-t) \exp(-2t^2) dt\\
&=\int_{-\infty}^{\infty}dt \int_{-\infty}^{\infty} ds 2\sqrt{\pi} e^{1/8} \left( \frac{1}{\cosh(2\pi s)} \right) \sin(-(t-s)) \exp(-2(t-s)^2) \\
&=0,
\end{aligned}
\end{equation}
where we have used the fact that $\int_{-\infty}^{\infty} \sin(-(t-s)) \exp(-2(t-s)^2)dt=0$. This then implies $B^{\mathrm{free}}=0$. 
\end{proof}

This allow us to focus on the dissipative part in $H^{\mathrm{parent}}_{\mathrm{D},\mathrm{free}}$ \eqref{eq:parent_ham_dissipative_free}. We will first analyze it for the simplest case, supposing that the Hamiltonian acts only on one fermionic mode.
\begin{lem}
    \label{lem:single_mode_parent_ham}
    We consider the free parent Hamiltonian $H^{\mathrm{parent}}_{\mathrm{free}}$ defined in \eqref{eq:parent_hamiltonian_explicit_free} generated by 
    \[
    H_0 = i\epsilon(\zeta_1\zeta_2-\zeta_2\zeta_1),
    \]
    where $\zeta_1$ and $\zeta_2$ are two Majorana operators satisfying $\{\zeta_j,\zeta_k\}=2\delta_{jk}$. Then defining the third quantized $a$-Majorana modes $\hat{\zeta}_j$, $j=1,2,3,4$, following \eqref{eq:third_quantization_majorana_ops}, we have
    \begin{equation}
    \label{eq:single_mode_parent_ham}
        H^{\mathrm{parent}}_{\mathrm{free}}=Ce^{-4\beta^2\epsilon^2}\left(-i\hat{d}_1^\dag\hat{d}_2+i\hat{d}_2^\dag\hat{d}_1-\sinh(2\beta\epsilon)\hat{d}_1^\dag\hat{d}_1 + \sinh(2\beta\epsilon)\hat{d}_2^\dag\hat{d}_2-\cosh(2\beta\epsilon)\right),
    \end{equation}
    where $\hat{d}_1=(\hat{\zeta}_1+i\hat{\zeta}_3)/2$, $\hat{d}_2=(\hat{\zeta}_2+i\hat{\zeta}_4)/2$, and $C$ is a universal constant.
\end{lem}

\begin{proof}
    Since $H_0$ is non-interacting and only involves one fermionic mode, we can directly compute $\zeta_1^{\mathrm{free}}(t)$ and $\zeta_2^{\mathrm{free}}(t)$ corresponding to \eqref{eq:defn_gamma_free}. This then allows us to explicitly compute $\tilde{A}_j^{\mathrm{free}}(\omega)$ in \eqref{eq:A_tilde_free}. The coherent part can be neglected since we know it is 0 from Lemma~\ref{lem:free_parent_zero}. The explicit expression for $H^{\mathrm{parent}}_{\mathrm{free}}=H^{\mathrm{parent}}_{\mathrm{D,free}}$ then follows from straightforward calculation.
\end{proof}

\begin{cor}
    \label{cor:gap_single_mode}
    Under the same setup as in Lemma~\ref{lem:single_mode_parent_ham}, the top eigenvalue of $H^{\mathrm{parent}}_{\mathrm{free}}$ is $0$, and the spectral gap separating it from the second largest eigenvalue is  $Ce^{-4\beta^2\epsilon^2}\cosh(2\beta\epsilon)$, where $C$ is the same universal constant as in Lemma~\ref{lem:single_mode_parent_ham}.
\end{cor}

\begin{proof}
    We note that in \eqref{eq:single_mode_parent_ham}, the part involving fermionic creation and annihilation operators 
    \[
    -i\hat{d}_1^\dag\hat{d}_2+i\hat{d}_2^\dag\hat{d}_1-\sinh(2\beta\epsilon)\hat{d}_1^\dag\hat{d}_1 + \sinh(2\beta\epsilon)\hat{d}_2^\dag\hat{d}_2
    \]
    has single-particle eigenenergies $\pm \cosh(2\beta\epsilon)$. Therefore the top eigenvalue of $H^{\mathrm{parent}}_{\mathrm{free}}$ is
    \[
    Ce^{-4\beta^2\epsilon^2}(\cosh(2\beta\epsilon)-\cosh(2\beta\epsilon))=0,
    \]
    and the second largest eigenvalue of it is
    \[
    Ce^{-4\beta^2\epsilon^2}(0-\cosh(2\beta\epsilon))=-Ce^{-4\beta^2\epsilon^2}\cosh(2\beta\epsilon).
    \]
    This then gives us the expression for the spectral gap.
\end{proof}

Next, we will deal with the more general case where the system consists of more than one fermionic modes. However, we will see that the general case can be reduced to this simple case.
We first rotate the Majorana operators $\gamma_j$ to turn the non-interacting part $H_0=\sum_{jk}h_{jk}\gamma_j\gamma_k$ into a standard form. We choose a $2n\times 2n$ real orthogonal matrix $Q=(Q_{jk})$ and perform the basis transform
\begin{equation}
    \label{eq:majorana_basis_transform}
    \gamma_j = \sum_{k} Q_{jk} \zeta_k.
\end{equation}
It is easy to see that $\zeta_k=\sum_j Q_{jk}\gamma_j$, and that $\braket{\zeta_k,\zeta_{k'}}=2\delta_{kk'}$ and therefore $\{\zeta_k\}$ for a valid set of Majorana operators. The orthogonal matrix $Q$ is chosen so that
\begin{equation}
\label{eq:free_canonical_form}
    H_0 = i\sum_{j=1}^n \epsilon_j \zeta_{2j-1}\zeta_{2j},
\end{equation}
for $\epsilon_j\geq 0$.
Such a $Q$ always exist because the coefficient matrix $(h_{jk})$ is Hermitian and pure imaginary.
We can see that the free Hamiltonian gets decoupled into terms with non-overlapping supports.
Here $\pm \epsilon_j$ are eigenvalues of the coefficient matrix $(h_{jk})$. Because $H_0$ is $(1,r_0)$-geometrically local as defined in Definition~\ref{defn:geometrically_local}, there are at most $Cr_0^D$ non-zero elements on each row of $(h_{jk})$, and therefore
\begin{equation}
    \label{eq:eigenvalue_H0_bound}
    |\epsilon_j|\leq \|H\|\leq Cr_0^D,
\end{equation}
where the constant $C>0$ depends only on $r_0$.

We will show that this basis transform decouples the free parent Hamiltonian as well. Essentially, this is true because the free parent Hamiltonian is a quadratic function of $\gamma_j$.

Recall that the free parent Hamiltonian $H^{\mathrm{parent}}_{\mathrm{free}}$ is of the form \eqref{eq:parent_hamiltonian_explicit_free}, and because $H^{\mathrm{parent}}_{\mathrm{C,free}}=0$, we have
\begin{equation}
    \begin{aligned}
        &H^{\mathrm{parent}}_{\mathrm{free}} = \Phi\Bigg( \sum_j\int\dd\omega \eta(\omega) \left( L_{\tilde{A}^{\mathrm{free}}_j(-\omega)}R_{\tilde{A}^{\mathrm{free}}_j(-\omega)^\dag}-\frac{1}{2}L_{\tilde{A}^{\mathrm{free}}_j(-\omega) \tilde{A}^{\mathrm{free}}_j(\omega)}-\frac{1}{2}R_{\tilde{A}^{\mathrm{free}}_j(\omega)^\dag \tilde{A}^{\mathrm{free}}_j(-\omega)^\dag}\right) \Bigg),
    \end{aligned}
\end{equation}
We define
\begin{equation}
    \label{eq:rotated_majorana_heisenberg}
    \zeta^{\mathrm{free}}_j(t) = e^{iH_0 t}\zeta_j e^{-iH_0 t}.
\end{equation}
We then have
\begin{equation}
    \gamma^{\mathrm{free}}_j(t) = \sum_k Q_{jk}\zeta^{\mathrm{free}}_k(t).
\end{equation}
With this, corresponding to $\tilde{A}^{\mathrm{free}}_j(\omega)$ in \eqref{eq:A_tilde_free} we define 
\begin{equation}
    \label{eq:xi_j}
    \xi_j(\omega) = e^{-\beta H_0/4}\left(\frac{1}{\sqrt{2\pi}}\int  \zeta^{\mathrm{free}}_j(t)  f(t) \mathrm{e}^{-i\omega t} \dd t\right) e^{\beta H_0/4},
\end{equation}
which is related to $\tilde{A}^{\mathrm{free}}_j(\omega)$ through
\begin{equation}
\label{eq:transform_A_tilde_j_free}
    \tilde{A}^{\mathrm{free}}_j(\omega) = \sum_{k} Q_{jk} \xi_k(\omega),
\end{equation}
Now we can consider how these transformations affect terms in \eqref{eq:parent_hamiltonian_explicit_free}. We have
\[
\sum_j L_{\tilde{A}^{\mathrm{free}}_j(-\omega)}R_{\tilde{A}^{\mathrm{free}}_j(-\omega)^\dag} = \sum_j \sum_{k,k'} Q_{jk} Q_{jk'} L_{\xi_k(-\omega)}R_{\xi_{k'}(-\omega)^\dag} = \sum_k L_{\xi_k(-\omega)}R_{\xi_{k}(-\omega)^\dag},
\]
where we have used $\sum_j Q_{jk}Q_{jk'}=\delta_{kk'}$.
Similarly
\[
\sum_j L_{\tilde{A}^{\mathrm{free}}_j(-\omega)\tilde{A}^{\mathrm{free}}_j(\omega)} = \sum_k L_{\xi_k(-\omega)\xi_k(\omega)},\quad \sum_j R_{\tilde{A}^{\mathrm{free}}_j(\omega)^\dag\tilde{A}^{\mathrm{free}}_j(-\omega)^\dag} = \sum_k R_{\xi_k(\omega)^\dag\xi_k(-\omega)^\dag}.
\]
From the above we can see that
\begin{equation}
    H^{\mathrm{parent}}_{\mathrm{free}} = \sum_{k=1}^{n} H^{\mathrm{free}}_k,
\end{equation}
where
\begin{equation}
    H^{\mathrm{free}}_k = \Phi\left(L_{\Xi_{kk}}+R_{\Xi_{kk}^\dag} + \int\dd\omega\eta(\omega) \left(L_{\xi_k(-\omega)}R_{\xi_{k}(-\omega)^\dag}-\frac{1}{2}L_{\xi_k(-\omega)\xi_k(\omega)}-\frac{1}{2}R_{\xi_k(\omega)^\dag\xi_k(-\omega)^\dag}\right)\right),
\end{equation}
only involves Majorana modes $\zeta_{2k-1}$ and $\zeta_{2k}$.

In the third quantization, we define
\begin{equation}
    \label{eq:rotated_third_quantized_majorana_modes}
    \hat{\zeta}_{2j-1} = \sum_k Q_{jk} \hat{\gamma}_{2k-1},\quad \hat{\zeta}_{2j} = \sum_k Q_{jk} \hat{\gamma}_{2k}.
\end{equation}
Then by \eqref{eq:repr_left_right_multiplication_more_majorana_ops} and linearity, we have
\[
    \begin{aligned}
        \varphi\circ L_{\zeta_{j_1}\cdots \zeta_{j_k}}\circ \varphi^{-1} &= \hat{\zeta}_{2j_1-1}\cdots \hat{\zeta}_{2j_k-1},\\
        \varphi\circ R_{\zeta_{j_1}\cdots \zeta_{j_k}}\circ \varphi^{-1} &= -i^k \hat{\zeta}_{2j_k}\cdots\hat{\zeta}_{2j_1} (-1)^{k\hat{n}}.
    \end{aligned}
\]
As a result Remark~\ref{rem:Phi_preserves_locality} corresponds to
\[
\Phi(L_{\zeta_{j_1}\cdots \zeta_{j_k}}) = \hat{\zeta}_{2j_1-1}\cdots \hat{\zeta}_{2j_k-1},\quad \Phi(R_{\zeta_{j_1}\cdots \zeta_{j_k}})=i^k \hat{\zeta}_{2j_k}\cdots\hat{\zeta}_{2j_1}.
\]
$H_k^{\mathrm{free}}$ is therefore supported only on $a$-Majorana modes $\hat{\gamma}_{4k-3}$, $\hat{\gamma}_{4k-2}$, $\hat{\gamma}_{4k-1}$, $\hat{\gamma}_{4k}$. Moreover, $H_k^{\mathrm{free}}$ is the parent Hamiltonian one get through \eqref{eq:parent_hamiltonian_explicit_free} by replacing $H_0$ with $i\epsilon_j(\zeta_1\zeta_2-\zeta_2\zeta_1)$. Therefore by Lemma~\ref{lem:single_mode_parent_ham} we have
\begin{equation}
\label{eq:H_free_k_explicit}
\begin{aligned}
    H_k^{\mathrm{free}} = Ce^{-4\beta^2\epsilon_k^2}&\big(-i\hat{d}_{2k-1}^\dag\hat{d}_{2k}+i\hat{d}_{2k}^\dag\hat{d}_{2k-1} \\
    &-\sinh(2\beta\epsilon_k)\hat{d}_{2k-1}^\dag\hat{d}_{2k-1} + \sinh(2\beta\epsilon_k)\hat{d}_{2k}^\dag\hat{d}_{2k}-\cosh(2\beta\epsilon_k)\big),
\end{aligned}
\end{equation}
where $\hat{d}_{2k-1}=(\hat{\zeta}_{4k-3}+i\hat{\zeta}_{4k-1})/2$, $\hat{d}_{2k}=(\hat{\zeta}_{4k-2}+i\hat{\zeta}_{4k})/2$, $\hat{\zeta}_j$, $j=1,2,\cdots,4n$, are third quantized $a$-Majorana operators defined through \eqref{eq:third_quantization_majorana_ops}, and $C$ is the same universal constant as in Lemma~\ref{lem:single_mode_parent_ham}.

We therefore have the following lemma:
\begin{lem}
    \label{lem:decouling_free_part}
    Let Majorana operators $\zeta_k$, $k=1,2,\cdots,2n$, be as defined in \eqref{eq:majorana_basis_transform} so that \eqref{eq:free_canonical_form} is satisfied. Then the parent Hamiltonian of the free part can be written as
    \[
    H^{\mathrm{parent}}_{\mathrm{free}} = \sum_k H^{\mathrm{free}}_k,
    \]
    where each $H^{\mathrm{free}}_k$ is given in \eqref{eq:H_free_k_explicit}.
\end{lem}

We are then ready to prove Proposition~\ref{prop:non_interacting_part_gap}.
\begin{proof}[Proof of Proposition~\ref{prop:non_interacting_part_gap}]
    Because $H^{\mathrm{parent}}_{\mathrm{free}} = \sum_k H^{\mathrm{free}}_k$ and each $H^{\mathrm{free}}_k$ are supported on a non-overlapping set of Majorana operators in the $\{\hat{\zeta}_j\}$ basis, the top eigenvalue of $H^{\mathrm{parent}}_{\mathrm{free}}$ is then the sum of the top eigenvalues of each $H^{\mathrm{free}}_k$, which is $0$. 
    The spectral gap $g$ of $H^{\mathrm{parent}}_{\mathrm{free}}$ is the minimum among the spectral gaps of each $H^{\mathrm{free}}_k$, and by Corollary~\ref{cor:gap_single_mode},
    \[
    g = \min_k Ce^{-4\beta^2\epsilon_k^2}\cosh(2\beta\epsilon_k),
    \]
    which because of \eqref{eq:eigenvalue_H0_bound} is lower bounded by a constant $\Delta>0$ that depends only on $\beta,D,r_0$.
\end{proof}

\section{Structure of the dissipative part}
\label{sec:dissipative}

In this section we will show that both $\tilde{A}^{\mathrm{free}}_j(\omega)$ and 
$\tilde{A}^{\mathrm{int}}_j(\omega)$ defined in \eqref{eq:A_tilde_free} and \eqref{eq:A_tilde_B_tilde_int} are quasi-local, and that the norm of $\tilde{A}^{\mathrm{int}}_j(\omega)$ is small (on the order $U$). To do this we will first express the jump operator $\tilde{A}^{\mathrm{free}}_j(\omega)$ solely in terms of $e^{iHt}\gamma_j e^{-iHt}$ and avoid having imaginary time evolution $e^{-\beta H}$ showing up in the expression.

For an operator $A\in B(\mathcal{H})$, we introduce the following notation: let $\{E_i\}$ be the eigenvalues of $H$, and $\{\Pi_i\}$ be the corresponding eigenspace projection operators, then we denote
\begin{equation}
    A_\nu = \sum_{ij:E_i-E_j=\nu} \Pi_i A \Pi_j.
\end{equation}
There is only a finite set of $\nu$ for which $A_\nu\neq 0$, they are called \emph{Bohr frequencies}.

Using the decomposition $\gamma_j = \sum_{\nu\in B_H} (\gamma_j)_\nu$, we have
\begin{equation}
    \tilde{A}_j(\omega) = \frac{1}{\sqrt{2\pi}}\sum_{\nu} \int e^{-\beta\nu/4}e^{i\nu t}(\gamma_j)_\nu e^{-i\omega t}f(t)\dd t.
\end{equation}
Now we define
\begin{equation}
\label{eq:F1_defn}
    F_1(\nu,\omega) = \frac{1}{\sqrt{2\pi}}\int e^{-\beta \nu/4}e^{i(\nu-\omega)t}f(t)\dd t,
\end{equation}
then we have
\begin{equation}
\label{eq:A_tilde_LC_time_evolution}
\begin{aligned}
    \tilde{A}_j(\omega) &= \sum_{\nu} F_1(\nu,\omega) (\gamma_j)_\nu \\
    &= \frac{1}{\sqrt{2\pi}}\sum_{\nu} \int \check{F}_1(t,\omega)e^{i\nu t}(\gamma_j)_\nu \dd t \\ 
    &= \frac{1}{\sqrt{2\pi}}\int \check{F}_1(t,\omega) e^{iHt}\gamma_j e^{-iHt}\dd t,
\end{aligned}
\end{equation}
where
\begin{equation}
\label{eq:F1_inverse_fourier_transform}
    \begin{aligned}
        \check{F}_1(t,\omega) &= \frac{1}{\sqrt{2\pi}}\int F_1(\nu,\omega)e^{-i\nu t}\dd \nu \\
        &=\frac{1}{2\pi}\int\dd\nu\int\dd t' e^{-\beta\nu/4}e^{i(\nu-\omega)t'}f(t')e^{-i\nu t} \\
        &=\frac{e^{1/16}}{\sqrt{\beta\sqrt{\pi/2}}}e^{-t^2/\beta^2}e^{-it(\omega-1/(2\beta))}e^{-\beta\omega/4}.
    \end{aligned}
\end{equation}
Similarly, for $\tilde{A}^{\mathrm{free}}_j(\omega)$ we have
\begin{equation}
    \label{eq:A_tilde_LC_time_evolution_free}
    \tilde{A}^{\mathrm{free}}_j(\omega) =  \frac{1}{\sqrt{2\pi}}\int \check{F}_1(t,\omega) e^{iH_0t}\gamma_j e^{-iH_0t}\dd t.
\end{equation}

\subsection{Quasi-locality of the jump operator}
\label{sec:quasi_locality_jump}

We will first prove that $\tilde{A}_j(\omega)$ (defined in \eqref{eq:A_tilde_and_B_tilde}) can be well-approximated by local operators.
\begin{lem}
    \label{lem:quasi_locality_jump}
    There exists universal constants $C,\mu_1,\mu_2>0$ such that for any $r\geq 1$, there exists an operator $A^{\mathrm{loc}}$ supported on the ball with radius $r$ centered at $j$ such that
    \begin{equation}
    \label{eq:jump_op_local_approx_err}
        \|\tilde{A}_j(\omega)-A^{\mathrm{loc}}\|\leq C\left(\sqrt{\frac{\beta}{r}}e^{-\mu_1 r}+\frac{1}{\sqrt{\beta}}e^{-\mu_2 r^2/(J\beta)^2}\right)e^{-\beta\omega/4},
    \end{equation}
    where $J$ is the Lieb-Robinson velocity of $H$ on $\Lambda$ for which an upper bound is given in \eqref{eq:LR_velocity_full_ham}. 
    This implies that there exists constants $C',\mu>0$ that depend only on $\beta$ and $J$ such that
    \begin{equation}
    \label{eq:jump_op_local_approx_err_simple}
        \|\tilde{A}_j(\omega)-A^{\mathrm{loc}}\|\leq C'e^{-\mu r}e^{-\beta\omega/4}.
    \end{equation}
\end{lem}

\begin{proof}
    We first decompose $\tilde{A}_j(\omega)=\mathrm{I}+\mathrm{II}$, where
    \begin{equation}
        \mathrm{I} = \frac{1}{\sqrt{2\pi}}\int_{-T}^{T}\check{F}_1(t,\omega)e^{iHt}\gamma_j e^{-iHt}\dd t,\quad \mathrm{II} = \frac{1}{\sqrt{2\pi}}\int_{|t|>T}\check{F}_1(t,\omega)e^{iHt}\gamma_j e^{-iHt}\dd t.
    \end{equation}
    In the above $T>0$ can be arbitrarily chosen.  
    From \eqref{eq:F1_inverse_fourier_transform}, we can see that
    \begin{equation}
        |\check{F}_1(t,\omega)|\leq C_1 e^{-t^2/\beta^2}e^{-\beta\omega /4}.
    \end{equation}
    Therefore we have
    \begin{equation}
        \label{eq:locality_jump_II}
        \|\mathrm{II}\|\leq \frac{C_2}{\sqrt{\beta}} e^{-T^2/\beta^2}e^{-\beta\omega /4}.
    \end{equation}
    We then approximate $e^{iHt}\gamma_j e^{-iHt}$ with a local operator. Let $H'$ be the Hamiltonian consisting of terms of $H$ that are supported on the ball around $j$ with radius $r$, then using the Lieb-Robinson bound \cite[Eq.~(A1)]{rouze2024efficient}, we have
    \begin{equation}
        \|e^{iHt}\gamma_j e^{-iHt}-e^{iH't}\gamma_j e^{-iH't}\|\leq \frac{(2J|t|)^r}{r!}.
    \end{equation}
    Here $e^{iH't}\gamma_j e^{-iH't}$ is supported only on this ball with radius $r$. We then let
    \begin{equation}
        A^{\mathrm{loc}} = \frac{1}{\sqrt{2\pi}}\int_{-T}^{T}\check{F}_1(t,\omega)e^{iH't}\gamma_j e^{-iH't}\dd t.
    \end{equation}
    This is a good approximate of $\mathrm{I}$:
    \begin{equation}
        \|\mathrm{I}-A^{\mathrm{loc}}\|\leq \frac{1}{\sqrt{2\pi}}\int_{-T}^T |\check{F}_1(t,\omega)|\frac{(2J|t|)^r}{r!}\dd t\leq C_3\sqrt{\beta} \frac{(2JT)^r}{r!}e^{-\beta\omega/4}.
    \end{equation}
    Therefore
    \begin{equation}
        \|\tilde{A}_j(\omega)-A^{\mathrm{loc}}\|\leq \|\mathrm{I}-A^{\mathrm{loc}}\|+\|\mathrm{II}\|\leq C_3\sqrt{\beta} \frac{(2JT)^r}{r!}e^{-\beta\omega/4}+\frac{C_2}{\sqrt{\beta}} e^{-T^2/\beta^2}e^{-\beta\omega /4}.
    \end{equation}
    Using Stirling's formula we have $r!\geq \sqrt{2\pi r}(r/e)^r$ for $r\geq 1$, and therefore we further have
    \begin{equation}
        \|\tilde{A}_j(\omega)-A^{\mathrm{loc}}\|\leq C\left(\sqrt{\frac{\beta}{r}}\left(\frac{2eJT}{r}\right)^r+\frac{1}{\sqrt{\beta}}e^{-T^2/\beta^2}\right)e^{-\beta\omega /4}.
    \end{equation}
    Since we can arbitrarily choose $T$, we let $T=r/(2e^2 J)$, and then we have \eqref{eq:jump_op_local_approx_err} with $\mu_1=1$ and $\mu_2=1/4e^4$.
\end{proof}

\subsection{Quasi-locality of the free part}
\label{sec:quasi_locality_free_part}

By the exact same way as for Lemma~\ref{lem:quasi_locality_jump} we can prove the following:
\begin{lem}
    \label{lem:quasi_locality_jump_free}
    There exists universal constants $C,\mu_1,\mu_2>0$ such that for any $r\geq 1$, there exists an operator $A^{\mathrm{loc}}(\omega)$ supported on the ball with radius $r$ centered at $j$ such that
    \begin{equation}
    \label{eq:jump_op_local_approx_err_free}
        \|\tilde{A}^{\mathrm{free}}_j(\omega)-A^{\mathrm{loc}}(\omega)\|\leq C\left(\sqrt{\frac{\beta}{r}}e^{-\mu_1 r}+\frac{1}{\sqrt{\beta}}e^{-\mu_2 r^2/(J_0\beta)^2}\right)e^{-\beta\omega/4},
    \end{equation}
    where $J_0$ is the Lieb-Robinson velocity of $H_0$ on $\Lambda$ for which an upper bound is given in \eqref{eq:LR_velocity_non_interacting}. 
    This implies that there exists constants $C',\mu>0$ that depend only on $\beta$ and $J_0$ such that
    \begin{equation}
    \label{eq:jump_op_local_approx_err_simple_free}
        \|\tilde{A}^{\mathrm{free}}_j(\omega)-A^{\mathrm{loc}}(\omega)\|\leq C'e^{-\mu r}e^{-\beta\omega/4}.
    \end{equation}
    Moreover, both $\tilde{A}^{\mathrm{free}}_j(\omega)$ and $A^{\mathrm{loc}}(\omega)$ are both linear combinations of Majorana operators.
\end{lem}
\begin{proof}
    We only need to prove the statement in the last sentence as the rest is the same as for Lemma~\ref{lem:quasi_locality_jump}.
    We note that because $H_0$ is quadratic and $\gamma_j$ is a Majorana operator, $e^{iH_0 t}\gamma_j e^{-iH_0 t}$ is a linear combination of Majorana operators. This ensures that $\tilde{A}^{\mathrm{free}}_j(\omega)$ is also a linear combination of Majorana operators.
    In the construction of $A^{\mathrm{loc}}(\omega)$, we restrict the Hamiltonian $H_0$ to a local region, but it is still quadratic. Therefore the same also holds for $A^{\mathrm{loc}}(\omega)$.
\end{proof}
 We therefore have the following corollary for the coefficients when expanding $\tilde{A}^{\mathrm{free}}_j(\omega)$ as a linear combination of Majorana operators:
\begin{cor}
    \label{cor:coef_jump_op_free}
    The free part of the jump operator as defined in \eqref{eq:A_tilde_free} satisfies
    \[
    \tilde{A}^{\mathrm{free}}_j(\omega) = \sum_{k} c_k \gamma_k,
    \]
    where $c_k\in\CC$, and for any $r>0$,
    \[
    \sum_{k:d(j,k)>r}|c_k|^2\leq C'^2 e^{-2\mu r}e^{-\beta\omega/2}.
    \]
    Here $C'$ and $\mu$ are the same as in Lemma~\ref{lem:quasi_locality_jump_free}.\footnote{\label{fn:distance_notation} With a slight abuse of notation, here we use $d(j,k)$ to mean $d(\mathfrak{s}(j),\mathfrak{s}(k))$ where $d(\cdot,\cdot)$ is the Euclidean distance on lattice $\Lambda$ and $\mathfrak{s}(\cdot)$ maps a Majorana mode to the lattice site it is on.}
\end{cor}

\begin{proof}
    By Lemma~\ref{lem:quasi_locality_jump_free}, we can write $\tilde{A}^{\mathrm{free}}_j(\omega) = \sum_{k} c_k \gamma_k$, and for any $r>0$, we have
    \[
    A^{\mathrm{loc}}(\omega) = \sum_{k:d(j,k)\leq r} c'_k \gamma_k,
    \]
    such that $\|\tilde{A}^{\mathrm{free}}_j(\omega)-A^{\mathrm{loc}}(\omega)\|\leq C'e^{-\mu r}e^{-\beta\omega/4}$. By Lemma~\ref{lem:norm_linear_comb_majorana}, we have
    \[
    \sum_{k:d(j,k)> r} |c_k|^2\leq \sum_{k:d(j,k)> r} |c_k|^2 + \sum_{k:d(j,k)\leq r} |c_k-c'_k|^2\leq \|\tilde{A}^{\mathrm{free}}_j(\omega)-A^{\mathrm{loc}}(\omega)\|^2\leq C'^2 e^{-2\mu r}e^{-\beta\omega/2}.
    \]
\end{proof}

This structure for $\tilde{A}^{\mathrm{free}}_j(\omega)$ gives us a structural result for $H_{\mathrm{D,free}}^{\mathrm{parent}}$, i.e., the part of the parent Hamiltonian corresponding to the free part of the dissipation terms:
\begin{lem}
\label{lem:quasi_locality_free_dissipative}
    For $H_{\mathrm{D,free}}^{\mathrm{parent}}$ defined in \eqref{eq:parent_ham_dissipative_free}, there exists a constant $C$ that depends only on $J_0$ in Lemma~\ref{lem:quasi_locality_jump_free} and $\beta$ such that
    \[
    H_{\mathrm{D,free}}^{\mathrm{parent}} = \sum_{l,l'=1}^{4n} m_{l,l'}\hat{\gamma}_{l}\hat{\gamma}_{l'},
    \]
    where
    \[
    |m_{l,l'}|\leq C e^{-\mu d(\mathfrak{s}'(l),\mathfrak{s}'(l'))}.
    \]
    Here $\mathfrak{s}'(l)$ indicate the lattice site in $\Lambda$ on which the $a$-Majorana mode $l$ resides.
    In the above the constant $\mu>0$ is the same as in Lemma~\ref{lem:quasi_locality_jump_free}.
\end{lem}

\begin{proof}
    From \eqref{eq:parent_ham_dissipative_free} we can write 
    \[
    H_{\mathrm{D,free}}^{\mathrm{parent}} = \mathrm{I}+\mathrm{II}+\mathrm{III},
    \]
    where
    \[
    \begin{aligned}
        \mathrm{I} &= \Phi\Bigg(\sum_j\int\dd\omega \eta(\omega)  L_{\tilde{A}^{\mathrm{free}}_j(-\omega)}R_{\tilde{A}^{\mathrm{free}}_j(-\omega)^\dag} \Bigg), \\
        \mathrm{II} &= -\frac{1}{2}\Phi\Bigg(\sum_j\int\dd\omega \eta(\omega)  L_{\tilde{A}^{\mathrm{free}}_j(-\omega) \tilde{A}^{\mathrm{free}}_j(\omega)} \Bigg),\\
         \mathrm{III} &=-\frac{1}{2}\Phi\Bigg(\sum_j\int\dd\omega \eta(\omega)R_{\tilde{A}^{\mathrm{free}}_j(\omega)^\dag \tilde{A}^{\mathrm{free}}_j(-\omega)^\dag}\Bigg).
    \end{aligned}
    \]
    We can simply prove that each of the above three terms has the structure as stated in the lemma. Because the proof for each of them is essentially the same, we only prove for the term $\mathrm{I}$. By Definition~\ref{defn:Phi_thrid_quantization} we have
    \[
    \mathrm{I} = \sum_j\int\dd\omega \eta(\omega)  \Phi(L_{\tilde{A}^{\mathrm{free}}_j(-\omega)})\Phi(R_{\tilde{A}^{\mathrm{free}}_j(-\omega)^\dag}). 
    \]
    By Lemma~\ref{cor:coef_jump_op_free}, we have
    \[
    \tilde{A}^{\mathrm{free}}_j(\omega) = \sum_k c_{j,k}(\omega)\gamma_k,
    \]
    where $\sum_{k:d(j,k)>r}|c_{j,k}(\omega)|^2\leq C'^2 e^{-2\mu r}e^{-\beta\omega/2}$. This means $\Phi(L_{\tilde{A}^{\mathrm{free}}_j(\omega)^\dag})$ and $\Phi(R_{\tilde{A}^{\mathrm{free}}_j(\omega)})$ admit similar decompositions by \eqref{eq:repr_left_right_multiplication_single_majorana}:
    \begin{equation}
        \Phi(L_{\tilde{A}^{\mathrm{free}}_j(-\omega)}) = \sum_k c_{j,k}(-\omega)\hat{\gamma}_{2k-1},\quad \Phi(R_{\tilde{A}^{\mathrm{free}}_j(-\omega)^\dag}) = -i\sum_k c_{j,k}(-\omega)^*\hat{\gamma}_{2k},
    \end{equation}
    Therefore
    \begin{equation}
        \mathrm{I} = -i\sum_{k,k'}\sum_j\int\dd\omega \eta(\omega)c_{j,k}(-\omega)c_{j,k'}(-\omega)^*\hat{\gamma}_{2k-1}\hat{\gamma}_{2k}.
    \end{equation}
    Let
    \begin{equation}
        \mathrm{I} = \sum_{k,k'}m^{(\mathrm{I})}_{k,k'} \hat{\gamma}_{2k-1}\hat{\gamma}_{2k},
    \end{equation}
    then 
    \begin{equation}
    \begin{aligned}
        |m^{(\mathrm{I})}_{k,k'}| &\leq  \sum_j\int\dd\omega \eta(\omega)|c_{j,k}(-\omega)c_{j,k'}(-\omega)^*| \\
        &\leq C'^2 \int\dd\omega \eta(\omega)e^{-\beta\omega/2}\sum_j e^{-\mu (d(j,k)+d(j,k'))} \\
        &\leq C''\sum_j e^{-\mu (d(k,k')+d(j,(k+k')/2))} \\
        &\leq Ce^{-\mu d(k,k')}.
    \end{aligned}
    \end{equation}
    Since $\mathrm{II}$ and $\mathrm{III}$ have the same structure as can be proved in the same way, we have proved the statement in the lemma.
\end{proof}

We can now prove Proposition~\ref{prop:non_interacting_part_decay}:
\begin{proof}[Proof of Proposition~\ref{prop:non_interacting_part_decay}]
    By Lemma~\ref{lem:quasi_locality_free_dissipative}, and the bound for the Lieb-Robinson velocity $J_0$ in \eqref{eq:LR_velocity_non_interacting}, we can see that $H^{\mathrm{parent}}_{\mathrm{D,free}}$ has $[K,\mu]$-decay for some constants $K,\mu>0$ that depend only on $\beta,r_0,D$. By Lemma~\ref{lem:free_parent_zero} we also know that $H^{\mathrm{parent}}_{\mathrm{C,free}}=0$. Therefore $H^{\mathrm{parent}}_{\mathrm{free}}=H^{\mathrm{parent}}_{\mathrm{D,free}}+H^{\mathrm{parent}}_{\mathrm{C,free}}=H^{\mathrm{parent}}_{\mathrm{D,free}}$ has $[K,\mu]$-decay.
\end{proof}

\subsection{The perturbation in the jump operator}
\label{sec:quasi_locality_perturbation}

In this section we show that the perturbation is also quasi-local and is small in $U$. We start from the time evolve operators $\gamma_j(t)=e^{iHt}\gamma_j e^{-iHt}$ as compared to $\gamma^{\mathrm{free}}_j(t)=e^{iH_0 t}\gamma_j e^{-iH_0 t}$:
\begin{equation}
    \gamma_j(t) = \gamma^{\mathrm{free}}_j(t) + \gamma^{\mathrm{int}}_j(t), 
\end{equation}
where $\gamma^{\mathrm{free}}_j(t)$ is given in \eqref{eq:defn_gamma_free}, and 
\begin{equation}
    \label{eq:gamma_int}
    \gamma^{\mathrm{int}}_j(t)=i\int_0^t e^{iH(t-s)}[V,e^{iH_0 s}\gamma_je^{-iH_0 s}]e^{-iH(t-s)}\dd s.
\end{equation}
By \eqref{eq:A_tilde_LC_time_evolution} and \eqref{eq:A_tilde_LC_time_evolution_free}, we then have
\begin{equation}
\label{eq:jump_op_difference}
    \tilde{A}_j^{\mathrm{int}}(\omega):=\tilde{A}_j(\omega)- \tilde{A}_j^{\mathrm{free}}(\omega)= \frac{1}{\sqrt{2\pi}}\int \dd t \check{F}_1(t,\omega)\gamma^{\mathrm{int}}_j(t).
\end{equation}
Because $V$ is an ``extensive'' operator in the sense that it acts on the entire system, we need to make sure that it does not blow up the approximation errors by a factor that depends on the system size. Therefore below we will first focus on $[V,e^{iH_0 s}\gamma_je^{-iH_0 s}]$.

\begin{lem}
    \label{lem:commutator_approx_V_gamma_free}
    Let $\gamma_j^{\mathrm{free}}(s)=\sum_k c_k(s)\gamma_k$, and for some $r>0$ let
    \[
    X_1(s) = \sum_{k:d(j,k)\leq r} c_k(s)\gamma_k.
    \]
    Then if $r\geq 2e^2 J_0|s|$ where $J_0$ is the same as in Lemma~\ref{lem:quasi_locality_jump_free}, we have
    \[
    \|[V,\gamma_j^{\mathrm{free}}(s)-X_1(s)]\|\leq C U (r+1)^D e^{-r},
    \]
    for some constant $C$ that only depends on the lattice dimension $D$ and $J_0$.
\end{lem}

\begin{proof}
    By the Lieb-Robinson bound and a derivation similar to Lemma~\ref{cor:coef_jump_op_free} we have 
    \[
    \left\|\sum_{k:d(j,k)> r'} c_k(s)\gamma_k\right\|\leq \frac{(2J_0 |s|)^{r'}}{r'!}.
    \]
    By Lemma~\ref{lem:norm_linear_comb_majorana} we then have
    \[
    \sum_{k:d(j,k)> r'} |c_k(s)|^2\leq \frac{(2J_0 |s|)^{2r'}}{(r'!)^2}.
    \]
    This further implies
    \begin{equation}
        |c_k(s)|\leq \inf_{r'\in \ZZ_+: d(j,k)> r'} \frac{(2J_0 |s|)^{r'}}{r'!} \leq \frac{(2J_0 |s|)^{d(j,k)-1}}{(d(j,k)-1)!}\leq \frac{1}{\sqrt{2\pi (d(j,k)-1)}}\left(\frac{2\pi J_0 |s|}{d(j,k)-1}\right)^{d(j,k)-1}.
    \end{equation}
    When $d(j,k)\geq 2e^2 J_0|s|+1$, we have
    \begin{equation}
        |c_k(s)|\leq C_1e^{-d(j,k)},
    \end{equation}
    for some constant $C_1>0$ that only depends on $J_0$. Therefore
    \[
    \sum_{k:d(j,k)\geq r}|c_k(s)|\leq C_1e^{-d(j,k)}\leq C_2\int_r^{\infty}\dd r' r'^D e^{-r'} \leq C_3(r+1)^D e^{-r},
    \]
    where $C_2>0$ and $C_3>0$ are both constants that only depend on $J_0$ and $D$.

    Because $V$ is $(U,r_0)$-geometrically local, $\|[V,\gamma_k]\|\leq C_4 U$ for some constant $C_4>0$ that depends only on $r_0$. Therefore
    \[
    \|[V,\gamma_j^{\mathrm{free}}(s)-X_1(s)]\|\leq C_3C_4 U(r+1)^D e^{-r}.
    \]
    We can therefore set $C=C_3C_4$.
\end{proof}

This directly allows us to bound the norm of $[V,e^{iH_0 s}\gamma_je^{-iH_0 s}]$.
\begin{cor}
    \label{cor:commutator_norm_growth}
    There exists constant $C',\mu'>0$ that only depends on $r_0,J_0,D$ such that
    \[
    \|[V,\gamma_j^{\mathrm{free}}(s)]\|\leq C'U(|s|+1)^D (1+e^{-\mu' |s|}).
    \]
\end{cor}

\begin{proof}
    In Lemma~\ref{lem:commutator_approx_V_gamma_free}, we can set $r=2e^2 J_0 |s|+1$, and then we have
    \begin{equation}
        \left\|\left[V,\sum_{k:d(j,k)\geq r} c_k(s)\gamma_k\right]\right\|=\left\|\left[V,\gamma_j^{\mathrm{free}}(s)-X_1(s)\right]\right\|\leq C'U(|s|+1)^D e^{-\mu' |s|}.
    \end{equation}
    Then we note that $\sum_{k:d(j,k)< r} c_k(s)\gamma_k$ commute with all but $C_1 (r+1)^D$ terms in $V$ for some constant $C_1>0$ that depends only on $r_0$, and therefore
    \begin{equation}
    \label{eq:X1_V_commutator_norm_bound}
        \left\|\left[V,\sum_{k:d(j,k)< r} c_k\gamma_k\right]\right\|=\left\|\left[V,X_1(s)\right]\right\|\leq C'U(|s|+1)^D\max_k|c_k(s)|\leq C'U(|s|+1)^D.
    \end{equation}
    Here we have used Lemma~\ref{lem:norm_linear_comb_majorana} to obtain $\max_k|c_k(s)|\leq \|\gamma_j^{\mathrm{free}}(s)\|=1$.
    Adding up the two parts yields the desired bound by the triangle inequality.
\end{proof}
This bound is useful because it is independent of the system size.

\begin{lem}
    \label{lem:X2(s,t)}
    We define $X_1(s)$ and $X_2(s,t)$ as follows:
    \[
    X_1(s) = \sum_{k:d(j,k)\leq r/2} c_k(s)\gamma_k,
    \]
    \[
    X_2(s,t) = ie^{iH'(t-s)}[V,X_1(s)]e^{-iH'(t-s)},
    \]
    where $H'$ consists of the Hamiltonian terms in $H$ that are within $r$ distance from $j$, and $c_k(s)$ is the same as in Lemma~\ref{lem:commutator_approx_V_gamma_free}. Then $X_s(s,t)$ \YZ{two $s$ here?} is supported on the ball with radius $r$ centered at $j$, and there exists constant $C$ that only depends on $r_0,D$ such that 
    \begin{equation}
    \label{eq:X2(s,t)_approx_err_bound}
        \|X_2(s,t)-e^{iH(t-s)}[V,\gamma_j^{\mathrm{free}}(s)]e^{-iH(t-s)}\|\leq CU(|s|+1)^D e^{-r/2},
    \end{equation}
    if $r\geq \max\{4e^2J|t|+2r_0,2e^2 J_0|s|\}$.
    In the above $J_0,J>0$ are the Lieb-Robinson velocities defined in \eqref{eq:LR_velocity_full_ham} and \eqref{eq:LR_velocity_non_interacting} respectively.
\end{lem}

\begin{proof}
    By the triangle inequality, we have
    \begin{equation}
    \begin{aligned}
        &\|X_2(s,t)-ie^{iH(t-s)}[V,\gamma_j^{\mathrm{free}}(s)]e^{-iH(t-s)}\|\\
        &\leq \underbrace{\|X_2(s,t)-ie^{iH(t-s)}[V,X_1(s)]e^{-iH(t-s)}\|}_{\mathrm{I}} \\
        &+\underbrace{\|e^{iH(t-s)} [V,X_1(s)]e^{-iH(t-s)}-e^{iH(t-s)}[V,\gamma_j^{\mathrm{free}}(s)]e^{-iH(t-s)}\|}_{\mathrm{II}}.
    \end{aligned}
    \end{equation}
    For $\mathrm{I}$, by the Lieb-Robinson bound, since $[V,X_1(s)]$ is supported on the ball of radius $r/2+r_0$ centered at $j$, when $r\geq r/2+r_0+2e^2J|t|$, we have
    \[
    \mathrm{I}\leq C_1 e^{-r/2} \|[V,X_1(s)]\|\leq C_2 e^{-r/2} U(|s|+1)^D,
    \]
    where the second inequality comes from \eqref{eq:X1_V_commutator_norm_bound}. 
    For $\mathrm{II}$, by Lemma~\ref{lem:commutator_approx_V_gamma_free} we have
    \[
    \mathrm{II}=\|[V,X_1(s)]-[V,\gamma_j^{\mathrm{free}}(s)]\|\leq C_3 e^{-r/2} U(|s|+1)^D.
    \]
    We thus have \eqref{eq:X2(s,t)_approx_err_bound} by combining the two inequalities for $\mathrm{I}$ and $\mathrm{II}$ above.
\end{proof}

\begin{cor}
    \label{cor:approx_gamma_int}
    Under the same assumptions as in Lemma~\ref{lem:X2(s,t)}, we have
    \[
    \left\|\gamma_j^{\mathrm{int}}(t)-\int_0^t X_2(s,t)\dd s\right\|\leq C e^{-r/2} U(|t|+1)^{D+1},
    \]
    Here $\int_0^t X_2(s,t)\dd s$ is supported on the ball with radius $r$ centered at $j$, satisfying
    \[
    \left\|\int_0^t X_2(s,t)\dd s\right\|\leq C' U(|t|+1)^{D+1}.
    \]
    $C$ and $C'$ are constants that depend only on $r_0,D$. In the above $J_0,J>0$ are the Lieb-Robinson velocities defined in \eqref{eq:LR_velocity_full_ham} and \eqref{eq:LR_velocity_non_interacting} respectively.
\end{cor}

\begin{proof}
    The first statement is a direct consequence of Lemma~\ref{lem:X2(s,t)}: one only needs to integrate over $s$ on both sides of the inequality. The second statement is a direct consequence of \eqref{eq:X1_V_commutator_norm_bound} likewise after integrating over $s$.
\end{proof}

\begin{lem}
    \label{lem:approx_A_tilde_j_int}
    Let
    \[
    X_3(\omega) = \frac{1}{\sqrt{2\pi}}\int_{-T}^T \check{F}_1(t,\omega)\int_0^t X_2(s,t)\dd s \dd t,
    \]
    where $X_2(s,t)$ is defined in Lemma~\ref{lem:X2(s,t)}. Then $X_3(\omega)$ is supported on the ball with radius $r$ centered at $j$, and with $T=(r-2r_0)/(2e^2(2J+J_0))$, 
    we have
    \[
    \|X_3(\omega)-\tilde{A}_j^{\mathrm{int}}(\omega)\|\leq CUe^{-\mu r}e^{-\beta\omega/4},\quad \|X_3(\omega)\|,\|\tilde{A}_j^{\mathrm{int}}(\omega)\|\leq C'Ue^{-\beta\omega/4},
    \]
    for all $\omega\in \RR$, and
    $C,C',\mu>0$ are constants that depend only on $r_0,D,\beta$. In the above $J_0,J>0$ are the Lieb-Robinson velocities defined in \eqref{eq:LR_velocity_full_ham} and \eqref{eq:LR_velocity_non_interacting} respectively.
\end{lem}

\begin{proof}
    By the triangle inequality,
    \[
    \begin{aligned}
        &\|X_3(\omega)-\tilde{A}_j^{\mathrm{int}}(\omega)\| \\
        &\leq \underbrace{\left\|\frac{1}{\sqrt{2\pi}}\int_{|t|\geq T}\check{F}_1(t,\omega
        )\gamma_j^{\mathrm{int}}(t)\dd t\right\|}_{\mathrm{I}} \\
        &+\underbrace{\left\|\frac{1}{\sqrt{2\pi}}\int_{-T}^{T}\check{F}_1(t,\omega
        )\left(\int_0^t X_2(s,t)\dd s-\gamma_j^{\mathrm{int}}(t)\right)\dd t\right\|}_{\mathrm{II}}.
    \end{aligned}
    \]
    For $\mathrm{I}$, we have $\mathrm{I}\leq C_1 U e^{-T^2/\beta^2}e^{-\beta\omega/4}$ for some constant $C_1$ that depends only on $r_0,J_0,J,D$ because Corollary~\ref{cor:approx_gamma_int} tells us that
    \[
    \|\gamma_j^{\mathrm{int}}(t)\|\leq C_1' U(|t|+1)^{D+1}.
    \]
    for some constant $C_1'$.
    For $\mathrm{II}$, when $T=(r-2r_0)/(2e^2(2J+J_0))$, we know that $r\geq \max\{4e^2J|t|+2r_0,2e^2 J_0|s|\}$, and therefore by Corollary~\ref{cor:approx_gamma_int}, we have
    \[
    \left\|\gamma_j^{\mathrm{int}}(t)-\int_0^t X_2(s,t)\dd s\right\|\leq C_2' U(|t|+1)^D e^{-r/2}e^{-\beta\omega/4}.
    \]
    This the implies
    \[
    \mathrm{II}\leq C_2 Ue^{-r/2}e^{-\beta\omega/4},
    \]
    where $C_2$ is a constant that depends only on $r_0,J_0,J,D$.
    Combining the two bounds we have
    \[
    \|X_3(\omega)-\tilde{A}_j^{\mathrm{int}}(\omega)\|\leq C_1 U e^{-T^2/\beta^2}e^{-\beta\omega/4} + C_2 U e^{-r/2}e^{-\beta\omega/4}.
    \]
    By the fact that $T=(r-2r_0)/(2e^2(2J+J_0))$, there exists constants $C,\mu>0$ that only depends on $\beta,r_0,J_0,J,D$ such that
    \[
    \|X_3(\omega)-\tilde{A}_j^{\mathrm{int}}(\omega)\|\leq CUe^{-\mu r}e^{-\beta\omega/4}.
    \]
    The bound for $\|X_3(\omega)\|$ follows directly from the definition of $X_3(\omega)$ and the norm bound for $\int_0^t X_2(s,t)\dd s$ in Corollary~\ref{cor:approx_gamma_int}. The bound for $\|\tilde{A}_j^{\mathrm{int}}(\omega)\|$ then follows from the triangle inequality.
\end{proof}

\subsection{The perturbation in the dissipative part}
\label{sec:perturbation_dissipative}

The above result shows that $\tilde{A}_j^{\mathrm{int}}(\omega)$ can be well-approximated by local operators and can be bounded linearly in $U$. We will then use it to study the perturbation $H^{\mathrm{parent}}_{\mathrm{D,int}}$.

We first decompose $H^{\mathrm{parent}}_{\mathrm{D,int}}$, as defined in \eqref{eq:parent_ham_dissipative_int}, into three parts:
\begin{equation}
    \label{eq:H_D_int_decompose}
    H^{\mathrm{parent}}_{\mathrm{D,int}} = H_{LR} + H_{LL} + H_{RR},
\end{equation}
where
\[
\begin{aligned}
    H_{LR} &= \sum_j \int \dd\omega \eta(\omega)\Phi(L_{\tilde{A}_j(-\omega)}R_{\tilde{A}_j(-\omega)^\dag}-L_{\tilde{A}^{\mathrm{free}}_j(-\omega)}R_{\tilde{A}^{\mathrm{free}}_j(-\omega)^\dag}), \\
    H_{LL} &= -\frac{1}{2}\sum_j\int \dd\omega \eta(\omega)\Phi( L_{\tilde{A}_j(-\omega)\tilde{A}_j(\omega)-\tilde{A}^{\mathrm{free}}_j(-\omega)\tilde{A}^{\mathrm{free}}_j(\omega)} ), \\
    H_{RR} &= -\frac{1}{2}\sum_j\int \dd\omega \eta(\omega)\Phi( R_{\tilde{A}_j(\omega)^\dag\tilde{A}_j(-\omega)^\dag - \tilde{A}^{\mathrm{free}}_j(\omega)^\dag\tilde{A}^{\mathrm{free}}_j(-\omega)^\dag} ).
\end{aligned}
\]
We will focus on $H_{LR}$ because the other two terms can be dealt with in exactly the same way. For this term we can further decompose it into three parts
\begin{equation}
    \label{eq:H_LR_decompose}
    H_{LR} = H_{LR,1} + H_{LR,2} + H_{LR,3},
\end{equation}
where
\[
\begin{aligned}
    H_{LR,1} &= \sum_j \int \dd\omega \eta(\omega)\Phi(L_{\tilde{A}^{\mathrm{int}}_j(-\omega)}R_{\tilde{A}^{\mathrm{free}}_j(-\omega)^\dag}), \\
    H_{LR,2} &= \sum_j \int \dd\omega \eta(\omega)\Phi(L_{\tilde{A}^{\mathrm{free}}_j(-\omega)}R_{\tilde{A}^{\mathrm{int}}_j(-\omega)^\dag}), \\
    H_{LR,3} &= \sum_j \int \dd\omega \eta(\omega)\Phi(L_{\tilde{A}^{\mathrm{int}}_j(-\omega)}R_{\tilde{A}^{\mathrm{int}}_j(-\omega)^\dag}).
\end{aligned}
\]
In the above we have used $\tilde{A}_j(\omega)=\tilde{A}^{\mathrm{int}}_j(\omega)+\tilde{A}^{\mathrm{free}}_j(\omega)$ and the linearity of the mappings. We will again only focus on $H_{LR,1}$ because the other terms can be dealt with in exactly the same way.

\begin{lem}
    \label{lem:H_LR_1_term_approx}
    Let 
    \[
    Y_j(\omega) = \Phi(L_{\tilde{A}^{\mathrm{int}}_j(-\omega)}R_{\tilde{A}^{\mathrm{free}}_j(-\omega)^\dag}),
    \]
    and
    \[
    Y_{j,r}(\omega) = \Phi(L_{X_3(-\omega)}R_{A^{\mathrm{loc}}(-\omega)^\dag}),
    \]
    where $X_3(\omega)$ and $A^{\mathrm{loc}}(\omega)$ are from Lemma~\ref{lem:approx_A_tilde_j_int} and Lemma~\ref{lem:quasi_locality_jump_free} respectively, with $r$ chosen to be the same value and $\omega$ as given. Then there exists constants $C,C',\mu>0$ that depend only on $\beta,r_0,D$ such that
    \[
    \|Y_j(\omega)-Y_{j,r}(\omega)\|\leq CU e^{-\mu r} e^{\beta\omega/2},\quad \|Y_{j}(\omega)\|\leq C'U  e^{\beta\omega/2},
    \]
    and $Y_{j,r}(\omega)$ is supported on the ball of radius $r$ centered at $j$.\footnote{More precisely, because $Y_{j,r}(\omega)\in B(\mathcal
    {H}')$ is an operator of $a$-fermions, it acts non-trivially only on the $a$-fermion modes within distance $r$ from Majorana mode $j$.}
\end{lem}

\begin{proof}
    We first have
    \[
    \begin{aligned}
        \|Y_j(\omega)-Y_{j,r}(\omega)\| &\leq \|\Phi(L_{\tilde{A}^{\mathrm{int}}_j(-\omega)-X_3(-\omega)}R_{\tilde{A}^{\mathrm{free}}_j(-\omega)})\| \\
        &+ \|\Phi(L_{X_3(-\omega)}R_{\tilde{A}^{\mathrm{free}}_j(-\omega)^\dag-A^{\mathrm{loc}}(-\omega)^\dag})\| \\
        &\leq \|\Phi(L_{\tilde{A}^{\mathrm{int}}_j(-\omega)^\dag-X_3(-\omega)})\|\|\Phi(R_{\tilde{A}^{\mathrm{free}}_j(-\omega)^\dag})\| \\
        &+\|\Phi(L_{X_3(-\omega)})\|\|\Phi(R_{\tilde{A}^{\mathrm{free}}_j(-\omega)^\dag-A^{\mathrm{loc}}(-\omega)^\dag})\| \\
        &\leq \underbrace{\|\tilde{A}^{\mathrm{int}}_j(-\omega)-X_3(-\omega)\| \|\tilde{A}^{\mathrm{free}}_j(-\omega)\|}_{\mathrm{I}} \\
        &+\underbrace{\|X_3(-\omega)\|\|\tilde{A}^{\mathrm{free}}_j(-\omega)-A^{\mathrm{loc}}(-\omega)\|}_{\mathrm{II}},
    \end{aligned}
    \]
    where in the third inequality we have used $\|\Phi(L_A)\|,\|\Phi(R_A)\|\leq \|A\|$ as stated in Lemma~\ref{lem:norm_preservation}.

    By Lemma~\ref{lem:approx_A_tilde_j_int}, we have $\|\tilde{A}^{\mathrm{int}}_j(\omega)-X_3(\omega)\|\leq C_1 U e^{-\mu_1 r} e^{-\beta\omega/4}$ and $\|X_3(\omega)\|\leq C_1'Ue^{-\beta\omega/4}$. By Lemma~\ref{lem:quasi_locality_jump_free}, we have $\|\tilde{A}^{\mathrm{free}}_j(\omega)-A^{\mathrm{loc}}(\omega)\|\leq C_2 e^{-\mu_2 r}e^{-\beta\omega/4}$. By \eqref{eq:A_tilde_free} we also have $\|\tilde{A}^{\mathrm{free}}_j(\omega)\|\leq C_2' e^{-\beta\omega/4}$. Therefore
    \[
    \mathrm{I}\leq C_1C_2' U e^{-\mu_1 r} e^{\beta\omega/2},\quad \mathrm{II}\leq C_1'C_2 U e^{-\mu_2 r} e^{\beta\omega/2}.
    \]
    Adding up the two terms we have
    \[
    \|Y_j(\omega)-Y_{j,r}(\omega)\|\leq CU e^{-\mu r} e^{\beta\omega/2}
    \]
    for some constants $C,\mu>0$. From \eqref{eq:jump_operators} we also have $\|\tilde{A}_j(\omega)\| C_3 e^{\beta\omega/4}$ for some constant $C_3$. Therefore
    \[
    \|Y_j(\omega)\|\leq C'e^{\beta\omega/2}
    \]
    for some constant $C'>0$ by Lemma~\ref{lem:norm_preservation}. In the above all constants depend only on $\beta,J_0,J,r_0,D$, and  $J_0,J>0$  defined in \eqref{eq:LR_velocity_full_ham} and \eqref{eq:LR_velocity_non_interacting} respectively depend only on $r_0,D$.
\end{proof}

\begin{lem}
    \label{lem:H_LR_1_decay}
    $H_{LR,1}$ as defined in \eqref{eq:H_LR_decompose} has $(CU,\mu)$-decay as defined in Definition~\ref{defn:decay_operator}, for some constants $C,\mu>0$ that depend only on $\beta,r_0,D$.
\end{lem}

\begin{proof}
    We can write
    \[
    H_{LR,1} = \sum_j W_j,
    \]
    where 
    \[
    W_j = \int \dd\omega\eta(\omega) Y_j(\omega),
    \]
    in which $Y_j(\omega)$ is defined in Lemma~\ref{lem:H_LR_1_term_approx}. We then define
    \[
    W_{j,r} = \int \dd\omega\eta(\omega) Y_{j,r}(\omega),
    \]
    where $Y_{j,r}(\omega)$ is also defined in Lemma~\ref{lem:H_LR_1_term_approx}. Note that $W_{j,r}$ is supported on the ball with radius $r$ centered at $j$.
    By Lemma~\ref{lem:H_LR_1_term_approx}, we have
    \[
    \|W_j-W_{j,r}\|\leq \int \dd\omega\eta(\omega)\|Y_j(\omega)-Y_{j,r}(\omega)\|\leq C_1U e^{-\mu r},
    \]
    and $\|W_j\|\leq C_2 U$. Therefore $W_j$ is $(CU/2,\mu)$-quasi-local around $j$ by Definition~\ref{defn:quasi_local}. By Lemma~\ref{lem:quasi_local_to_decay}, we have that $H_{LR,1} = \sum_j W_j$ has $(CU,\mu)$-decay.
    In the above, $C,C_1,C_2,\mu>0$ are all constants that depend only on $\beta,J_0,J,r_0,D$. Since  $J_0,J>0$ in \eqref{eq:LR_velocity_full_ham} and \eqref{eq:LR_velocity_non_interacting} respectively depend only on $r_0,D$, all constants only depend on $\beta,r_0,D$.
\end{proof}

Repeating this procedure for $H_{LR,2}$ and $H_{LR,3}$ allows us to show that $H_{LR}$ has $(3C,\mu)$-decay. We then repeat it form $H_{LL}$ and $H_{RR}$ in \eqref{eq:H_D_int_decompose}, which leads to the following result:
\begin{lem}
    \label{lem:H_dissipative_int_decay}
    $H^{\mathrm{parent}}_{\mathrm{D,int}}$ as defined in \eqref{eq:parent_ham_dissipative_int} has $(CU,\mu)$-decay as defined in Definition~\ref{defn:decay_operator}, for some constants $C,\mu>0$ that depend only on $\beta,r_0,D$.
\end{lem}

\section{Structure of the coherent terms}
\label{sec:coherent_term}
In this section we show the coherent terms are quasi-local. The coherent term is given by
\begin{equation}
B=\sum_{j \in A} B_j=\sum_{j \in A} \int_{-\infty}^{\infty} b_1(t) \int_{-\infty}^{\infty} b_2\left(t^{\prime}\right) \mathrm{e}^{\mathrm{i} \beta H (t^{\prime}-t)} \gamma_j \mathrm{e}^{-2 \mathrm{i} \beta H t^{\prime}} \gamma_j \mathrm{e}^{\mathrm{i} \beta H (t^{\prime}+t)} \mathrm{d} t^{\prime} \mathrm{~d} t
\end{equation}
where the functions $b_1(t)$ and $b_2(t')$ are given in \eqref{eq:f(t),b1(t),b2(t)}. We can decompose any operator $A\in B(\mathcal{H})$ in the energy basis and regroup in Bohr frequency:
\begin{equation}
A=\sum_{E_1, E_2 \in \operatorname{spec}(H)} P_{E_2} A P_{E_1}=\sum_{\nu \in \mathrm{Bohr}(H)} \sum_{E_2-E_1=\nu} P_{E_2} A P_{E_1}=: \sum_{\nu \in \mathrm{Bohr}(H)} A_\nu
\end{equation}
where $P_{E_1}$ and $P_{E_2}$ are the projectors on the eigenstates with eigenvalues $E_1$ and $E_2$ of $H$. The Heisenberg evolution operator $A(t)$ can be expressed as
\begin{equation}
    A(t)=\mathrm{e}^{i H t}A\mathrm{e}^{-i H t}=\sum_\nu \mathrm{e}^{i \nu t}A_\nu
\end{equation}
Using the decomposition, the coherent term can be expressed as 
\begin{equation}
\label{eq:dec}
\begin{aligned}
B_j&= \int_{-\infty}^{\infty} b_1(t)\mathrm{~d} t \int_{-\infty}^{\infty} b_2\left(t^{\prime}\right)\mathrm{d} t^{\prime}  \gamma_j(\beta(t'-t)) \gamma_j(-\beta(t+t'))   \\
&=\int_{-\infty}^{\infty} b_1(t)\mathrm{~d} t \int_{-\infty}^{\infty} b_2\left(t^{\prime}\right)\mathrm{d} t^{\prime}  \sum_\nu \mathrm{e}^{i \beta\nu (t'-t)} (\gamma_j)_\nu \sum_{\nu'} \mathrm{e}^{-i \beta\nu' (t+t')} (\gamma_j)_{\nu'}\\
&=\sum_{\nu,\nu'} (\gamma_j)_\nu (\gamma_j)_{\nu'}\int_{-\infty}^{\infty} b_1(t)\mathrm{~d} t \int_{-\infty}^{\infty} b_2\left(t^{\prime}\right)\mathrm{d} t^{\prime}\mathrm{e}^{i \beta\nu (t'-t)}\mathrm{e}^{-i \beta\nu' (t+t')}
\end{aligned}
\end{equation}
As we show in \eqref{eq:A_tilde_and_B_tilde}, the parent Hamiltonian has as its component $\tilde{B}_j=\sigma^{1/4}B_j\sigma^{-1/4}$. Following the decomposition in \eqref{eq:dec}, we have
\begin{equation}
    \tilde{B}_j=\sum_{\nu,\nu'} (\gamma_j)_\nu (\gamma_j)_{\nu'} \mathrm{e}^{-\beta \nu/4} \mathrm{e}^{\beta \nu'/4}\int_{-\infty}^{\infty} b_1(t)\mathrm{~d} t \int_{-\infty}^{\infty} b_2\left(t^{\prime}\right)\mathrm{d} t^{\prime}\mathrm{e}^{i \beta\nu (t'-t)}\mathrm{e}^{-i \beta\nu' (t+t')}
\end{equation}
We define $F_2(\nu,\nu')$ such that $\tilde{B}_j=\sum_{\nu,\nu'}F_2(\beta\nu,\beta\nu')(\gamma_j)_\nu (\gamma_j)_{\nu'}$:
\begin{equation}
\begin{aligned}
F_2(\nu,\nu')&= \mathrm{e}^{ \nu/4} \mathrm{e}^{- \nu'/4}\int_{-\infty}^{\infty} b_1(t)\mathrm{~d} t \int_{-\infty}^{\infty} b_2\left(t^{\prime}\right)\mathrm{d} t^{\prime}\mathrm{e}^{i \nu (t'-t)}\mathrm{e}^{-i \nu' (t+t')}   \\
&=\mathrm{e}^{ (\nu-\nu')/4} \int_{-\infty}^{\infty} b_1(t) \mathrm{e}^{i (\nu'-\nu) t}\mathrm{~d} t \int_{-\infty}^{\infty} b_2\left(t^{\prime}\right)\mathrm{e}^{-i (\nu+\nu') t}\mathrm{d} t^{\prime}\\
&=2\pi \mathrm{e}^{ (\nu-\nu')/4} \hat{b}_1\left( \nu-\nu'\right) \hat{b}_2\left( \nu'+\nu\right)
\end{aligned}
\end{equation}
We define the function $\check{F}_2(t,t')$, which we call the \emph{kernel function}, to be the inverse Fourier transformation of $F_2(\nu,\nu')$:
\begin{equation}
\label{eq:checkF_2(t,t')_defn}
    \check{F}_2(t,t')=\frac{1}{2\pi}\int_{-\infty}^{\infty}\mathrm{d} t\int_{-\infty}^{\infty}\mathrm{d} t^{\prime}F_2(\nu,\nu') \mathrm{e}^{i \nu t}\mathrm{e}^{i \nu' t}
\end{equation}
so we have
\begin{equation}
\begin{aligned}
\tilde{B}_j&=\frac{1}{2\pi}\sum_{\nu,\nu'}(\gamma_j)_\nu (\gamma_j)_{\nu'}\int_{-\infty}^{\infty}\mathrm{d} t\int_{-\infty}^{\infty}\mathrm{d} t^{\prime}\check{F}_2(t,t') \mathrm{e}^{-i \beta\nu t}\mathrm{e}^{-i \beta\nu' t}\\
&=\frac{1}{2\pi}\int_{-\infty}^{\infty}\mathrm{d} t\int_{-\infty}^{\infty}\mathrm{d} t^{\prime}\check{F}_2(t,t') \gamma_j(\beta t)\gamma_j(-\beta t')
\end{aligned}
\end{equation}

Below we will show that the kernel function $\check{F}_2(t,t')$ decays rapidly as $|t|,|t'|$ increases.
\begin{lem}
    \label{lem:decay_of_checkF2(t,t')}
    For the kernel function $\check{F}_2(t,t')$ defined in \eqref{eq:checkF_2(t,t')_defn}, we have
    \[
    |\check{F}_2(t,t')|\leq Ce^{-\lambda (t^2+t'^2)^{1/2}},
    \]
    for some universal constants $C,\lambda>0$.
\end{lem}
For the proof see Section~\ref{sec:decay_kernel_function_coherent_part} in the appendix.
Recall that in \eqref{eq:parent_ham_decompose_4_parts} we decomposed the coherent part of the Hamiltonian into $H^{\mathrm{parent}}_{\mathrm{C}}=H^{\mathrm{parent}}_{\mathrm{C,free}}+H^{\mathrm{parent}}_{\mathrm{C,int}}$. In Lemma~\ref{lem:free_parent_zero}. 
Therefore we only need to show that $H^{\mathrm{parent}}_{\mathrm{C,int}}$ has $(CU,\mu)$-decay for some $C,\mu>0$. 

For each $j$, we decompose
\begin{equation}
    \label{eq:defn_B_tilde_int}
    \tilde{B}^{\mathrm{int}}_j = \tilde{B}_j-\tilde{B}^{\mathrm{free}}_j,
\end{equation}
where
\begin{equation}
    \label{eq:defn_B_tilde_free}
    \tilde{B}^{\mathrm{free}}_j= \frac{1}{2\pi}\int_{-\infty}^{\infty}\mathrm{d} t\int_{-\infty}^{\infty}\mathrm{d} t^{\prime}\check{F}_2(t,t') \gamma^{\mathrm{free}}_j(\beta t)\gamma^{\mathrm{free}}_j(-\beta t')
\end{equation}
Therefore we have
\begin{equation}
\tilde{B}^{\mathrm{int}}_j=\frac{1}{2\pi}\int_{-\infty}^{\infty}\mathrm{d} t\int_{-\infty}^{\infty}\mathrm{d} t^{\prime}\check{F}(t,t') \left(\gamma_j(\beta t)\gamma_j(-\beta t')-\gamma_j^{\mathrm{free}}(\beta t)\gamma_j^{\mathrm{free}}(-\beta t')\right).
\end{equation}

We will first show that $\tilde{B}^{\mathrm{int}}_j$ can be well-approximated by an integral on the ball with radius $T$.
\begin{lem}
\label{lem:approximate_B_tilde_int_finite_T}
    For any $T>0$, let
    \[
    Z(T) = \frac{1}{2\pi}\int_{t^2+t'^2\leq T^2}\mathrm{d} t\mathrm{d} t^{\prime}\check{F}(t,t') \left(\gamma_j(\beta t)\gamma_j(-\beta t')-\gamma_j^{\mathrm{free}}(\beta t)\gamma_j^{\mathrm{free}}(-\beta t')\right).
    \]
    Then
    \[
    \|\tilde{B}^{\mathrm{int}}_j-Z(T)\|\leq C'Ue^{-\lambda'T},
    \]
    for some constants $C',\lambda'>0$ that depend only on $r_0,D,\beta$.
\end{lem}

\begin{proof}
    We first note that 
    \[
    \begin{aligned}
        &\gamma_j(\beta t)\gamma_j(-\beta t')-\gamma_j^{\mathrm{free}}(\beta t)\gamma_j^{\mathrm{free}}(-\beta t') \\
        &= \gamma_j^{\mathrm{int}}(\beta t)\gamma_j^{\mathrm{free}}(-\beta t')+\gamma_j^{\mathrm{free}}(\beta t)\gamma_j^{\mathrm{int}}(-\beta t')+\gamma_j^{\mathrm{int}}(\beta t)\gamma_j^{\mathrm{int}}(-\beta t').
    \end{aligned}
    \]
    Therefore because $\|\gamma_j^{\mathrm{free}}(s)\|\leq 1$, we have
    \[
    \|\gamma_j(\beta t)\gamma_j(-\beta t')-\gamma_j^{\mathrm{free}}(\beta t)\gamma_j^{\mathrm{free}}(-\beta t')\|\leq \|\gamma_j^{\mathrm{int}}(\beta t)\| +\|\gamma_j^{\mathrm{int}}(-\beta t')\| + \|\gamma_j^{\mathrm{int}}(\beta t)\|\|\gamma_j^{\mathrm{int}}(-\beta t')\|.
    \]
    By Corollary~\ref{cor:approx_gamma_int} we have
    \[
    \|\gamma_j^{\mathrm{int}}(\beta t)\|, \|\gamma_j^{\mathrm{int}}(-\beta t')\|\leq C_1 U(|t|+1)^{D+1}.
    \]
    Therefore
    \[
    \begin{aligned}
        &\|\gamma_j(\beta t)\gamma_j(-\beta t')-\gamma_j^{\mathrm{free}}(\beta t)\gamma_j^{\mathrm{free}}(-\beta t')\| \\
        &\leq C_1 U(|t|+1)^{D+1}+C_1 U(|t'|+1)^{D+1} + C_1^2 U^2(|t|+1)^{D+1}(|t'|+1)^{D+1}.
    \end{aligned}
    \]

    By Lemma~\ref{lem:decay_of_checkF2(t,t')} we have $|\check{F}(t,t')|\leq Ce^{-\lambda (t^2+t'^2)^{1/2}}$. Therefore
    \begin{equation}
        \begin{aligned}
            &\|\tilde{B}^{\mathrm{int}}_j-Z(T)\| \\
            &\leq \int_{t^2+t'^2> T^2}\mathrm{d} t\mathrm{d} t^{\prime} |\check{F}(t,t')| \big(C_1 U(|t|+1)^{D+1}+C_1 U(|t'|+1)^{D+1} \\
            &\quad\quad\quad\quad\quad\quad  +C_1^2 U^2(|t|+1)^{D+1}(|t'|+1)^{D+1}\big) \\
            &\leq C\int_{t^2+t'^2> T^2}\mathrm{d} t\mathrm{d} t^{\prime} e^{-\lambda (t^2+t'^2)^{1/2}} \big(C_1 U(|t|+1)^{D+1}+C_1 U(|t'|+1)^{D+1} \\
            &\quad\quad\quad\quad\quad\quad+ C_1^2 U^2(|t|+1)^{D+1}(|t'|+1)^{D+1}\big) \\
            &\leq C'U e^{-\lambda' T},
        \end{aligned}
    \end{equation}
    for $U>0$ below any universal constant. $\lambda'>0$ is a constant that depends only on $D$, and $C'>0$ is a constant that depends only on $r_0,D,\beta$.
\end{proof}

We decompose further the $Z(t)$ introduced in Lemma~\ref{lem:approximate_B_tilde_int_finite_T} into three parts, similar to \eqref{eq:H_LR_decompose},
\begin{equation}
    \label{eq:decompose_Z(T)_three_parts}
    Z(T) = Z_1(T)+Z_2(T)+Z_3(T),
\end{equation}
where,
\begin{equation}
\label{eq:three_parts_of_Z(T)}
    \begin{aligned}
        Z_1(T)&=\frac{1}{2\pi}\int_{t^2+t'^2\leq T^2}\mathrm{d} t\mathrm{d} t^{\prime}\check{F}(t,t') \gamma_j^{\mathrm{int}}(\beta t)\gamma_j^{\mathrm{int}}(-\beta t')\\
        Z_2(T)&=\frac{1}{2\pi}\int_{t^2+t'^2\leq T^2}\mathrm{d} t\mathrm{d} t^{\prime}\check{F}(t,t') \gamma_j^{\mathrm{int}}(\beta t)\gamma_j^{\mathrm{free}}(-\beta t')\\
        Z_3(T)&=\frac{1}{2\pi}\int_{t^2+t'^2\leq T^2}\mathrm{d} t\mathrm{d} t^{\prime}\check{F}(t,t') \gamma_j^{\mathrm{free}}(\beta t)\gamma_j^{\mathrm{int}}(-\beta t')
    \end{aligned}
\end{equation}

We will show that each of the terms can be well-approximated by local operators. We will start with $Z_1(T)$.
\begin{lem}
    \label{lem:approx_Z_1(T)}
    Let $X_1(t)$ and $X_2(s,t)$ be as defined in Lemma~\ref{lem:commutator_approx_V_gamma_free} and Lemma~\ref{lem:X2(s,t)} respectively. We define
    \[
    Z^{\mathrm{loc}}_{1}(T) = \frac{1}{2\pi}\int_{t^2+t'^2\leq T^2}\mathrm{d} t\mathrm{d} t^{\prime}\check{F}(t,t') \int_0^{\beta t}X_2(s,t)\dd s X_1(-\beta t'),
    \]
    which is supported on the ball with radius $r$ centered at $j$.
    Then
    \begin{equation}
    \label{eq:approx_Z_1(T)}
        \|Z^{\mathrm{loc}}_{1}(T)-Z_1(T)\|\leq CU (T+1)^{D+1}(r+1)^D e^{-r/2},
    \end{equation}
    and
    \begin{equation}
        \label{eq:norm_bound_Z_1(T)}
        \|Z_1(T)\|\leq C'U,
    \end{equation}
    when $r\geq C_1(T+1)$. Moreover, when setting $T=r/C_1-1$, we have
    \begin{equation}
    \label{eq:approx_Z_1(T)_simplified}
        \|Z^{\mathrm{loc}}_{1}(T)-Z_1(T)\|\leq C_1'U e^{-\mu r}.
    \end{equation}
    Here $C,C_1,C',C_1',\mu>0$ are all constants that only depend on $r_0,D,\beta$.
\end{lem}

\begin{proof}
    We first recall that by Corollary~\ref{cor:approx_gamma_int},
    \[
    \left\|\gamma^{\mathrm{int}}_j(\beta t)-\int_0^{\beta t}X_2(s,t)\dd s\right\|\leq C_2 e^{-r/2}U(|t|+1)^{D+1},
    \]
    and
    \[
    \left\|\int_0^{\beta t}X_2(s,t)\dd s\right\|\leq C_3 U(|t|+1)^{D+1}.
    \]
    By the proof of Lemma~\ref{lem:commutator_approx_V_gamma_free}, we also have
    \[
    \|\gamma^{\mathrm{free}}_j(-\beta t')-X_1(-\beta t')\|\leq C_4 (r+1)^D e^{-r}.
    \]
    All of the above require $r\geq J_1(|t|+1)$. These results allow us to approximate $\gamma_j^{\mathrm{int}}(\beta t)\gamma_j^{\mathrm{int}}(-\beta t')$ with $\int_0^{\beta t}X_2(s,t)\dd s X_1(-\beta t')$. We have
    \[
    \begin{aligned}
        &\left\|\gamma_j^{\mathrm{int}}(\beta t)\gamma_j^{\mathrm{int}}(-\beta t')-\int_0^{\beta t}X_2(s,t)\dd s X_1(-\beta t)\right\| \\
        &\leq \left\|\gamma^{\mathrm{int}}_j(\beta t)-\int_0^{\beta t}X_2(s,t)\dd s\right\|\|\gamma^{\mathrm{free}}_j(-\beta t')\| 
        +\left\|\int_0^{\beta t}X_2(s,t)\dd s\right\|\|\gamma^{\mathrm{free}}_j(-\beta t')-X_1(-\beta t')\|\\
        &\leq C_5 U (T+1)^{D+1}(r+1)^D e^{-r/2},
    \end{aligned}
    \]
    when $t^2+t'^2\leq T^2$. Using the fact that $\int\dd t\dd t' |\check{F}_2(t,t')|$ is upper bounded by a constant, we can integrate the two sides of the above inequality to get \eqref{eq:approx_Z_1(T)}.  \eqref{eq:approx_Z_1(T)_simplified} follows by setting $T=r/J_1-1$.

    An upper bound for the norm $\|Z_1(T)\|$ can be obtained from the norm bound for $\gamma^{\mathrm{int}}_j(t)$ (available from Corollary~\ref{cor:approx_gamma_int} by the triangle inequality) and $\|\gamma^{\mathrm{int}}_j(-\beta t')\|\leq 1$, which after multiplying to an exponentially decaying kernel function $\check{F}_2(t,t')$ and integrating over $t,t'$ yield \eqref{eq:norm_bound_Z_1(T)}.
\end{proof}

This procedure can be repeated for $Z_2(T)$ and $Z_3(T)$, yielding
\begin{equation}
\label{eq:approx_Z2_Z3}
    \|Z_2(T)-Z^{\mathrm{loc}}_2(T)\|\leq C_2 U e^{-\mu_2 r},\quad \|Z_3(T)-Z^{\mathrm{loc}}_3(T)\|\leq C_3 U^2 e^{-\mu_3 r},
\end{equation}
and
\begin{equation}
    \label{eq:norm_Z2_Z3}
    \|Z_2(T)\|\leq C_2'U, \quad\|Z_3(T)\|\leq C_3' U_2.
\end{equation}
Here $Z^{\mathrm{loc}}_2(T)$ and $Z^{\mathrm{loc}}_3(T)$ are both supported on the ball with radius $r$ centered at $j$. By adding up the three approximations to get
\begin{equation}
\label{eq:defn_Z_loc_T}
    Z^{\mathrm{loc}}(T) = Z^{\mathrm{loc}}_1(T)+Z^{\mathrm{loc}}_2(T)+Z^{\mathrm{loc}}_3(T),
\end{equation}
we then have
\begin{equation}
\label{eq:approx_Z(T)_and_Z(T)_norm}
    \|Z(T)-Z^{\mathrm{loc}}(T)\|\leq CUe^{-\mu r},\quad \|Z(T)\|\leq C'U.
\end{equation}
In the above $C,C',C_2,C_2',C_3,C_3'$ are all constants that depend only on $r_0,D$.
We are now ready to approximate $\tilde{B}^{\mathrm{int}}_j$.

\begin{lem}
    \label{lem:approx_coherent}
    The term $\tilde{B}^{\mathrm{int}}_j$ can be approximated with a local operator $\tilde{B}^{\mathrm{loc}}_j(r)$. More precisely, we have
    \[
    \left\|\tilde{B}^{\mathrm{int}}_j-\tilde{B}^{\mathrm{loc}}_j(r)\right\|\leq C e^{-\mu r} U,
    \]
    Here $\tilde{B}^{\mathrm{loc}}_j(r)$ is supported on the ball with radius $r$ centered at $j$. Moreover,
    \[
    \left\| \tilde{B}_j \right\|\leq C' U.
    \]
    $C,C',\mu$ are positive constants that depend only on $r_0,D,\beta$.
\end{lem}

\begin{proof}
    We define $\tilde{B}^{\mathrm{loc}}_j(r) = Z^{\mathrm{loc}}(r/J_1-1)$ where $J_1$ is the same as in Lemma~\ref{lem:approx_Z_1(T)} and $Z^{\mathrm{loc}}(T)$ is defined in \eqref{eq:defn_Z_loc_T}.
    Combining \eqref{eq:approx_Z(T)_and_Z(T)_norm} with Lemma~\ref{lem:approximate_B_tilde_int_finite_T} we have the bound for $\left\|\tilde{B}^{\mathrm{int}}_j-\tilde{B}^{\mathrm{loc}}_j(r)\right\|$ by the triangle inequality. The bound for $\left\| \tilde{B}_j \right\|$ can be derived by the first bound and $\|Z(T)\|\leq C_1 U$ in \eqref{eq:approx_Z(T)_and_Z(T)_norm}.
\end{proof}

\begin{cor}
    \label{cor:quasi_locaity_H_int_C}
    The part of the parent Hamiltonian $H^{\mathrm{parent}}_{\mathrm{C,int}}$ defined in \eqref{eq:parent_ham_decompose_4_parts} has $(CU,\mu)$-decay as defined in Definition~\ref{defn:decay_operator}, where $C,\mu>0$ depend only on $r_0,D,\beta$.
\end{cor}

\begin{proof}
    Since $\tilde{B}=\sum_j \tilde{B}_j$, and that $\tilde{B}^{\mathrm{int}}_j$ is $(C_1 U,\mu)$-quasi-local around $j$ (Definition~\ref{defn:quasi_local}) as proved in Lemma~\ref{lem:approx_coherent}, by Lemma~\ref{lem:quasi_local_to_decay} we know that $\tilde{B}$ has $(2C_1 U,\mu)$-decay. Because $H^{\mathrm{parent}}_{\mathrm{C,int}} = \Phi(L_{\tilde{B}_j}-R_{\tilde{B}_j^\dag})$ and that locality is preserved in the mapping $\Phi$ by Remark~\ref{rem:Phi_preserves_locality}, we have the $(CU,\mu)$-decay of $H^{\mathrm{parent}}_{\mathrm{C,int}}$.
\end{proof}

We are then ready to prove Proposition~\ref{prop:interacting_part_decay}:
\begin{proof}[Proof of Proposition~\ref{prop:interacting_part_decay}]
    By Lemma~\ref{lem:H_dissipative_int_decay} we know that  $H^{\mathrm{parent}}_{\mathrm{D,int}}$ has $(C_1 U,\mu_1)$-decay, and by Corollary~\ref{cor:quasi_locaity_H_int_C} we know that $H^{\mathrm{parent}}_{\mathrm{C,int}}$ has $(C_2 U,\mu_2)$-decay. As a result, by setting $C=C_1+C_2$, $\mu=\min\{\mu_1,\mu_2\}$, we can see that
    $H^{\mathrm{parent}}_{\mathrm{int}}=H^{\mathrm{parent}}_{\mathrm{C,int}}+H^{\mathrm{parent}}_{\mathrm{D,int}}$ has $(CU,\mu)$-decay.
\end{proof}

\section{Discussions}

In this work we proved that, at any temperature, when the interacting strength in a geometrically local fermionic Hamiltonian is below a constant threshold that depends only on the lattice geometry ($r_0,D$) and the inverse temperature ($\beta$), the mixing time of an efficiently implementable Lindbladian that prepares the Gibbs state can be bounded linearly in the system size. This immediately implies efficient preparation of the Gibbs state on a quantum computer. In the proof, we obtained a constant lower bound for the Lindbladian spectral gap restricted to the even parity subspace by mapping the Lindbladian to a parent Hamiltonian, for which the result from \cite{Hastings2019stability} provides a spectral gap lower bound. The mapping from the Lindbladian to the parent Hamiltonian utilizes the third quantization, and we paid special attention to ensure that the mapping preserves locality and the relevant part of the Lindbladian spectrum.

Our result motivates new questions about the competition between classical and quantum algorithms in the setting of weakly interacting fermions. The provably efficient quantum algorithm that comes as a result of our mixing time bound makes it natural to ask whether one can compute Gibbs state properties for weakly interacting fermions on a classical computer to attain arbitrarily high accuracy with provable efficiency. This may require us to use structural properties of the Gibbs states akin to what was done in \cite{BakshiLiuMoitraTang2024high}.

It is also natural to ask whether one can obtain spectral gap results and hence mixing time bounds by directly studying the Lindbladian rather than going through the parent Hamiltonian. If possible, this approach would enable us to deal with Lindbladians that may not satisfy the detailed balance condition. Such an approach has been studied for Lindbladian satisfying a special set of conditions in \cite{Fang2024mixing}, from which the idea of constructing a Lyapunov functional may be more generally useful. 

The quantum algorithm for Gibbs state preparation discussed in Corollary~\ref{cor:efficient_prep_gibbs_state} also motivates further development of Lindbladian simulation algorithms. Note that to prepare the Gibbs state we need to simulate Lindbladian dynamics for a time that is linear in $n$ on a system of size that is also linear in $n$. One would naturally expect the gate complexity to be proportional to $n^2$, but we in the end have a $n^3$ scaling because of an additional $n$ factor from normalizing the Lindbladian. This is in contrast with Hamiltonian simulation algorithms where a gate complexity that is almost linear in the space-time volume can be attained \cite{HaahHastingsKothariLow2021quantum,ChildsSuTranWiebeZhu2021theory}. This problem also arises in other Lindbladian simulation algorithms \cite{DingLiLin2024simulating,CleveWang2016efficient,Pocrnic2023quantum,ChildsLi2016efficient,Kliesch2011dissipative}. Whether one can attain the space-time volume scaling in Lindbladian simulation remains an interesting open question.


\section*{Acknowledgments}

We thank Lin Lin and Zhiyan Ding for helpful discussion.
Y.Z. acknowledges funding from the National Science Foundation (PHY-1733907). The Institute for Quantum Information and Matter is an NSF Physics Frontiers Center. In the process of writing this paper, the authors became aware of an independent work by Mario Berta, Roberto Bondesan, Richard Meister, and {\v S}t{\v e}p{\' a}n {\v S}m{\' i}d on a similar topic that has been carried out in parallel.

\bibliographystyle{ieeetr}
\bibliography{ref}

\begin{thebibliography}{10}

\bibitem{Landau2015polynomial}
Z.~Landau, U.~Vazirani, and T.~Vidick, ``A polynomial time algorithm for the ground state of one-dimensional gapped local {H}amiltonians,'' {\em Nature Physics}, vol.~11, no.~7, pp.~566--569, 2015.

\bibitem{BrandaoKastoryano2019finite}
F.~G. Brandao and M.~J. Kastoryano, ``Finite correlation length implies efficient preparation of quantum thermal states,'' {\em Communications in Mathematical Physics}, vol.~365, pp.~1--16, 2019.

\bibitem{HwangJiang2024gibbs}
Y.~Hwang and J.~Jiang, ``{G}ibbs state preparation for commuting {H}amiltonian: Mapping to classical {G}ibbs sampling,'' {\em arXiv preprint arXiv:2410.04909}, 2024.

\bibitem{KastoryanoBrandao2016quantum}
M.~J. Kastoryano and F.~G. Brandao, ``Quantum {G}ibbs samplers: The commuting case,'' {\em Communications in Mathematical Physics}, vol.~344, pp.~915--957, 2016.

\bibitem{KochanowskiAlhambraCapelRouze2024rapid}
J.~Kochanowski, A.~M. Alhambra, A.~Capel, and C.~Rouz{\'e}, ``Rapid thermalization of dissipative many-body dynamics of commuting {H}amiltonians,'' {\em arXiv preprint arXiv:2404.16780}, 2024.

\bibitem{BardetCapelLiEtAl2023rapid}
I.~Bardet, {\'A}.~Capel, L.~Gao, A.~Lucia, D.~P{\'e}rez-Garc{\'\i}a, and C.~Rouz{\'e}, ``Rapid thermalization of spin chain commuting {H}amiltonians,'' {\em Physical Review Letters}, vol.~130, no.~6, p.~060401, 2023.

\bibitem{BravyiGosset2017complexity}
S.~Bravyi and D.~Gosset, ``Complexity of quantum impurity problems,'' {\em Communications in Mathematical Physics}, vol.~356, pp.~451--500, 2017.

\bibitem{EhrenbergDeshpandeEtAl2022simulation}
A.~Ehrenberg, A.~Deshpande, C.~L. Baldwin, D.~A. Abanin, and A.~V. Gorshkov, ``Simulation complexity of many-body localized systems,'' {\em arXiv preprint arXiv:2205.12967}, 2022.

\bibitem{ChenKastoryanoBrandaoGilyen2023quantum}
C.-F. Chen, M.~J. Kastoryano, F.~G. Brand{\~a}o, and A.~Gily{\'e}n, ``Quantum thermal state preparation,'' {\em arXiv preprint arXiv:2303.18224}, 2023.

\bibitem{chen2023efficient}
C.-F. Chen, M.~J. Kastoryano, and A.~Gily{\'e}n, ``An efficient and exact noncommutative quantum {Gibbs} sampler,'' {\em arXiv preprint arXiv:2311.09207}, 2023.

\bibitem{temme2014hypercontractivity}
K.~Temme, F.~Pastawski, and M.~J. Kastoryano, ``Hypercontractivity of quasi-free quantum semigroups,'' {\em Journal of Physics A: Mathematical and Theoretical}, vol.~47, no.~40, p.~405303, 2014.

\bibitem{ding2024polynomial}
Z.~Ding, B.~Li, L.~Lin, and R.~Zhang, ``Polynomial-time preparation of low-temperature {G}ibbs states for 2d toric code,'' {\em arXiv preprint arXiv:2410.01206}, 2024.

\bibitem{bardet2024entropy}
I.~Bardet, {\'A}.~Capel, L.~Gao, A.~Lucia, D.~P{\'e}rez-Garc{\'\i}a, and C.~Rouz{\'e}, ``Entropy decay for davies semigroups of a one dimensional quantum lattice,'' {\em Communications in Mathematical Physics}, vol.~405, no.~2, p.~42, 2024.

\bibitem{alicki2009thermalization}
R.~Alicki, M.~Fannes, and M.~Horodecki, ``On thermalization in {K}itaev's 2d model,'' {\em Journal of Physics A: Mathematical and Theoretical}, vol.~42, no.~6, p.~065303, 2009.

\bibitem{ramkumar2024mixing}
A.~Ramkumar and M.~Soleimanifar, ``Mixing time of quantum {G}ibbs sampling for random sparse {H}amiltonians,'' {\em arXiv preprint arXiv:2411.04454}, 2024.

\bibitem{RouzeFrancaAlhambra2024optimal}
C.~Rouz{\'e}, D.~S. Fran{\c{c}}a, and {\'A}.~M. Alhambra, ``Optimal quantum algorithm for {G}ibbs state preparation,'' {\em arXiv preprint arXiv:2411.04885}, 2024.

\bibitem{capel2024quasi}
{\'A}.~Capel, P.~Gondolf, J.~Kochanowski, and C.~Rouz{\'e}, ``Quasi-optimal sampling from {G}ibbs states via non-commutative optimal transport metrics,'' {\em arXiv preprint arXiv:2412.01732}, 2024.

\bibitem{kastoryano2013quantum}
M.~J. Kastoryano and K.~Temme, ``Quantum logarithmic sobolev inequalities and rapid mixing,'' {\em Journal of Mathematical Physics}, vol.~54, no.~5, 2013.

\bibitem{temme2015fast}
K.~Temme and M.~J. Kastoryano, ``How fast do stabilizer {H}amiltonians thermalize?,'' {\em arXiv preprint arXiv:1505.07811}, 2015.

\bibitem{ChenLiLuYing2024randomized}
H.~Chen, B.~Li, J.~Lu, and L.~Ying, ``A randomized method for simulating {L}indblad equations and thermal state preparation,'' {\em arXiv preprint arXiv:2407.06594}, 2024.

\bibitem{Rakovszky2024bottlenecks}
T.~Rakovszky, B.~Placke, N.~P. Breuckmann, and V.~Khemani, ``Bottlenecks in quantum channels and finite temperature phases of matter,'' {\em arXiv preprint arXiv:2412.09598}, 2024.

\bibitem{LiLu2024quantum}
B.~Li and J.~Lu, ``Quantum space-time {P}oincare inequality for {L}indblad dynamics,'' {\em arXiv preprint arXiv:2406.09115}, 2024.

\bibitem{Barthel2022solving}
T.~Barthel and Y.~Zhang, ``Solving quasi-free and quadratic {L}indblad master equations for open fermionic and bosonic systems,'' {\em Journal of Statistical Mechanics: Theory and Experiment}, vol.~2022, no.~11, p.~113101, 2022.

\bibitem{Fang2024mixing}
D.~Fang, J.~Lu, and Y.~Tong, ``Mixing time of open quantum systems via hypocoercivity,'' {\em arXiv preprint arXiv:2404.11503}, 2024.

\bibitem{DingLiLin2024efficient}
Z.~Ding, B.~Li, and L.~Lin, ``Efficient quantum {G}ibbs samplers with {Kubo--Martin--Schwinger} detailed balance condition,'' {\em arXiv preprint arXiv:2404.05998}, 2024.

\bibitem{DingChenLin2024single}
Z.~Ding, C.-F. Chen, and L.~Lin, ``Single-ancilla ground state preparation via lindbladians,'' {\em Physical Review Research}, vol.~6, no.~3, p.~033147, 2024.

\bibitem{rouze2024efficient}
C.~Rouz{\'e}, D.~S. Fran{\c{c}}a, and {\'A}.~M. Alhambra, ``Efficient thermalization and universal quantum computing with quantum {Gibbs} samplers,'' {\em arXiv preprint arXiv:2403.12691}, 2024.

\bibitem{MichalakisZwolak2013stability}
S.~Michalakis and J.~P. Zwolak, ``Stability of frustration-free {H}amiltonians,'' {\em Communications in Mathematical Physics}, vol.~322, pp.~277--302, 2013.

\bibitem{BakshiLiuMoitraTang2024high}
A.~Bakshi, A.~Liu, A.~Moitra, and E.~Tang, ``High-temperature {G}ibbs states are unentangled and efficiently preparable,'' {\em arXiv preprint arXiv:2403.16850}, 2024.

\bibitem{Hastings2019stability}
M.~Hastings, ``The stability of free fermi {H}amiltonians,'' {\em Journal of Mathematical Physics}, vol.~60, no.~4, 2019.

\bibitem{prosen2008third}
T.~Prosen, ``Third quantization: a general method to solve master equations for quadratic open fermi systems,'' {\em New Journal of Physics}, vol.~10, no.~4, p.~043026, 2008.

\bibitem{LiZhanLin2024dissipative}
H.-E. Li, Y.~Zhan, and L.~Lin, ``Dissipative ground state preparation in ab initio electronic structure theory,'' {\em arXiv preprint arXiv:2411.01470}, 2024.

\bibitem{LiebRobinson1972finite}
E.~H. Lieb and D.~W. Robinson, ``The finite group velocity of quantum spin systems,'' {\em Communications in Mathematical Physics}, vol.~28, no.~3, pp.~251--257, 1972.

\bibitem{Hastings2004decay}
M.~B. Hastings, ``Decay of correlations in fermi systems at nonzero temperature,'' {\em Physical review letters}, vol.~93, no.~12, p.~126402, 2004.

\bibitem{HaahHastingsKothariLow2021quantum}
J.~Haah, M.~B. Hastings, R.~Kothari, and G.~H. Low, ``Quantum algorithm for simulating real time evolution of lattice {H}amiltonians,'' {\em SIAM Journal on Computing}, vol.~52, no.~6, pp.~FOCS18--250, 2021.

\bibitem{CapelMoscolariTeufelWessel2023decay}
{\'A}.~Capel, M.~Moscolari, S.~Teufel, and T.~Wessel, ``From decay of correlations to locality and stability of the gibbs state,'' {\em arXiv preprint arXiv:2310.09182}, 2023.

\bibitem{BluhmCapelEtAl2022exponential}
A.~Bluhm, {\'A}.~Capel, and A.~P{\'e}rez-Hern{\'a}ndez, ``Exponential decay of mutual information for gibbs states of local hamiltonians,'' {\em Quantum}, vol.~6, p.~650, 2022.

\bibitem{KuwaharaKatoBrandao2020clustering}
T.~Kuwahara, K.~Kato, and F.~G. Brand{\~a}o, ``Clustering of conditional mutual information for quantum gibbs states above a threshold temperature,'' {\em Physical Review Letters}, vol.~124, no.~22, p.~220601, 2020.

\bibitem{Hastings2004lieb}
M.~B. Hastings, ``{Lieb-Schultz-Mattis} in higher dimensions,'' {\em Physical Review B}, vol.~69, no.~10, p.~104431, 2004.

\bibitem{nachtergaele2006lieb}
B.~Nachtergaele and R.~Sims, ``{Lieb-Robinson} bounds and the exponential clustering theorem,'' {\em Communications in Mathematical Physics}, vol.~265, pp.~119--130, 2006.

\bibitem{Hastings2006spectral}
M.~B. Hastings and T.~Koma, ``Spectral gap and exponential decay of correlations,'' {\em Communications in Mathematical Physics}, vol.~265, pp.~781--804, 2006.

\bibitem{Hastings2010locality}
M.~B. Hastings, ``Locality in quantum systems,'' {\em Quantum Theory from Small to Large Scales}, vol.~95, pp.~171--212, 2010.

\bibitem{ChildsSuTranWiebeZhu2021theory}
A.~M. Childs, Y.~Su, M.~C. Tran, N.~Wiebe, and S.~Zhu, ``Theory of trotter error with commutator scaling,'' {\em Physical Review X}, vol.~11, no.~1, p.~011020, 2021.

\bibitem{DingLiLin2024simulating}
Z.~Ding, X.~Li, and L.~Lin, ``Simulating open quantum systems using {H}amiltonian simulations,'' {\em PRX Quantum}, vol.~5, no.~2, p.~020332, 2024.

\bibitem{CleveWang2016efficient}
R.~Cleve and C.~Wang, ``Efficient quantum algorithms for simulating {L}indblad evolution,'' {\em arXiv preprint arXiv:1612.09512}, 2016.

\bibitem{Pocrnic2023quantum}
M.~Pocrnic, D.~Segal, and N.~Wiebe, ``Quantum simulation of {L}indbladian dynamics via repeated interactions,'' {\em arXiv preprint arXiv:2312.05371}, 2023.

\bibitem{ChildsLi2016efficient}
A.~M. Childs and T.~Li, ``Efficient simulation of sparse {M}arkovian quantum dynamics,'' {\em arXiv preprint arXiv:1611.05543}, 2016.

\bibitem{Kliesch2011dissipative}
M.~Kliesch, T.~Barthel, C.~Gogolin, M.~Kastoryano, and J.~Eisert, ``Dissipative quantum {C}hurch-{T}uring theorem,'' {\em Physical Review Letters}, vol.~107, no.~12, p.~120501, 2011.

\end{thebibliography}

\clearpage
\appendix

\section{Table of symbols}
\label{sec:table_of_symbols}
\begin{center}
  \begin{threeparttable}
  \begin{tabular}{cl}
    \toprule
     Symbol &  Meaning \\
    \midrule
    $\Lambda$ & The physical lattice whose sites each contain a Dirac fermion \\
    \midrule
    $D$ & The physical lattice dimension \\
    \midrule
    $d(x,y)$ & The Euclidean distance between lattice sites $x$ and $y$ \\
    \midrule
    $\gamma_j$ & The $j$th Majorana operator \\
    \midrule
    $\mathfrak{s}(j)$ & The lattice site that contains the $j$th Majorana mode \\
    \midrule
    $\hat{\gamma}_j$ & The $j$th third quantized $a$-Majorana operator \\
    \midrule
    $\mathfrak{s}'(j)$ & The lattice site that contains the $j$th $a$-Majorana mode \\
    \midrule
    $\mathcal{H}$ & The Hilbert space of $n$ Dirac fermions \\
    \midrule
    $B(\mathcal{H})$ & The operator algebra on $\mathcal{H}$ \\
    \midrule
    $\mathcal{H}'$ & The enlarged Hilbert space of $2n$ Dirac fermions \\
    \midrule
    $B(B(\mathcal{H}))$ & The operator algebra on $B(\mathcal{H})$ (or the algebra of super-operators) \\
    \midrule
    $\mathcal{H}_{\mathrm{even}/\mathrm{odd}}$ & The even/odd-parity sector of $\mathcal{H}$ (Definition~\ref{defn:parity_of_states})\\
    \midrule
    $B(\mathcal{H})_{\mathrm{even}/\mathrm{odd}}$ & The subspace space spanned by Majorana monomials of even/odd degree (Definition~\ref{defn:parity_of_operators})\\
    \midrule
    $H$ & The fermionic Hamiltonian \eqref{eq:general_hamiltonian}\\
    \midrule
    $\beta$ & The inverse temperature \\
    \midrule
    $\sigma$ & The Gibbs state $\sigma\propto e^{-\beta H}$ \\
    \midrule
    $H_0$ & The non-interacting (quadratic or free) part of $H$ \\
    \midrule
    $V$ & The interacting part of $H$ \\
    \midrule
    $r_0$ & The maximal range at which fermions can interact \\
    \midrule
    $U$ & The interaction strength \\
    \midrule
    $J$ & The Lieb-Robinson velocity for $H$ \eqref{eq:LR_velocity_full_ham}\\
    \midrule
    $J_0$ & The Lieb-Robinson velocity for $H_0$ \eqref{eq:LR_velocity_non_interacting} \\
    \midrule
    $\mathcal{L}$ & The Lindbladian we study \eqref{eq:QMC_lindbladian} \\
    \midrule
    $\varphi$ & The isometry from $B(\mathcal{H})$ to $\mathcal{H}'$ used for third quantization \eqref{eq:third_quantization_map} \\
    \midrule
    $\Phi$ & The third quantization map that preserves locality (Definition~\ref{defn:Phi_thrid_quantization}) \\
    \midrule
    $\tilde{\Phi}$ & The map $B(B(\mathcal{H}))\to B(\mathcal{H}')$ that incorporates $\sigma$ \eqref{eq:defn_Phi_tilde_parent_ham_map} \\
    \midrule
    $L_A$/$R_A$ & Left/right multiplication with $A$ \eqref{eq:left_and_right_multiplication}\\
    \midrule
    $H^{\mathrm{parent}}$ & The parent Hamiltonian \eqref{eq:parent_hamiltonian_defn} (explicit form given in \eqref{eq:parent_hamiltonian_explicit}) \\
    \midrule
    $H^{\mathrm{parent}}_{0}$ & The non-interacting part of the parent Hamiltonian \eqref{eq:H_0_parent_and_V_defn}\\
    \midrule
    $V^{\mathrm{parent}}$ & The interacting part of the parent Hamiltonian \eqref{eq:H_0_parent_and_V_defn}\\
    \bottomrule
  \end{tabular}
  \end{threeparttable}
\end{center}

\section{How \emph{not} to vectorize fermionic super-operators}
\label{sec:how_not_to_vectorize}

A crucial step used in \cite{chen2023efficient} is to ``vectorize'' a super-operator $\mathcal{C}:X\mapsto AXB$ into an operator $\mathrm{vec}(\mathcal{C})=A\otimes B^\top$. 
Now we apply it to  fermionic super-operators
\[
\mathcal{C}_1:X\mapsto \gamma_1 X, \quad \mathcal{C}_2:X\mapsto X\gamma_2,
\]
They are then vectorized to be
\[
\mathrm{vec}(\mathcal{C}_1) = \gamma_1\otimes I,\quad \mathrm{vec}(\mathcal{C}_2) = I\otimes \gamma_2.
\]
First for the operation $\top$, we can interpret it as the linear involution in the Clifford algebra (as opposed to the anti-linear involution $\dag$).

We then consider how one should interpret the tensor product $\otimes$. One may be tempted to equate $\gamma_1\otimes I$ with $\gamma_1$, and $I\otimes \gamma_2$ with $\gamma_{n+2}$, which is the ``image'' of $\gamma_2$ in an ancillary system. Adopting this approach, we have
\[
\mathrm{vec}(\mathcal{C}_1) = \gamma_1,\quad \mathrm{vec}(\mathcal{C}_2) = \gamma_{n+2}.
\]
Now we have a problem: while $\mathcal{C}_1$ and $\mathcal{C}_2$ commute with each other, $\mathrm{vec}(\mathcal{C}_1)$ and $\mathrm{vec}(\mathcal{C}_2)$ do not; in fact they anti-commute with each other. Therefore this naive vectorization approach fails to preserve the algebraic relations between fermionic super-operators.

With a small modification one can make $\mathrm{vec}(\mathcal{C}_1)$ and $\mathrm{vec}(\mathcal{C}_2)$ commute with each other. Instead of the what was done above, we define
\[
\mathrm{vec}(\mathcal{C}_1) = \gamma_1,\quad \mathrm{vec}(\mathcal{C}_2) = \gamma_{n+2}(-1)^N.
\]
where $N$ is the particle number operator and $(-1)^N$ is the parity operator. This is in fact one of the insights of the third quantization approach \cite{prosen2008third}. However, we can see that this mapping maps an operator that is very local ($\mathcal{C}_2$) to something extremely non-local ($\gamma_{n+2}(-1)^N$ acts non-trivially on every mode except for mode $n+2$). This will prevent us from applying the result in \cite{Hastings2019stability}.

The approach used in \cite{prosen2008third} and in this work is based on the observation that $(-1)^N$ yields a factor of $1$ when restricted to the even-parity sector. However, this does not mean we can treat $(-1)^N$ as if it does not exist: by applying single Majorana operators one will switch between the even- and odd-parity sectors. This is why a great deal of effort of Section~\ref{sec:constructing_the_parent_hamiltonian} was made to prove things that seem obvious (but are not) from \cite{chen2023efficient}.


\section{Transforming the Majorana operators}
\label{sec:transforming_the_majorana_operators}

\begin{lem}
\label{lem:fermion_basis_transform}
Let $H_0 = \sum_{jk} h_{jk} \gamma_j \gamma_k,$ where the matrix $(h_{jk})$ is Hermitian and pure imaginary, then
\begin{equation}
    e^{iH_0t}\gamma_l e^{-iH_0t} = \sum_{j=1}^{n} \left(e^{4iht}\right)_{jl}\gamma_j.
\end{equation}
\end{lem} 
\begin{proof}
We use the notation $\operatorname{ad}_A(B)=[A,B]$,
\begin{equation}
\begin{aligned}
    \operatorname{ad}_{H_0}(\gamma_l)&=[H_0,\gamma_l]
    =\sum_{jk}h_{jk}[\gamma_j\gamma_k,\gamma_l].
\end{aligned}
\end{equation}
Using the commutation relation, we can calculate the commutator
\begin{equation}
\begin{aligned}
    [\gamma_j\gamma_k,\gamma_l]&=\gamma_j\gamma_k\gamma_l-\gamma_l\gamma_j\gamma_k\\
    &=\gamma_j\gamma_k\gamma_l+\gamma_j\gamma_l\gamma_k-\gamma_j\gamma_l\gamma_k-\gamma_l\gamma_j\gamma_k\\
    &=\gamma_j\{\gamma_k,\gamma_l\}-\{\gamma_j,\gamma_l\}\gamma_k\\
    &=2(\delta_{kl}\gamma_j-\delta_{jl}\gamma_k).
\end{aligned}
\end{equation}
Therefore
\begin{equation}
\begin{aligned}
    \operatorname{ad}_{H_0}(\gamma_l)&=2\left(\sum_j h_{jl}\gamma_j -\sum_k h_{lk} \gamma_k\right)
    =2\left(\sum_j h_{jl}\gamma_j +\sum_k h_{kl} \gamma_k \right)
    =4\sum_{j}h_{jl}\gamma_j. \\
\end{aligned}
\end{equation}
As a result we have
\begin{equation}
    \operatorname{ad}_{H_0}^r(\gamma_l)=\sum_j\left(4h\right)^r_{jl}\gamma_j
\end{equation}
Using the Baker–Campbell–Hausdorff (BCH) formula 
we get
\begin{equation}
    \begin{aligned}
        e^{iH_0t}\gamma_l e^{-iH_0t}&=\sum_r \frac{1}{r!}\operatorname{ad}_{H_0}^r(\gamma_l)
        =\sum_{j=1}^{n} \left(e^{4iht}\right)_{jl}\gamma_j.
    \end{aligned}
\end{equation}
\end{proof}

\section{Linear combination of Majorana operators}

\begin{lem}
    \label{lem:norm_linear_comb_majorana}
    For $x=(x_j), y=(y_j)\in\RR^{2n}$, let
    \begin{equation}
        A = \sum_{j}x_j \gamma_j + i\sum_j y_j\gamma_j.
    \end{equation}
    Then 
    \begin{equation}
    \label{eq:eigenvalues_LC_majorana}
    \|A\|^2 = \|x\|_2^2 + \|y\|_2^2 + 2\sqrt{\|x\|_2^2\|y\|_2^2-(x^\top y)^2}.
\end{equation}
\end{lem}
\begin{proof}
First introduce vectors $a,b\in\RR^{2n}$ such that
$\|a\|_2=\|b\|_2=1$, $a^\top b=0$, and
\[
\begin{pmatrix}
    x & y
\end{pmatrix} = 
\begin{pmatrix}
    a & b
\end{pmatrix}
\underbrace{
\begin{pmatrix}
    M_{11} & M_{12} \\
    M_{21} & M_{22}
\end{pmatrix}}_{M}.
\]
Then we define $\omega_1$ and $\omega_2$ to be
\[
\omega_1 = \sum_j a_j \gamma_j,\quad \omega_2 = \sum_j b_j \gamma_j.
\]
We can directly verify that $\omega_1$ and $\omega_2$ are Majorana operators satisfying the canonical anti-commutation relation. We can rewrite $A$ as
\[
A = M_{11}\omega_1 + M_{21}\omega_2 + i(M_{12}\omega_1+M_{22}\omega_2).
\]
This allows us to compute
\[
A^\dag A = \|M\|_F^2 + i
\begin{pmatrix}
    \omega_1 & \omega_2
\end{pmatrix}
\begin{pmatrix}
    0 & \det(M) \\
    -\det(M) & 0
\end{pmatrix}
\begin{pmatrix}
    \omega_1 \\
    \omega_2
\end{pmatrix}.
\]
Therefore we can see that $A^\dag A$ has two eigenvalues $\|M\|_F^2 \pm 2\det(M)$. We can further compute that
\[
\|M\|_F^2 = \|x\|_2^2 + \|y\|_2^2,\quad \det(M)^2=\|x\|_2^2\|y\|_2^2-(x^\top y)^2.
\]
Therefore we have \eqref{eq:eigenvalues_LC_majorana}.
\end{proof}


\section{Norm preservation in third quantization}
\label{sec:norm_preserve_third_quantization}

In Section~\ref{sec:third_quantization} we have introduced the super-operators $L_A$ and $R_A$ for an operator $A$, and how these super-operators can be mapped back to operators in an enlarged Hilbert space using $\Phi$ defined in Definition~\ref{defn:Phi_thrid_quantization}. In this section we will show that the operator spectral norm is preserved in this process. More precisely, we have
\begin{lem}
    \label{lem:norm_preservation}
    For $A\in B(\mathcal{H})$ with fixed parity (even or odd as defined in Definition~\ref{defn:parity_of_operators}), we have
    \begin{equation}
    \label{eq:left_norm_preserve}
        \|\Phi(L_A)\| = \|A\|,
    \end{equation}
    and
    \begin{equation}
        \label{eq:right_norm_preserve}
        \|\Phi(R_A)\|\leq \|A\|.
    \end{equation}
\end{lem}

Before we present the proof we will first prove some useful lemmas:
\begin{lem}
    \label{lem:hermitian_preserving}
    For any $A\in B(\mathcal{H})$, $\varphi\circ R_{A^\dag}\circ\varphi^{-1}=(\varphi\circ R_{A}\circ\varphi^{-1})^\dag$.
\end{lem}
\begin{proof}
    We can write $A=\sum_{\alpha\in\{0,1\}^{2n}} A_{\alpha}\gamma_1^{\alpha_1}\gamma_2^{\alpha_2}\cdots \gamma_{2n}^{\alpha_{2n}}$.
    By \eqref{eq:repr_left_right_multiplication_single_majorana}, \eqref{eq:left_right_multiply_repr_rules}, and $\Phi(R_I)=I$, we have
    \[
    \begin{aligned}
        \varphi\circ R_{A}\circ\varphi^{-1} &= \sum_{\alpha} A_\alpha \varphi\circ R_{\gamma_{2n}}^{\alpha_{2n}}\cdots R_{\gamma_{2}}^{\alpha_{2}} R_{\gamma_{1}}^{\alpha_{1}}\circ\varphi^{-1} \\
        &= \sum_{\alpha} A_\alpha (-i)^{|\alpha|}\hat{\gamma}_{4n}^{\alpha_{2n}}(-1)^{\alpha_{2n}\hat{n}}\cdots \hat{\gamma}_{4}^{\alpha_{2}}(-1)^{\alpha_2\hat{n}} \hat{\gamma}_{2}^{\alpha_{1}}(-1)^{\alpha_1\hat{n}}.
    \end{aligned}
    \]
    Therefore
    \[
    \begin{aligned}
        (\varphi\circ R_{A}\circ\varphi^{-1})^\dag  &= \sum_{\alpha} A_\alpha^* i^{|\alpha|}  (-1)^{\alpha_1\hat{n}}\hat{\gamma}_{2}^{\alpha_{1}} (-1)^{\alpha_2\hat{n}}\hat{\gamma}_{4}^{\alpha_{2}}\cdots (-1)^{\alpha_{2n}\hat{n}}\hat{\gamma}_{4n}^{\alpha_{2n}} \\
        &=\sum_{\alpha} A_\alpha^* (-i)^{|\alpha|}  \hat{\gamma}_{2}^{\alpha_{1}} (-1)^{\alpha_{1}\hat{n}}\hat{\gamma}_{4}^{\alpha_{2}}(-1)^{\alpha_{2}\hat{n}}\cdots \hat{\gamma}_{4n}^{\alpha_{2n}}(-1)^{\alpha_{2n}\hat{n}},
    \end{aligned}
    \]
    where for the second equality we have used $(-1)^{\hat{n}}\hat{\gamma}_{2j} = - \hat{\gamma}_{2j}(-1)^{\hat{n}}$.

    On the other hand,
    \[
    \begin{aligned}
        \varphi\circ R_{A^\dag}\circ\varphi^{-1}  &= \sum_{\alpha} A_\alpha^* \varphi\circ(-i)^{|\alpha|}   R_{\gamma_{1}}^{\alpha_{1}} R_{\gamma_{2}}^{\alpha_{2}} \cdots R_{\gamma_{2n}}^{\alpha_{2n}} \circ \varphi^{-1}\\
        &=\sum_{\alpha} A_\alpha^* (-i)^{|\alpha|} \hat{\gamma}_{2}^{\alpha_{1}} (-1)^{\alpha_{1}\hat{n}}\hat{\gamma}_{4}^{\alpha_{2}}(-1)^{\alpha_{2}\hat{n}}\cdots \hat{\gamma}_{4n}^{\alpha_{2n}}(-1)^{\alpha_{2n}\hat{n}}.
    \end{aligned}
    \]
    We therefore have $\varphi\circ R_{A^\dag}\circ\varphi^{-1}=(\varphi\circ R_{A}\circ\varphi^{-1})^\dag$.
\end{proof}

\begin{cor}
    \label{cor:hermitian_preserving_or_not}
    For $A\in B(\mathcal{H})$, if $A$ has even parity, then $\Phi(R_{A^\dag})=\Phi(R_A)^\dag$. If $A$ has odd parity, then $\Phi(R_{A^\dag})=-\Phi(R_A)^\dag$.
\end{cor}
\begin{proof}
    The even case follows directly from the definition of $\Phi$ in Definition~\ref{defn:Phi_thrid_quantization} and Lemma~\ref{lem:hermitian_preserving}. For the odd case, we have
    \[
    \Phi(R_A)^\dag = (\varphi\circ R_A \circ\varphi^{-1}(-1)^{\hat{n}})^\dag = (-1)^{\hat{n}} \varphi\circ R_{A^\dag} \circ\varphi^{-1}=-\varphi\circ R_{A^\dag} \circ\varphi^{-1}(-1)^{\hat{n}}=-\Phi(R_{A^\dag}).
    \]
\end{proof}

\begin{lem}
    \label{lem:preserve_psd}
    If $A\in B(\mathcal{H})$ is positive semi-definite with even parity, then $\Phi(R_A)$ is also positive semi-definite.
\end{lem}
\begin{proof}
    Since $A$ is positive semi-definite, there exists $Q\in B(\mathcal{H})$ with even parity such that $A=QQ^\dag$. Therefore by Lemma~\ref{lem:hermitian_preserving}
    \[
    \Phi(R_A) = \varphi \circ R_{Q^\dag}R_{Q}\circ\varphi^{-1}=(\varphi \circ R_{Q^\dag}\circ\varphi^{-1})(\varphi \circ R_{Q}\circ\varphi^{-1})=(\varphi \circ R_{Q}\circ\varphi^{-1})^\dag (\varphi \circ R_{Q}\circ\varphi^{-1})
    \]
    and is therefore positive semi-definite.
\end{proof}

\begin{proof}[Proof of Lemma~\ref{lem:norm_preservation}]
    We can write 
    \[
    A=\sum_{\alpha\in\{0,1\}^{2n}} A_{\alpha}\gamma_1^{\alpha_1}\gamma_2^{\alpha_2}\cdots \gamma_{2n}^{\alpha_{2n}}.
    \]
    By \eqref{eq:repr_left_right_multiplication_more_majorana_ops}, we can explicitly write down
    \[
    \Phi(L_A) = \sum_{\alpha\in\{0,1\}^{2n}} A_{\alpha}\hat{\gamma}_1^{\alpha_1}\hat{\gamma}_3^{\alpha_2}\cdots \hat{\gamma}_{4n-1}^{\alpha_{2n}}.
    \]
    We can see that $\Phi(L_A)$ simply comes from replacing the Majorana operators in $A$ by a different set of Majorana operators, and therefore we have \eqref{eq:left_norm_preserve}.

    For $\Phi(R_A)$, first we note that $\|A\|^2 I - A^\dag A$ is positive semi-definite. Therefore by Lemma~\ref{lem:preserve_psd}, we know that
    \[
    \Phi(R_{\|A\|^2 I - A^\dag A}) = \|A\|^2 I - \Phi(R_{A^\dag A}) = \|A\|^2 I - \Phi(R_A)\Phi(R_A)^\dag
    \]
    is also positive semi-definite. We therefore have \eqref{eq:right_norm_preserve}. In the above, for the case where $A$ has odd parity, we have used
    \[
    \Phi(R_A)\Phi(R_A)^\dag = (\varphi\circ R_A\circ \varphi^{-1})(-1)^{\hat{n}}(-1)^{\hat{n}}(\varphi\circ R_A\circ \varphi^{-1})^\dag =\varphi\circ R_A R_{A^\dag}\circ\varphi^{-1}=\Phi(R_{A^\dag A}).
    \]
\end{proof}

\section{Decay of the kernel function in the coherent part}
\label{sec:decay_kernel_function_coherent_part}

We will first prove some elementary properties of $b_1(t)$ defined in \eqref{eq:f(t),b1(t),b2(t)}. We recall that
\begin{equation*}
    \begin{aligned}
        b_1(t) = 2\sqrt{\pi}e^{\frac{1}{8}}\left(\frac{1}{\cosh(2\pi t)}*_t \sin(-t)e^{-2t^2}\right).
    \end{aligned}
\end{equation*}
We will need to use the lemma below:

\begin{lem}
\label{lem:decay_convolution}
    If two real functions $f(t)$ and $g(t)$ satisfy $|f(t)|\leq C_1 e^{-\lambda_1|t|}$ and $|g(t)|\leq C_2 e^{-\lambda_2|t|}$ with positive constants $C_1,C_2,\lambda_1,\lambda_2$, then 
    \begin{equation}
        |(f*g)(t)|\leq 2C_1C_2\left(\frac{e^{-\lambda_2|t|/2}}{\lambda_1 }+\frac{e^{-\lambda_1|t|/2}}{\lambda_2}\right).
    \end{equation}
\end{lem}
\begin{proof}
The convolution can be bounded by:
    \begin{equation}
    \begin{aligned}
        |(f*g)(t)|&=\left|\int_{-\infty}^{+\infty}\mathrm{d}s f(s)g(t-s)\right|<\int_{-\infty}^{+\infty}\left| f(s)g(t-s)\right|\mathrm{d}s\\
        &=C_1C_2\int_{-\infty}^{+\infty} \mathrm{e}^{-\lambda_1|s|}\mathrm{e}^{-\lambda_2|t-s|}\mathrm{d}s \\
        &=C_1C_2\left(\int_{|s-t|\leq |t|/2}+\int_{|s-t|> |t|/2}\right) \mathrm{e}^{-\lambda_1|s|}\mathrm{e}^{-\lambda_2|t-s|}\mathrm{d}s 
    \end{aligned}
    \end{equation}

    For the integral for $|s-t|\leq |t|/2$, we have
    \[
    \int_{|s-t|\leq |t|/2} \mathrm{e}^{-\lambda_1|s|}\mathrm{e}^{-\lambda_2|t-s|}\mathrm{d}s \leq e^{-\lambda_1 |t|/2}\int \mathrm{e}^{-\lambda_2|t-s|} \dd s \leq \frac{2e^{-\lambda_1 |t|/2}}{\lambda_2}.
    \]
    For the integral for $|s-t|> |t|/2$, we have
    \[
    \int_{|s-t|>|t|/2} \mathrm{e}^{-\lambda_1|s|}\mathrm{e}^{-\lambda_2|t-s|}\mathrm{d}s \leq e^{-\lambda_2 |t|/2}\int \mathrm{e}^{-\lambda_1|s|} \dd s \leq \frac{2e^{-\lambda_2 |t|/2}}{\lambda_1}.
    \]
    Combining the two parts yields the desired inequality.
\end{proof}

\begin{lem}
    \label{lem:decay_b1}
    Let $b_1(t)$ be as defined in \eqref{eq:f(t),b1(t),b2(t)}, then there exist universal constants $C,\lambda>0$ such that
    \[
    b_1(t-i\delta)\leq C\cosh(\delta)e^{2\delta^2}e^{-\lambda t}.
    \]
\end{lem}

\begin{proof}
    Note that $b_1(t)$ is the convolution between two parts:
    \[
    h_1(t) = 2\sqrt{\pi}e^{\frac{1}{8}}\frac{1}{\cosh(2\pi t)},\quad h_2(t) \sin(-t)e^{-2t^2}.
    \]
    Therefore
    \[
    b_1(t-i\delta) = \int h_1(s) h_2(t-i\delta-s) \dd s.
    \]
    Let $h_{2,\delta}(t)=h_2(t-i\delta)$, then $b_1(t) = (h_1*h_{2,\delta})(t)$. Therefore we only need to prove the decay properties of $h_1(t)$ and $h_{2,\delta}(t)$ and then apply Lemma~\ref{lem:decay_convolution}.

    For $h_1(t)$, we can readily get universal constants $C_1,\lambda_1>0$ such that $|h_1(t)|\leq C_1 e^{-\lambda_1 |t|}$. Therefore we only need to focus on $h_{2,\delta}(t)$. Note that
    \[
    |h_{2,\delta}(t)| = |\sin(-t+i\delta)e^{-2(t-i\delta)^2}|\leq \cosh(\delta)e^{2\delta^2}e^{-2t^2}.
    \]
    Because $e^{-2t^2}$ decays faster than any exponential function, we also have universal constants $C_2,\lambda_2>0$ such that $|h_{2,\delta}(t)|\leq C_2\cosh(\delta)e^{2\delta^2} e^{-\lambda_2 |t|}$. The result then immediately follows from Lemma~\ref{lem:decay_convolution}.
\end{proof}

We are now ready to prove the decay property of $\check{F}_2(t,t')$ defined in \eqref{eq:checkF_2(t,t')_defn}, which we rewrite here
\[
\check{F}_2(t,t')=\frac{1}{2\pi}\int_{-\infty}^{\infty}\mathrm{d} t\int_{-\infty}^{\infty}\mathrm{d} t^{\prime}F_2(\nu,\nu') \mathrm{e}^{i \nu t}\mathrm{e}^{i \nu' t}.
\]
\begin{proof}[Proof of Lemma~\ref{lem:decay_of_checkF2(t,t')}]
We perform a change-of-variable as follows:
\[
\begin{pmatrix}
    \nu_1 \\ \nu_2
\end{pmatrix}
=\frac{1}{\sqrt{2}}
\begin{pmatrix}
    1 & -1 \\
    1 & 1
\end{pmatrix}
\begin{pmatrix}
    \nu \\ \nu'
\end{pmatrix},\quad 
\begin{pmatrix}
    t_1 \\ t_2
\end{pmatrix}
=\frac{1}{\sqrt{2}}
\begin{pmatrix}
    1 & -1 \\
    1 & 1
\end{pmatrix}
\begin{pmatrix}
    t \\ t'
\end{pmatrix}.
\]
This allows us to write
\[
\begin{aligned}
    \check{F}_2\left(\frac{t_1+t_2}{\sqrt{2}},\frac{t_2-t_1}{\sqrt{2}}\right) &= \frac{1}{2\pi}\int \dd\nu_1\int\dd\nu_2 F_2\left(\frac{\nu_1+\nu_2}{\sqrt{2}},\frac{\nu_2-\nu_1}{\sqrt{2}}\right)e^{i(\nu_1 t_1+\nu_2 t_2)} \\
    &=\int \dd\nu_1 e^{i\nu_1 (t_1-i\sqrt{2}/4)}\hat{b}_1(\sqrt{2}\nu_1)\int\dd\nu_2 e^{i\nu_2 t_2} \hat{b}_2(\sqrt{2}\nu_2) \\
    &=\frac{1}{2}b_1(t_1/\sqrt{2}-i/4)b_2(t_2/\sqrt{2})
\end{aligned}
\]
Since the exponential decay of $b_1(t_1/\sqrt{2}-i/4)$ from Lemma~\ref{lem:decay_b1} and the exponential decay of $b_2(t_2/\sqrt{2})$ from the definition of $b_2(t)$ in \eqref{eq:f(t),b1(t),b2(t)}, we have the exponential decay of $\check{F}_2\left(\frac{t_1+t_2}{\sqrt{2}},\frac{t_2-t_1}{\sqrt{2}}\right)$ in $|t_1|+|t_2|$. Due to norm equivalence we then have the exponential decay of $\check{F}_2(t,t')$ in $(t^2+t'^2)^{1/2}$. We therefore have
\[
|\check{F}_2(t,t')|\leq Ce^{-\lambda (t^2+t'^2)^{1/2}},
\]
for some universal constants $C,\lambda>0$.
\end{proof}

\end{document}